\documentclass[sigconf,nonacm]{aamas} %For anonymized submission
\renewcommand\footnotetextcopyrightpermission[1]{} % just in case a footnote sneaks in
\usepackage{stmaryrd}
\usepackage{tikz}
\usepackage{comment}
\usetikzlibrary{positioning}
\usetikzlibrary{arrows.meta}
\usepackage{amsmath}
\usepackage{algpseudocode}
\usepackage{thmtools}
\usepackage{thm-restate}
\usepackage{cleveref}
\usepackage{xcolor} 
\usepackage{leftindex} 
\usepackage{algorithm}
\usepackage{amsmath}
\usepackage{newfloat}
\usepackage{listings}
\usepackage{balance}
\usepackage{enumitem}
\usepackage[normalem]{ulem} % in preamble
\usepackage{soul} % in preamble
\usepackage{float}
\usepackage{array}
\usepackage{nicematrix}
\usepackage{subcaption} 
\usepackage{dsfont}

\usepackage{framed}

\title{Synthesis of Safety Specifications for Probabilistic Systems}

\author{Gaspard Ohlmann$^*$}
\affiliation{
    \city{Mulhouse}
    \country{France}
    }
\email{gaspard.ohlmann@outlook.com}
\author{Edwin Hamel-De le Court$^*$}
\affiliation{
    \institution{Imperial College}
    \city{London}
    \country{United Kingdom}
}
\email{e.hamel-de-le-court@imperial.ac.uk}
\author{Francesco Belardinelli}
\affiliation{
    \institution{Imperial College}
    \city{London}
    \country{United Kingdom}
}
\email{francesco.belardinelli@imperial.ac.uk}

\begin{abstract}
Ensuring that agents satisfy safety specifications can be crucial in safety-critical environments. While methods exist for controller synthesis with safe temporal specifications, most existing methods restrict safe temporal specifications to probabilistic-avoidance constraints. Formal methods typically offer more expressive ways to express safety in probabilistic systems, such as Probabilistic Computation Tree Logic (PCTL) formulas. Thus, in this paper, we develop a new approach that supports more general temporal properties expressed in PCTL. Our contribution is twofold. First, we develop a theoretical framework for the Synthesis of safe-PCTL specifications. We show how the reducing global specification satisfaction to local constraints, and define CPCTL, a fragment of safe-PCTL. We demonstrate how the expressiveness of CPCTL makes it a relevant fragment for the Synthesis Problem. Second, we leverage these results and propose a new Value Iteration-based algorithm to solve the synthesis problem for these more general temporal properties, and we prove the soundness and completeness of our method.

\end{abstract}
\newtheorem{definition}{Definition}

\newcommand{\alit}[1]{\mathtt{L}\left(#1\right)}

\newcommand{\last}[1]{\text{last}\left(#1\right)}
\newcommand{\paths}[1]{\text{paths}\left(#1\right)}

\newcommand{\dclosure}[1]{\uparrow\left(#1\right)}

\newcommand{\BibTeX}{\rm B\kern-.05em{\sc i\kern-.025em b}\kern-.08em\TeX}
\newcommand{\CPCTL}{\texttt{CPCTL}}

\newcommand{\MCpathSatisfiesIndex}[1]{(\mathcal{M}, #1)\models\,}
\newcommand{\MCstateSatisfies}[1]{(\mathcal{M}, {#1})\models\,}

\newcommand{\WW}{\,\mathbf{W}\,}
\newcommand{\GG}{\mathbf{G}\,}

\newcommand{\safePCTL}{$\texttt{PCTL}_{safe}$}
\newcommand{\cosafe}{$\texttt{PCTL}_{cosafe}$}
\newcommand{\LPCTL}{\texttt{L-PCTL}}
\newcommand{\livePCTL}{$\texttt{PCTL}_{live}$}
\newcommand{\MultiObjectives}{\texttt{MOA}}

\newcommand{\counterValuationSet}[1]{\{ 0,1\}^{sf(#1)}\times [0,1]^{pf(#1)}}

\newcommand{\SubformulaAndConvention}[2]{denote the set of state subformulas $\mathcal{SF}(#1)=\{#1_1,\dots,#1_{sf(#1)}\}$ and path formulas $\mathcal{PF}(#1)=\{#2_1,\dots,#2_{pf(#1)}\}$, with the added convention that for all $j\leq pf(#1)$, $#1_j=\mathbb P_{\geq p_j} (#2_j)$}
\newcommand{\aug}[1]{\hat{#1} }

\newtheorem{lemma}{Lemma}
\newtheorem{corollary}{Corollary}
\newtheorem{example}{Example}
\newtheorem{theorem}{Theorem}

\begin{document}

\newcommand{\PCTL}{\texttt{PCTL}}

\newcommand{\MDP}{\mathcal{M}} %For the non-augmented MDP - The MDP
\newcommand{\St}{\mathcal{S}} %For the non-augmented MDP - Set of states
\newcommand{\Act}{\mathcal{A}(s)} %For the non-augmented MDP - Set A(s) of actions

\newcommand{\AugMDP}{\aug{\mathcal{M}}} %For the augmented by unsafe MDP
\newcommand{\AugSt}{\aug{\mathcal{S}}} %For the augmented by unsafe MDP
\newcommand{\AugAct}{\aug{\mathcal{A}}(\aug{s})} %%For the augmented by unsafe MDP

\newcommand{\gac}[1]{\color{blue}{Gac: #1}}
\newcommand{\del}[1]{\color{orange}{\textbf{Will be deleted, do not remove yet}: #1} \color{black}}

\maketitle{}

$*$ \emph{These authors contributed equally to this work.}

%\tableofcontents

\section{Introduction}

Synthesizing policies that provably satisfy rich temporal specifications, while optimizing reward, is a long-standing challenge at the interface of Formal Methods and Reinforcement Learning. In probabilistic settings, \emph{Probabilistic Computation-Tree Temporal Logic} (PCTL) \cite{Ciesinski2004,baier2008principles} provides a natural language to express safety and performance requirements over Markov decision processes (MDPs). However, general PCTL synthesis is computationally hard and, under the assumption of history-dependent strategies, undecidable \cite{brazdil2006stochastic,DBLP:conf/icalp/BrazdilFK08}. Moreover, randomness and memory are necessary for \PCTL$\,$ synthesis,
%and can be necessary 
even for restricted fragments \cite{DBLP:conf/ifipTCS/BolligC04}. Therefore, further investigations 
%into \PCTL\ synthesis 
are required to construct strategies that satisfy such specifications.

This paper develops a new   framework for \emph{Safe PCTL} (\safePCTL) that we apply to the synthesis of specifications in {\em Continuing \PCTL} (\CPCTL), a fragment of \safePCTL\ that we here introduce. Our starting point is a structural insight: the satisfaction of a broad class of PCTL safety properties can be enforced by local, per-transition inequalities in a suitably \emph{augmented MDP}. We introduce two local conditions -- \emph{state compatibility} and \emph{path compatibility} -- that together imply global satisfaction of the original specification through a coherence theorem. 
Building on these foundations, we identify a syntactic fragment -- 
CPCTL -- for which these local constraints yield constructive algorithms, while still allowing nesting of probabilistic operators and thus providing a new computable class of safety specifications.

On the algorithmic side, we propose \textsc{CPCTL-VI}, a value-iteration type \cite{10.1007/978-3-540-69850-0_7} algorithm that monotonically tightens inductive lower bounds on satisfaction probabilities. \textsc{CPCTL-VI} represents a constructive method for the synthesis problem of \CPCTL, representing complex nested probabilistic and temporal behaviors.
\paragraph{Contributions.} The key contributions of this paper can be summarized as follows:
\begin{enumerate}%
\item \textbf{An analysis of temporal and probabilistic specifications:} We introduce Continuing \PCTL\, (\CPCTL), a fragment of safe \PCTL$\,$ generalizing multi-objective avoidance specifications, allowing nesting of probabilistic operators. We show that there currently exists no decidability result for such specifications. Additionally, we establish structural results for \CPCTL$\,$ such as weak reduction to literal projections.
  \item \textbf{A key theoretical result:} A new \emph{augmented MDP} construction for Safe PCTL that encodes global satisfaction as local linear inequalities. We define two conditions, the \emph{state compatibility} and \emph{path compatibility}, and show that their satisfaction guarantees the satisfaction of the corresponding formula. 
  \item \textbf{An Algorithm for the \CPCTL\ Synthesis Problem:} \textsc{CPCTL-VI} is a value-iteration algorithm that computes lower bounds on satisfaction probabilities and certifies realizability. Moreover, the algorithm is optimal under a generalized version of Slater's assumption.
\end{enumerate}

The paper is organized as follows. 
In Sec.~\ref{sec:related_work} we review the results on the synthesis problem for \PCTL \ and some of its significant fragments. 
In Sec.~\ref{sec:preliminaries}
%\emph{Preliminaries}, 
we provide background on MDPs, RL, and \PCTL. 
In Sec.~\ref{sec:CPCTL}
we define the Continuing \PCTL fragment and analyse its structural properties, including expressivity.
In Sec.~\ref{sec:augmented}
we introduce the augmented MDP construction, as well as the state/path-compatibility conditions, and prove the coherence theorem.  In Sec.~\ref{sec:VI} we present our value iteration  algorithm for \CPCTL \ synthesis -- including soundness and optimality results -- which is then evaluated empirically in Sec.~\ref{sec:experiments}.
Finally, we conclude in Sec.~\ref{sec:conclusions} pointing to future work.

\section{Related Work} \label{sec:related_work}

We present the inclusion relations between \PCTL \ and its significant fragments -- including \CPCTL -- in
Figure \ref{fig:exphier}.
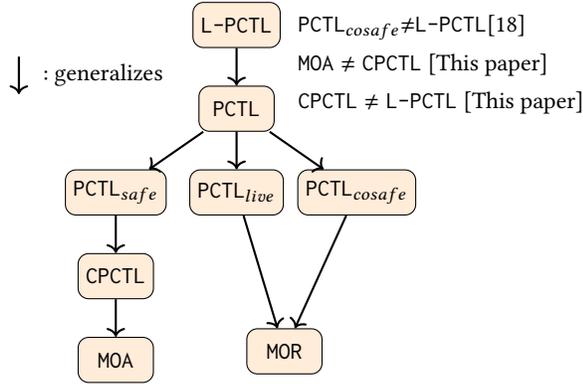
\begin{figure}[H]
\begin{tikzpicture}[
  node distance=0.5cm and 0.3cm,
  every node/.style={draw, rounded corners, align=center, minimum width=1cm, minimum height=0.6cm, fill=orange!15},
  arrow/.style={->, thick}
]

\node (lpctl) {\LPCTL};
% Top level
\node (pctl) [below =of lpctl] {\PCTL};

% Second level
\node (live) [below =of pctl] {\livePCTL};
\node (safe) [left=of live] {\safePCTL};
\node (cosafe) [below right=of pctl] {\cosafe};
\node (cpctl) [below=of safe] {\CPCTL};
\node (multi) [below=of cpctl] {\MultiObjectives};

\node (MOR) [below right=1.5cm and -0.5cm of live] {\texttt{MOR}};

% Arrows
\draw[arrow] (lpctl) -- (pctl);
\draw[arrow] (pctl) -- (safe);
\draw[arrow] (pctl) -- (live);
\draw[arrow] (pctl) -- (cosafe);
\draw[arrow] (safe) -- (cpctl);
\draw[arrow] (cpctl) -- (multi);
\draw[arrow] (live) -- (MOR);
\draw[arrow] (cosafe) -- (MOR);

\node[draw=none,fill=none] (ineg) [below right =-1.7cm and 0.2cm of pctl] {\cosafe$\neq $\LPCTL  \cite{DBLP:conf/ijcai/SongFZ15} };
\node[draw=none,fill=none] (ineg3) [below right=-1.2cm and 0.2cm of pctl] {$\MultiObjectives \neq \CPCTL$  [This paper]};
\node[draw=none,fill=none] (ineg4) [below right=-0.7cm and 0.2cm of pctl] {$\CPCTL \neq \LPCTL$  [This paper]};

\node[draw=none,fill=none] (test1) [below left= -0.5cm and 1.8cm of lpctl] {};
\node[draw=none,fill=none] (test2) [below=of test1] {};
\node[draw=none,fill=none] (test3) [above right= -0.05cm and  -0.3cm of test2] {: generalizes};

\draw[arrow] (test1) -- (test2);

\end{tikzpicture}
\caption{Expressivity of \CPCTL$\,$ in the PCTL hierarchy}\label{fig:exphier}
\end{figure}

\paragraph{Undecidability of \PCTL\ Synthesis for HR policies.}
For history-dependent (perfect-recall) stochastic controllers (\textbf{HR}), controller synthesis is undecidable for logics slightly more expressive than \PCTL, such as $L$-\PCTL, which allows linear combinations of probabilities~\cite{bordais2022strategy}, or $\texttt{PCTL}_{G\&L}$, which allows finite window operators such as $\mathbf{U}^{\leq n}$ or $\mathbf{W}^{\leq n}$ \cite{brazdil2006stochastic,DBLP:conf/icalp/BrazdilFK08}. While the proof of undecidability, such as reduction from Minsky machines, cannot be directly applied to \PCTL, it suggests that the \PCTL\ synthesis problem is either undecidable or highly complex. Moreover, memory and randomization are necessary for \PCTL\ synthesis, even when no nesting is allowed \cite{DBLP:conf/ifipTCS/BolligC04}. For the co-safe fragment of \PCTL -- \cosafe -- decidability can be recovered, although with very high complexity (3EXPTIME \cite{DBLP:conf/ijcai/SongFZ15}). 

\paragraph{Complexity for MR and MD policies.}
For memoryless randomized strategies (\textbf{MR}), 
\cite{kuvcera2005controller} provides a construction that yields a $\mathbf{PSPACE}$ upper bound for the general \PCTL\ synthesis problem. 
In the easier case of memoryless deterministic strategies (\textbf{MD}), the synthesis problem for both \PCTL\ and $L$-\PCTL\ objectives is already $\mathbf{NP}$-complete~\cite{brazdil2006stochastic}. Furthermore, 
the PCTL synthesis problem for \textbf{MD} policies is NP-complete \cite{DBLP:conf/ifipTCS/BolligC04} and is in \textsc{EXPTIME} for \textbf{MR} policies \cite{DBLP:journals/fuin/KuceraS08}.

\paragraph{Multi-Objective Reachability and Avoidance.}
Multi-objective reachability (\texttt{MOR}) and avoidance (\texttt{MOA}) queries have been shown to be solvable through Linear Programming \cite{DBLP:journals/lmcs/EtessamiKVY08}, and thus can be solved in polynomial time. A value iteration scheme for multiple objectives based on computing Pareto frontiers appears in \cite{WHITE82}, and a practical algorithm that generates successive approximations of the Pareto frontier for multi-objective reachability and avoidance appears in \cite{DBLP:conf/atva/ForejtKP12}. These results are summarized in Table \ref{fig:comp}.

\paragraph{Summary.} 
We summarize these results in Table~\ref{fig:complmerged}.
The multi-objective avoidance (resp. reach) fragment \MultiObjectives \ (resp.~\texttt{MOR}) is composed of all the formulas of the form $\bigwedge_{j=1}^n \mathbb{P}_{\geq p_j} (\mathbf{G}\neg a_j)$ (respectively $\bigwedge_{j=1}^n \mathbb{P}_{\geq p_j} (\mathbf{F} a_j)$) for $a_j\in AP$, where $\mathbf{G}$ is the globally operator and $\mathbf{F}$ is the reach operator. These two fragments have been studied individually, as they are of theoretical and practical relevance. The definition of the other fragments along with a study of their expressivity and/or complexity can be found in \cite{bordais2022strategy} for \LPCTL, \cite{bertrand2012bounded,brazdil2006stochastic} for $\PCTL_{G\&L}$, \cite{katoen2014probably,DBLP:conf/ijcai/SongFZ15} for \cosafe and \safePCTL , and \cite{DBLP:journals/lmcs/EtessamiKVY08} for \MultiObjectives.

\begin{table}[H]
\begin{center}
\begin{NiceTabular}{c c c}[hvlines]
\Hline
{\bf Class} & {\bf Strategies} & {\bf Complexity}\\
\Block{2-1}{\texttt{L-PCTL}} & \textbf{HR} & $\Sigma_{1}^1$ hard \cite{bordais2022strategy} \\
 & \textbf{MD} & $\mathbf{NP}$-complete \cite{DBLP:conf/ifipTCS/BolligC04} \\
\Hline
\Block{2-1}{$\texttt{PCTL}$}
 & \textbf{HR} & Unknown \\
 & \textbf{MD} & $\mathbf{NP}$-complete \cite{DBLP:conf/ifipTCS/BolligC04} \\
\Hline
\Block{1-1}{$\texttt{PCTL}_{G\&L}$}
 & \textbf{HR} & $\Sigma_{1}^1$ hard \cite{brazdil2006stochastic} \\
\Hline
\Block{1-1}{\safePCTL}
 & \textbf{HR} & Unknown \\
\Hline
\Block{1-1}{\cosafe}
 & \textbf{HR} & $\Delta_1^0$ \cite{DBLP:conf/ijcai/SongFZ15} \\
\Hline
\Block{1-1}{\MultiObjectives}
 & \textbf{HR} & \textbf{PTIME} \cite{DBLP:journals/lmcs/EtessamiKVY08} \\
\Hline
\Block{1-1}{\texttt{MOR}}
 & \textbf{HR} & \textbf{PTIME} \cite{DBLP:journals/lmcs/EtessamiKVY08} \\
\end{NiceTabular}
\end{center}
\vspace{0.2cm}
\caption{Complexity of \PCTL$\,$ Synthesis for \PCTL\ fragments under different strategy assumptions. }
\label{fig:complmerged}
\end{table}

\paragraph{Continuing \PCTL.}
We identify the fragment $\CPCTL$ of \safePCTL\ that strictly extends multi-objective avoidance specifications, and we provide a synthesis algorithm for this logic
that solves the decidability problem and outputs a safe policy whenever one exists under a generalized Slater's Assumption. This assumption intuitively means that the formula can be satisfied when $\geq$ are replaced by strict $>$. Precise definitions are provided in Sec.~\ref{sec:preliminaries}. In addition, \CPCTL \ contains formulas that do not belong to any of the known computable classes. Common classes for which an algorithm has been described are summarized in Table~\ref{fig:comp}. Although the decidability of the synthesis problem for the whole \safePCTL remains open, our contribution is a step towards
understanding the subtle boundaries of decidability in \PCTL\ synthesis. As no existing algorithm solves the synthesis problem for $\CPCTL$ or any larger class of properties to the best of our knowledge, performance comparisons are not available.

\begin{table}[H]
\begin{tabular}{|c|c|}
     \hline
     {\bf Class} & {\bf Safe \& Complete Algorithm} \\ \hline
     Avoid & Yes \cite{DBLP:journals/tcs/HaddadM18,DBLP:journals/lmcs/EtessamiKVY08} \\ \hline 
     \MultiObjectives & Yes \cite{DBLP:journals/lmcs/EtessamiKVY08}  \\ \hline
     \CPCTL & Yes [This paper] \\ \hline
     \safePCTL & No \\ \hline 
     \cosafe & Yes \cite{DBLP:conf/ijcai/SongFZ15} \\ \hline
     \texttt{MOR} & Yes \cite{DBLP:journals/lmcs/EtessamiKVY08} \\ \hline
\end{tabular}
\vspace{0.1cm}
\caption{Currently computable classes, \textbf{HR} strategies.}\label{fig:comp}
\end{table}
\vspace{-1cm}

\section{Preliminaries} \label{sec:preliminaries}

In this section we provide preliminaries on Markov decision chains and processes, and Probabilistic CTL. We refer the interested reader to \cite{baier2008principles} for an in-depth presentation.

\subsection{Stochastic Models}

\paragraph{Markov Decision Processes}

A (labeled) \textit{ Markov Decision Process} \cite{10.5555/528623} is a tuple $\mathcal
M=\langle S,A,P,s_{0}, \allowbreak AP, L, R\rangle$, where $S$ is a
set of \textit{states}; $A$ is a mapping that associates every state
$s\in S$ to a nonempty finite set $A(s)$ of \emph{actions}; $P: S \times A \to \mathcal{D}(S)$ is a
\textit{transition probability function} that maps every state-action
pair $(s,a)$ to a probability
measure $P(s, a)$ over $S$; $s_{0}\in S$ is the \textit{initial state}\footnote{This can be assumed wlog compared to a model with an \emph{initial probability distribution} since it is always possible to add a new initial state to such a model with an action from this initial state whose associated probability distribution is the aforementioned initial probability distribution.};
$AP$ is a set of \emph{atomic propositions} (or atoms); $L:S\mapsto2^{AP}$ is a \emph{labeling function}; and $R:S\mapsto \mathbb R$ is the \emph{reward function}. 

In the following, all MDPs are labeled so we will simply call them Markov Decision Processes (MDP). For the sake of simplicity, we will write
$P(s'|s, a)$ instead of $P(s,a)(s')$. An MDP is finite if the sets of
states and actions are finite. 

\paragraph{Paths.}
A finite (resp.~infinite) {\em path} (or history) in an MDP $\mathcal M$ is a finite
(resp.~infinite) word $\rho=s_0 \cdots s_{n-1}s_n$
(resp.~$\cdots$) such that  
for any 
$i\leq n-1$ (resp. for any 
$i \in \mathbb{N}$), $s_i$ is a state of $\mathcal M$ and there exists an action $a_{i}\in A(s_i)$ such that $s_{i+1}$ is in the support of $P(s_{i},a_{i})$. In addition, 
$\paths{\mathcal M}$ denotes the set
of {\em infinite} paths of $\mathcal M$.
We denote as $\mathbb{P}_{\mathcal{M},\pi}$ the usual probability measure on infinite paths $s_0s_1\cdots\in S^\omega$ induced by a policy $\pi$ on a $\mathcal{M}$ (c.f.~\cite{baier2008principles} for full details). Finally, for a path  $\rho=s_0s_1\cdots$ and integer $i\ge 0$, we let $\rho_i=\rho[i]=s_i$.

\paragraph{Policies.}

For any measurable space $E$ \cite{halmos2013measure}, we denote by $\mathcal D(E)$ the set of probability measures over $E$. A {\em policy} (also called \emph{controller} or \emph{strategy}) $\pi$ of $\mathcal M$ is a
mapping that associates any finite path $\rho$
of $\mathcal M$ to a probability measure over  $A(\last{\rho})$.
It is memoryless if $\pi(\rho)$ only depends on $\last{\rho}$, in which case we denote $\pi(\rho)=\pi(s)$, where $s=\last{\rho}$. It is deterministic if for any finite path $\rho$,
$\pi(\rho)$ is a Dirac measure. We let 
\textbf{HR} denote the \emph{history-dependent and randomized} policies.

\paragraph{Markov Chain Induced by a Policy on a MDP}
For any policy
$\pi$ of $\mathcal M$, 
we let $\mathcal
M_\pi = (S_{\pi},P_{\pi},s_{0}, AP, L )$ denote the (labelled) {\em Markov chain} induced by $\pi$ in $\mathcal M$ such that
$S_{\pi}$ is the set of finite histories $\rho$ of $\mathcal M$, $AP$ is the same as in $\mathcal{M}$, $L_{\pi}(\rho) = L(\last{\rho})$, 
and $$P_{\pi}(s'|\rho) = \sum_{a \in A(s)} \pi(a|\rho)P(s'| \last{\rho},a).$$  When the policy is memoryless, the states of $\MDP_{\pi}$ correspond to the states of $\MDP$ and will be denoted the same. When the policy is history-dependent, the states of $\MDP_{\pi}$ correspond to all possible histories.  
For more details on MDPs and induced Markov chains, see \cite{baier2008principles,BSSstochastic}.

\subsection{Probabilistic and Temporal Specifications}

In this section, we define the probabilistic temporal Logic \PCTL \cite{Ciesinski2004} and its safe fragment \safePCTL \cite{katoen2014probably}. 
These two logics distinguish between two kind of formulas: path and the state formulas, and allow for probabilistic nesting of temporal specifications.

\subsubsection{Probabilistic Temporal Logic (\PCTL)} 

Let $AP$ be a finite set of atomic propositions and $q\in[0,1]$.

A formula of \PCTL$\,$ is generated by the nonterminal $\Phi$ in the following grammar:

\begin{eqnarray*}
\Phi &::= & a \mid \Phi\wedge \Phi 
\mid \Phi\vee\Phi \mid \neg \Phi
\mid \mathbb{P}_{\geq q}(\varphi),\quad a \in AP\cup \{\bot,\top\},\\
\varphi &::= &
\mathbf{X}\,\Phi_1  \mid 
\Phi_1\,\mathbf{W}\,\Phi_2 \mid \Phi_1\, \mathbf{U}\, \Phi_2.
\end{eqnarray*}

We call the formulas generated by $\Phi$ the \emph{state formulas}, and the formulas generated by $\varphi$ the \emph{path formulas}. The satisfaction relation $\models$ is defined on both sets by induction as follows.

\begin{definition}[Semantics] Given a Markov chain $\mathcal{M}$, state $s$, and path $\xi$, the satisfaction relation $\models$ is defined inductively as follows: 
\begin{tabbing}
\textbf{State formulas}: \ \ \ \ \ \ \ \ \ \ \ \ \ \= \\
$\MCstateSatisfies{s} a$ 
\ \  \ \ \ \ 
\ \ \ \ \  \ \ \ \ \ \ 
\= iff \ \ \= $a\in L(s)$\\
$\MCstateSatisfies{s} \neg \Phi_1$ \> iff \> $ (\mathcal{M},s)\not\models  \Phi_1$\\
$\MCstateSatisfies{s} \Phi_1\wedge \Phi_2$ \> iff \> $\MCstateSatisfies{s}  \Phi_1$ and $\MCstateSatisfies{s} \Phi_2$\\
$\MCstateSatisfies{s} \Phi_1\vee \Phi_2$ \> iff \> $\MCstateSatisfies{s} \Phi_1$ or $\MCstateSatisfies{s} \Phi_2$\\ 
$\MCstateSatisfies{s} \mathbb{P}_{\geq q}(\varphi)$ \> iff \>  $\mathbb{P}_{\mathcal{M}^s}\big(\{\xi\in\text{Paths}(s),\, \MCpathSatisfiesIndex{\xi} \varphi\}\big)\geq q$.\\
\textbf{Path formulas}:\\
$\MCpathSatisfiesIndex{\xi} \mathbf{X} \Phi_1$ \> iff \> $\mathcal{M}^{\xi[1]}\models \, \Phi_1 $ \\
$\MCpathSatisfiesIndex{\xi} \Phi_1 \, \mathbf{U} \, \Phi_2$ \> iff \> $\exists \,j\in \mathbb{N}$, $(\mathcal{M},{\xi[j]})\models \, \Phi_2 $ \\
\> \> and $\forall\, i \leq j$, $\MCstateSatisfies{\xi[i]} \Phi_1, $\\
$\MCpathSatisfiesIndex{\xi} \Phi_1\,\mathbf{W} \, \Phi_2 $ \> iff \> \= $\MCpathSatisfiesIndex{\xi} \left( \Phi_1\,\mathbf{U}\, \Phi_2\right) \vee (\mathbf{G}\, \Phi_1)$,\\
where $\MCpathSatisfiesIndex{\xi} \mathbf{G}\,\Phi_1$ \>  iff \> $\forall j\in \mathbb{N}$, $\xi[j]\models \Phi_1$.
\end{tabbing}
\end{definition}

\paragraph{Satisfaction, Markov Chain}
$(\MDP,s)\models \Phi$ is alternatively denoted $\MDP^s \models \Phi$. When clear from the context, we write for state formulas (resp. path formulas) $s\models \Phi::=\MCstateSatisfies{s} \Phi$ (resp. $\xi \models \phi::=\MCpathSatisfiesIndex{\xi} \phi$). Moreover, for a Markov chain $\mathcal{M}$, we define its \emph{semantics} of a state formula $\Phi$ as  $\llbracket \Phi \rrbracket_{\mathcal{M}} = \{ s ,\, (\mathcal{M}, s) \models \Phi \}$ and its \emph{semantics along a path $\xi$} as $\llbracket \Phi \rrbracket_{\xi,\mathcal{M}} = \{ i\in \mathbb{N},\,\MCstateSatisfies{\xi[i]} \Phi \}$. For readability, we may write $\llbracket \Phi \rrbracket::=\llbracket \Phi \rrbracket_{\mathcal{M}}$ and $\llbracket \Phi \rrbracket_\xi::=\llbracket \Phi \rrbracket_{\xi,\mathcal{M}}$.

\paragraph{Satisfaction, Probabilities}
For a Markov Chain $\MDP$ and a path formula $\phi$, we write $\mathbb{P}(\phi|\MDP)$ (resp. $\mathbb{P}(\phi|s,\mathcal{M})$) to denote the probability of the set of paths starting from the initial state of $\MDP$ (resp. $s$) and satisfying $\phi$ in $\MDP$. When clear from the context, we denote $\mathbb{P}(\phi|\pi,s)::=\mathbb{P}(\phi|s,\MDP_\pi)$.

\paragraph{Satisfaction, MDP}
%For a Markov Decision Process $\mathcal{M}$ with initial state $s_0$ and a (state) \PCTL formula $\Phi$, we write $\mathcal{M}\models \Phi$ when there exists $\pi$ such that $\mathcal{M}_\pi^s \models \Phi$. 
When the Markov chain is of the form $\mathcal{M}_\pi$ and is induced by policy $\pi$ on the MDP $\mathcal{M}$, the states of $\mathcal{M}_\pi$ are histories $\rho\in \mathcal{H}(\mathcal{S})$. Instead of $(\mathcal{M}_\pi, \xi) \models \phi$, where $\xi$ is a path of histories $\rho_i\in \mathcal{H}(\mathcal{M})$, for $i\in \mathbb{N}$, we may write $\xi\models_\pi \phi$ or simply $\xi \models \phi$. Similarly, instead of $\mathcal{M}_\pi^{\rho}\models \Phi$, we may write $\rho \models \Phi$. When the policy $\pi$ is memoryless, we write $s \models \Phi$ for $\mathcal{M}_\pi^{\rho}\models \Phi$, when $s = \last{\rho}$.

\subsubsection{Safe PCTL.}
The safe fragment of \PCTL, denoted \safePCTL, is the subset of \PCTL formulas, whose syntax is given by 
\begin{eqnarray*}
\Phi &::= & a \mid \neg a\mid \Phi\wedge \Phi 
\mid \Phi\vee\Phi 
\mid \mathbb{P}_{\geq q}(\varphi),\quad a \in AP\cup \{\bot,\top\},\\
\varphi &::= &
\mathbf{X}\,\Phi_1  \mid 
\Phi_1\,\mathbf{W}\,\Phi_2.
\end{eqnarray*}

The semantics of the relevant operators remains the same. As opposed to \PCTL, the negation is only allowed on atomic propositions, and the until operator $\mathbf{U}$ is not allowed. Note that the operator $\GG$ is in the safety fragment as $\GG \Phi = \Phi \WW \bot$.

The safety fragment of \PCTL$\,$ extends the multi-objective avoidance case, and corresponds to formulas whose
%such that, on syntactic trees, their 
violation can be witnessed in finite time, but not their satisfaction. For example, the \safePCTL\, formula $\mathbb{P}_{\geq{1/2}}(\mathbf{G}\neg a)$ -- stating that at least $50\%$ of the paths should always avoid $a$ -- satisfies this property, as any path can always \emph{a priori} reach $a$ later. For more in-depth intuition and analysis of the safety fragment, we refer the reader to \cite{katoen2014probably}.

%For example, knowing that a policy always avoids $a$ for a finite number of steps does not imply that the policy satisfies $\GG \neg a$.

When a policy satisfies a safety formula, we will say that the policy is \emph{safe} (for this formula). Finally, the term \emph{safety specification} is synonymous of \emph{safety formula}.

\section{Continuing PCTL} \label{sec:CPCTL}

%\fb{define the fragment earlier.}

\subsection{Definition and Problem Statement}

%\paragraph{Syntax.}
%
 Continuing PCTL (CPCTL) is the fragment of \PCTL\, built on \emph{state} formulas $\Phi$ and \emph{path} formulas $\varphi$, according to the following BNF:
%\[
\begin{eqnarray*}
\Phi &::= & a \mid \neg a \mid \Phi\wedge \Phi
%\mid \Phi\vee\Phi 
\mid \mathbb{P}_{\geq q}(\varphi),\quad a\in AP,\\
\varphi &::= &
%\mathbf{X}\,\Phi_1  \mid 
\Phi_1\,\mathbf{W}\,(\Phi_1 \land \Phi_2).
\end{eqnarray*}
%\]

Continuing \PCTL$\,$, as opposed to \safePCTL, does not allow disjunctions or the next operator $\mathbf{X}$. Moreover, the syntax of the weak until $\mathbf{W}$ is restricted such that condition $\Phi_1$ appears as a conjunct in goal $\Phi_1 \land \Phi_2$.
%$ is replaced by the continuing weak until $\mathbf{W_c}$ where $\mathbf{W}=\Phi_1 \mathbf{W}_c (\Phi_1 \wedge \Phi_2)$.

Hereafter let $\mathcal M$ be an MDP and $\Phi$ a \CPCTL$\,$state formula, and \SubformulaAndConvention{\Phi}{\phi}.

\paragraph{The Synthesis Problem.} In this contribution we design an algorithm for the following problem: Let $\MDP$ be an MDP and $\Phi\in \CPCTL$. 
%a fragment of \safePCTL\ that we introduce in the next section, 
Then, find a \textbf{HR} policy $\pi$ such that $\MDP_\pi \models \Phi$.

%\fb{How is the following different from model checking the flat fragment of PCTL$^*$?}
%Also, let's move this where the problem has actually been %defined.

%\fb{possibly move the following to related work.}
We emphasize first that we adopt the standard semantics of \PCTL\ operators for the synthesis problem:
%as defined in the previous section: 
a single policy is used throughout the formula evaluation.
%of a formula. 
This differs from the semantics often used in the model-checking problem, where
%in the model-checking problem, 
each probabilistic operator in the formula might be associated with a different strategy in principle. On the other hand, in the synthesis problem as stated above, a unique strategy is 
%fixed and 
used to evaluate all the probabilistic operators. 
%Phrase de Raskin 

%The synthesis problem for \PCTL\ is known to be algorithmically challenging and, in several cases, undecidable. We briefly summarize the main results. 

\paragraph{Slater's Generalized Assumption}
We provide an analogue of Slater's assumption in the case where nesting are considered, that we call Slater's generalized assumption. Intuitively, a Markov Chain $\MDP$ and a formula $\Phi$ satisfy Slater's generalized assumption if $\MDP \models \Psi$ for any $\Psi$, where $\Psi$ is obtained by replacing all $p_i$ in the formula $\Phi$ by $q_i>p_i$.

\begin{definition}\textbf(Slater's generalized assumption)
    We consider $\MDP$ a MDP and $\Phi$ a \CPCTL$\,$formula. %We denote $SF(\Phi)= \{\Phi_1,\dots,\Phi_{sf(\Phi)}\}$ and $PF(\Phi) = \{\phi_1,\dots,\phi_{pf(\Phi)}\}$, respectively, the state subformulas and the path subformulas of $\Phi$. Since $pf(\Phi) \leq sf(\Phi)$, we also order such that for any $i\leq pf(\Phi)$, $\Phi_i=\mathbb{P}_{\geq p_i} (\phi_i)$. 
    For $\boldsymbol d = (d_1,\dots,d_{pf(\Phi)})$, we define $\mathtt{T}(\Phi,\boldsymbol d)\in CPTCL$ by induction as $\mathtt{T}(\Phi,\boldsymbol d) = \mathtt{T}(\Phi_{j_1},\boldsymbol d)\wedge \mathtt{T}(\Phi_{j_2},\boldsymbol d)$ for $\Phi=\Phi_{j_1}\wedge \Phi_{j_2}$, $$\mathbb{P}_{\geq d_j}[\mathtt{T}(\Phi_{j_1},\boldsymbol d)\mathbf{W}\mathtt{T}(\Phi_{j_1}\wedge \Phi_{j_2},\boldsymbol d)]$$ for $\Phi=\mathbb P_{\geq q_j}\left[\Phi_{j_1}\mathbf{W}\Phi_{j_2}\right]$, $a$ for $\Phi=a\in AP$ and $\neg a$ for $\Phi=\neg a,~a\in AP$.

    We say that $\MDP,\Phi$ satisfy the generalized Slater's Assumption $(gSA)$ if there exists $\boldsymbol p' = (p'_1,\dots,p_{pf(\Phi)}')$ satisfying $p'_i > p_i$ for all $i$, such that there exists $\pi$ satisfying $
    \MDP_{\pi} \models \mathtt{T}(\Phi,\boldsymbol p')$.
\end{definition}

\subsection{Structural \CPCTL$\,$ Properties}

The results presented hereafter are essential for the initialization of the value iteration algorithm to solve the synthesis problem, introduced in Sec.~\ref{sec:VI}. They describe the set $\mathbb{P}_{=1}(\phi)$, for path formula $\phi$, which corresponds to the starting point of the algorithm.

The intuition for the name "Continuing \PCTL" comes from the following lemma. It states that for $\Phi\in$\CPCTL, the satisfaction of $\Phi$ can be postponed, as long as we stay on the literal projection of $\Phi$, a boolean logic formula that we introduce now, and $\Phi$ is 
%always 
satisfied later on.

\begin{definition}[Literal projection]
    We define the literal projection $\alit{\phi}$ of a \CPCTL$\,$formula $\phi$ by induction on $\phi$.
    \begin{itemize}
        \item For any literal $a$, $\alit{a}=a$,
        \item for any \CPCTL$\,$state formulas $\Phi_1$ and $\Phi_2$, we have $\alit{\Phi_1\land\Phi_2}=\alit{\Phi_1}\land \alit{\Phi_2}$ and $$\alit{\mathbb P_{\geq p}\left[\Phi_1\mathbf{W}(\Phi_1\land \Phi_2)\right]}=\begin{cases}
            \top & \textbf{ if } p=0\\
            \alit{\Phi_1} & \textbf{ otherwise. }
        \end{cases}$$
    \end{itemize}
\end{definition}

The Literal Projection is the strongest necessary condition in boolean logic that a state has to satisfy to satisfy a state \CPCTL$\,$ formula. For example, for the following formulas
\[
\begin{aligned}
&\Phi_1 = \mathbb{P}_{\geq 0.7} \left( (a\wedge  \neg b) \mathbf{W}((a\wedge \neg b) \mathbf{W} c) \Rightarrow\alit{\Phi_1}  \right) = a \wedge \neg b, \\
&\Phi_2 = \mathbb{P}_{\geq 0.5} \left( \phi_2 \right) \wedge \neg b \Rightarrow \alit{\Phi_2} = a \wedge \neg b.
\end{aligned}
\]
%we have the necessary conditions $\alit{\Phi_1} = a \wedge \neg b$ and $ \alit{\Phi_2} = a \wedge \neg b$.
% \[
% \begin{aligned}
% &\quad \alit{\Phi_1} = a \wedge \neg b,\\
% &\quad \alit{\Phi_2} = a \wedge c \wedge \neg b.
% \end{aligned}
% \]
Define $\phi_1,\phi_2$ such that $\Phi_1 = \mathbb{P}_{\geq p} (\phi_1)$ and $\Phi_2 = \mathbb{P}_{\geq 0.5} (\phi_2)$. For any state $s$ in a Markov Chain $\mathcal{M}$, we have $s \models \Phi_1 \Rightarrow a \wedge \neg b$, otherwise we would have $\mathbb{P}(\phi_1|s,\mathcal{M})=0$ since for every path $\xi$, we would have $\xi \not \models \phi_1$. Similarly, $s\models \Phi_2 \Rightarrow s \models \neg b$, and $s \models \mathbb{P}_{\geq 0.5} (\phi_2)$. 
%Again, $s \models \Phi_2 \Rightarrow s \models a$, otherwise we would have $\mathbb{P}(\phi_2|s,\mathcal{M})=0$.
%\fb{Say where the lemma is needed or why useful}
Conversely, a quick induction shows that satisfying $\mathbb{P}_{\geq 1} (\mathbf{G}\alit{\Phi})$ implies the satisfaction of $\Phi$. For example, $\mathbb{P}_{=1}(\mathbf{G} (a \wedge \neg b))$ implies $\Phi_1$.
%, as $b$ is satisfied at $s$ and every path satisfies $\mathbf{G}(a \wedge \neg b)$, which implies in particular $(a \wedge \neg b)\mathbf{W} ( (a\wedge \neg b)\mathbf{W} c)$.
The following lemma precisely formalizes this intuition.

\begin{restatable}{lemma}{forsurecharacterization}\label{lemma-alit-W}
    For every MDP $\mathcal M$, path formula $\phi=\Phi_1\mathbf{W}(\Phi_1\land\Phi_2)$, and policy $\pi$,
    %of $\mathcal M$, 
    we have  $\mathcal M_{\pi}\models \mathbb P_{\geq 1}[\phi]$ iff $\mathcal M_{\pi}\models \mathbb P_{\geq 1}[\alit{\Phi_1}\mathbf{W}(\Phi_1\land\Phi_2)]$.
\end{restatable}

\begin{corollary}
    For every MDP $\mathcal M$, $\Phi\in$\CPCTL$\,$ and policy $\pi$,
    %of $\mathcal M$, 
    we have $\mathcal M_{\pi}\models \mathbb P_{\geq 1}(\mathbf{G} (\alit{\Phi})\Rightarrow \mathcal{M}_\pi \models \Phi$. 
\end{corollary}

% Note that the direction $ \mathcal M_{\pi}\models \mathbb P_{\geq 1}[\phi] \Rightarrow \mathcal M_{\pi}\models \mathbb P_{\geq 1}[\alit{\Phi_1}\mathbf{W}(\Phi_1\land\Phi_2)]$ is trivial. The converse however does not hold even for $\Phi\in$ \safePCTL, 

The intuition of the results \textbf{cannot} be extended further. 
\begin{itemize}[leftmargin=4mm]
\item (Equivalence) %for $\Phi$ a \CPCTL$\,$ formula, one can \textbf{not} write $\Phi = \mathbb{P}_{\geq 1} (\mathbf{G} \alit{\Phi})$. 
For $\Phi=\mathbb{P}_{\geq p} [\Psi_1\mathbf{W}\Psi_2]$, one cannot write $\Phi = \mathbb{P}_{\geq 1} (\mathbf{G} \alit{\Phi})$ even assuming that $\Psi_2 = \bot$ so that $\Phi$ never "stops" and always continues.
\item (Paths) The result can not be extended to paths. For instance, if we consider the Markov Chain $\mathcal{M}$ with $\mathcal{S}=\{ s_1,s_2 \}$, $L(s_1)=\emptyset$, $L(s_2)=\{a\}$, and $\mathbb{P}(s_1|s_1) = \mathbb{P}(s_2|s_1)=1/2$, $\mathbb{P}(s_2|s_2)=1$ and the path formula $\phi= (\mathbb{P}_{\geq 2/3} \mathbf{G} a  ) \mathbf{W} \bot$, then the path $\xi =s_1,s_1,s_1,\dots$ satisfies $\left( \alit{\mathbb{P}_{\geq 1/2} \mathbf{G} a } \right) \mathbf{W} \bot $ but does not satisfy $\phi$.
\item (\safePCTL) This key structural property is not true for general safe formulas. For example, no boolean logic formula on a state $s$ can be deduced provided that $s\models \mathbb{P}_{\geq 1}(\mathbf{X} a )$.
\end{itemize}

\subsection{Expressivity of \CPCTL}

The restrictions imposed on \CPCTL \  formulas do not imply a one dimensional monotonicity, and maximizing several nested properties does not reduce to maximizing the last one, as shown in example \ref{ex:expressivityone}.

%\fb{let's put this as a theorem} 
\begin{example}\textbf{(No collapse of nested objectives)}\label{ex:expressivityone}

%Even in the most simple cases, satisfying a \CPCTL$\,$formula does not reduce to maximizing a unique goal. 

We consider the formula
\[
\Phi = \mathbb{P}_{\geq p_1} [\phi],\quad \phi=\mathbf G\Phi_a,\quad \Phi_a = \mathbb{P}_{\geq p_1} [\mathbf G \neg a],\quad p_1=\frac{7}{12}.
\]

The path formula $\mathbf G [\Phi_a]$ means that we must stay on states such that the probability of reaching $a$ from those states is lower than $p_1$. We look for $\pi$ that maximizes the probability of those paths. One intuition may be that there is never a reason to visit $a$ more often. However, the safest policy that minimizes the probability of reaching $a$ from any state does not maximize $\Phi$.
We consider the $\MDP$ presented in Figure \ref{fig:exampleexpressivity1}.
%with $\mathcal{S}=\{s_0,\dots,s_8\}$. 
The actions are represented by arrows labelled by the associated probabilities. The states without an outgoing arrow have only one action looping on themselves with probability $1$. Finally, $\llbracket a \rrbracket= \{s_5,s_7\}$.

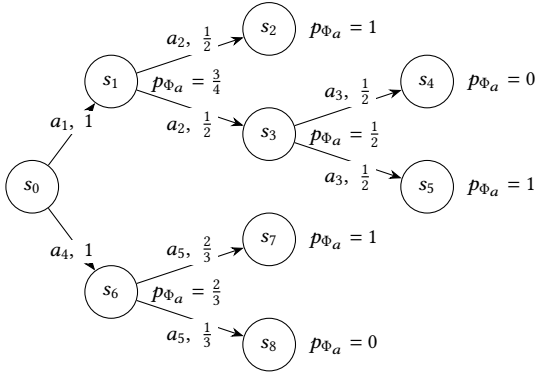
\begin{figure}[!h]
\begin{tikzpicture}[
  x=7mm, y=7mm, 
  >=Stealth,
  state/.style={circle, draw, minimum size=7mm, inner sep=0pt, font=\small},
  el/.style={midway, fill=white, inner sep=1.2pt, font=\small},
  plabel/.style={font=\small, fill=white, inner sep=0.8pt},
  every label/.style={plabel, xshift=1.5mm}
]

\node[state] (s0) at (0, 0) {$s_0$};

\node[state, label=right:{$p_{\Phi_a}=\tfrac{3}{4}$}] (s1) at (1.5,  2) {$s_1$};
\node[state, label=right:{$p_{\Phi_a}=\tfrac{2}{3}$}] (s6) at (1.5, -2) {$s_6$};

\node[state, label=right:{$p_{\Phi_a}=1$}]            (s2) at (4.5,  3) {$s_2$};
\node[state, label=right:{$p_{\Phi_a}=\tfrac{1}{2}$}] (s3) at (4.5,  1) {$s_3$};

\node[state, label=right:{$p_{\Phi_a}=1$}]            (s7) at (4.5, -1) {$s_7$};
\node[state, label=right:{$p_{\Phi_a}=0$}]            (s8) at (4.5, -3) {$s_8$};

\node[state, label=right:{$p_{\Phi_a}=0$}]            (s4) at (7.5,  2) {$s_4$};
\node[state, label=right:{$p_{\Phi_a}=1$}]            (s5) at (7.5,  0) {$s_5$};

\draw[->] (s0) -- (s1) node[el, above] {$a_1,\;1$};
\draw[->] (s0) -- (s6) node[el, below] {$a_4,\;1$};

\draw[->] (s1) -- (s2) node[el, above] {$a_2,\;\tfrac{1}{2}$};
\draw[->] (s1) -- (s3) node[el, below] {$a_2,\;\tfrac{1}{2}$};

\draw[->] (s3) -- (s4) node[el, above] {$a_3,\;\tfrac{1}{2}$};
\draw[->] (s3) -- (s5) node[el, below] {$a_3,\;\tfrac{1}{2}$};

\draw[->] (s6) -- (s7) node[el, above] {$a_5,\;\tfrac{2}{3}$};
\draw[->] (s6) -- (s8) node[el, below] {$a_5,\;\tfrac{1}{3}$};

\end{tikzpicture}
\caption{The Markov Decision Process $\mathcal{M}_1$}\label{fig:exampleexpressivity1}
\end{figure}

% Since after having reached $s_1$ or $s_6$, there is no choice involved anymore for the policy, we can already compute the probability of non-reaching $a$ for $s_i$, $i\neq 0$.
% \[
% \begin{aligned}
%     &\mathbb{P}(\mathbf{G}\neg a|s_4,\pi) = \mathbb{P}(\mathbf{G}\neg a|s_8,\pi) = 0,\\
%     &\mathbb{P}(\mathbf{G}\neg a|s_2,\pi) = \mathbb{P}(\mathbf{G}\neg a|s_5,\pi) = \mathbb{P}(\mathbf{G}\neg a|s_7,\pi) = 1,\\
%     &\mathbb{P}(\mathbf{G}\neg a|s_3,\pi)=1/2,\quad\mathbb{P}(\mathbf{G}\neg a|s_6,\pi)=2/3,\\
%     &\mathbb{P}(\mathbf{G}\neg a|s_1,\pi) =3/4.
% \end{aligned}
% \]

The policy that maximizes the probability of avoiding $a$ $\pi_1$ chooses $a_1$ with probability $1$. In this case, $
\mathbb{P}(\mathbf{G}\neg a|\pi,s_0) = 3/4$.
Now, the states that satisfy $\Phi_a = \mathbb{P}_{\geq 7/12} [\mathbf{G} \neg a]$ are $
\llbracket \Phi_a \rrbracket = \{ s_0,s_1,s_2,s_5,s_6\}$.
Hence, the probability of never reaching a state in $\mathcal{S}\setminus \llbracket \Phi_a \rrbracket$ is equal to $1/2$. If we consider $\pi_2$ however, the policy taking $a_1$ with probability $0$ and $a_4$ with probability $1$, since $s_6$ and $s_7$ are in $\llbracket \Phi_a \rrbracket $ and $s_8 \notin \llbracket \Phi_a \rrbracket $, we obtain
$\mathbb{P}(\phi|\pi_1,s_0) = \frac{1}{2} < \frac{2}{3} = \mathbb{P}(\phi|\pi_2,s_0)$. Thus, in this case, we have $s_0,\pi_1 \not\models \Phi$ while $s_0,\pi_2 \models \Phi$.

\end{example}

%\subsection{No reduction from \CPCTL$~$ to flat formulas}
\subsection{A New Computable Class}

The following theorem shows that \CPCTL\ specifications cannot be reduced to flat \PCTL\ formulas and form a new class of computable specifications.

\begin{theorem}
    There exists formulas $\Phi\in$ \CPCTL\ such that for
    %there exists 
    no flat formula (nesting depth one) $\Psi\in$ \PCTL, we have that for any Markov chain $\MDP$,
    $
    \MDP \models \Phi \Leftrightarrow \MDP \models \Psi.
    $
    \end{theorem}

    Equivalently, for such formulas, there exists no flat formula $\Psi \in \PCTL$  such that for any MDP $\mathcal{M}$ and any policy $\pi$,
    %on $\MDP$,
     we have $    \MDP_\pi \models \Phi \Leftrightarrow \MDP_\pi \models \Psi.
    $
    In particular, currently  there exists no algorithm for the synthesis of $\CPCTL$ specifications.

\begin{proof}
    We provide an example of a \CPCTL\ formula that is not equivalent to any flat formula. We consider the nested formula
$
\Phi = \mathbb P_{\geq 1}\left(c \mathbf{W} \mathbb P_{\geq 1/2}(c \mathbf{W} (c\wedge a)) \right) \in \CPCTL,
$
and assume the existence of a flat formula $\Psi$ of the form 
    $
    \Psi = \mathcal{B}(\Psi_1,\dots,\Psi_l),\quad \Psi_k = \mathbb P_{\geq p_k}(\psi_k)
    $
    where $\mathcal{B}$ is an operator that involves only a finite number of boolean operations and atomic propositions, and $\psi_k$ are formulas of the form $\psi_k = \Psi_{k,1} \mathbf{O}_k \Psi_{k,2}$ where $\mathbf{O}_k\in \left\{ \mathbf{U},~\mathbf{W} \right\}$ and $\Psi_{k,1},\Psi_{k,2}$ are boolean formulas. We introduce the Markov Chain $\mathcal{M}_{\alpha,\varepsilon}$ as defined in Figure \ref{fig:exampleexpressivity2}.

    \begin{figure}[h]
    \centering
\begin{tikzpicture}[
  state/.style={draw,circle,minimum size=6mm,
  node distance=0.7cm}, 
  >=Stealth
]
\node[state] (root) {$\{c\}$};

% Level 1
\node[state] (left1) [below left=of root] {$\{c\}$};
\node[state] (right1) [below right=of root] {$\emptyset$};

% Level 2 (from left1)
\node[state] (left2) [below left=of left1] {$\{c,a\}$};
\node[state] (right2) [below right=of left1] {$\{b\}$};

% Edges from root
\draw[->] (root) -- node[above left] {$1-\varepsilon$} (left1);
\draw[->] (root) -- node[above right] {$\varepsilon$} (right1);
\draw[->] (left1) -- node[above left] {$\alpha$} (left2);
\draw[->] (left1) -- node[above right] {$1-\alpha$} (right2);
\draw[->] (left2) to [out=230,in=310,looseness=3] (left2);
\draw[->] (right2) to [out=230,in=310,looseness=3] (right2);
\draw[->] (right1) to [out=230,in=310,looseness=3] (right1);

\end{tikzpicture}
    \caption{The Markov Chain $\MDP_{\alpha,\varepsilon}$.}\label{fig:exampleexpressivity2}
    \end{figure}
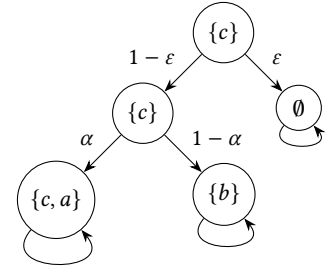

For any $\varepsilon>0$, we have $\mathbb{P}(\phi|\MDP_{\alpha,\varepsilon})\leq 1-\varepsilon <1$, so $\MDP \not\models \Phi$. For $\varepsilon=0$, we denote $\xi$ the path obtained by going to the left once and to the right once. $\xi$ has label $\{c\}\{c\}\{b\}^*$ and does not satisfy the path formula $c \mathbf{W} [c \mathbf{W} (c\wedge a)]_{\geq 1/2}$ because $c \notin \{b\}$, except if there exists $i$ such that 
$
\xi_i \models [c \mathbf{W} (c\wedge a)]_{\geq 1/2}.
$
$i<2$ which means $i = 0$ or $i=1$. We have $\mathbb{P}(c \mathbf{W} (c\wedge a)|\xi_i)=\alpha$ for $i=1,2$, so $\xi \models \phi$ if and only if $\alpha \geq 1/2$.

%Hence, the set $\mathcal{E}_\Phi$ of $\alpha,\varepsilon$ such that $\MDP_{\alpha,\varepsilon}\models \Phi$ satisfies $\mathcal{E}_{\Phi}=\left[ \frac{1}{2},1 \right]\times \{0\}$.

Hence, $\mathcal{E}_\Phi = \{ (\alpha,\varepsilon),~\MDP_{\alpha,\varepsilon}\models \Phi \} = \left[ \frac{1}{2},1 \right]\times \{0\}$.

The following conditions are incompatible for a flat formula $\Psi$:
\[
(i)~\forall \varepsilon >0,~\forall \alpha \in [0,1], ~(\alpha,\varepsilon)\notin \mathcal{E}_{\Psi},~(ii)~\left( \frac{1}{2},0 \right)\in \mathcal{E}_\Psi \wedge \left( \frac{1}{3},0 \right)\notin \mathcal{E}_\Psi.
\]

% \begin{itemize}
%     \item $\forall \varepsilon >0$, $\forall \alpha \in [0,1]$, $(\alpha,\varepsilon)\notin \mathcal{E}_{\Psi}$.
%     \item $\left( \frac{1}{2},0 \right)\in \mathcal{E}_\Psi$.
%     \item $\left( \frac{1}{3},0 \right)\notin \mathcal{E}_\Psi$.
% \end{itemize}

% By induction, $(i)$ implies the existence of a flat formula $\Psi_j$ such that $\Psi_j = [\psi_j]_{\geq 1}$ and satisfying
% \[
% \left\{
% \begin{aligned}
% &\varepsilon  >0 \wedge \alpha = \frac{1}{2} \Rightarrow (\MDP_{\alpha,\varepsilon} \not\models \Psi_j), \\
% & \varepsilon =0 \wedge \alpha = \frac{1}{2}\Rightarrow (\MDP_{\alpha,\varepsilon} \models \Psi_j)
% \end{aligned}
% \right.
% \]

To see it, note that the evaluation of path formulas $\psi_k$ on a path $\xi=(\xi_0\xi_1\dots)$ only depend on the sequence of labels $L(\xi_0)L(\xi_1)\dots$. With $\xi_{ll}, \xi_{lr}, \xi_{r}$, the paths respectively going to the left then left, to the left then right, and to the right, we have with $\mathbf{B}(True)=1$ and $\mathbf{B}(False)=0$,
\[
\begin{aligned}
\mathbb{P}(\psi_j|\MDP_{\alpha,\varepsilon}) = &(1-\varepsilon) \alpha \mathbf{B}( \xi_{ll}\models \psi_j) \\
&+(1-\varepsilon) (1-\alpha) \mathbf{B}(\xi_{lr} \models \psi_j) + \varepsilon \mathbf{B}(\xi_{r}\models \psi_j).
\end{aligned}
\]
This condition restricts expressivity and prevents the whole formula $\Psi$ from satisfying the previous condition. 

\end{proof}
% Example \ref{ex:expressivity2} provides a $\CPCTL$ formula $\Phi$ and a family of MDP $\MDP_{\alpha,\varepsilon}$ such that there exists no flat formula $\Psi$ in $\PCTL$ that coincides with $\Phi$ on every $\MDP_{\alpha,\varepsilon}$. In particular, there exists currently no algorithm for the synthesis of $\CPCTL$ specifications, and its decidability is not yet established.

\section{From safe-\PCTL$~$ satisfaction to local constraints}
\label{sec:augmented}

%In this section, we construct from an MDP $\mathcal M$ and a Safe PCTL formula $\Phi$ a safe augmented MDP $\AugMDP$ in which every policy satisfies the $\Phi$, and whose optimal policy easily gives a \emph{memoryful} policy that is a solution to Problem 2 for $\mathcal M$ and $\Phi$.

In this section we construct, for a safe-PCTL formula $\Phi$  and  an MDP $\MDP$, an augmented MDP $\MDP\times [\Phi]$, such that satisfying $\Phi$ in $\MDP$ is equivalent to satisfying local constraints provided in Def.~\ref{def:path-compatibility} in $\MDP \times [\Phi]$.
%\fb{which ones?} in $\MDP[\Phi]$. \fb{} 
%\fb{we should explain why we consider safe-PCTL in this section.}

%\fb{move realizability to the next sectio.}
%We say that a state $\aug{s}$ of $\MDP\times [\Phi]$ is {\em realizable} if there exists a policy such that starting from this state, all constraints are satisfied. 
%Moreover, those constraints can be expressed as inequalities. It is not straightforward to solve those constraints as $\MDP[\Phi]$ will be infinite. So it may be that, starting from the initial state, $\aug \pi$ satisfies the first inequalities, but reaches a state in which those inequalities can't be satisfied. We say an augmented state is realizable if there exists a policy such that starting from this state, all constraints are satisfied. 
%We will show a variety of properties that will enable us, in the next section, to provide a value iteration algorithm for solving the synthesis problem under a generalized Slater's assumption. An immediate property for instance, is that if there is a distribution of actions from a state $\aug s$ in the augmented MDP, satisfying the constraints, and such that all the successors of $\aug s$ are realizable, then $\aug s$ is realizable. 

\begin{definition}[Augmented MDP]
We define $\MDP \times [\Phi]$ -- {\em the  MDP augmented by formula $\Phi$} -- as follows:
\begin{itemize}
    \item Set of States: Let $\aug{\St}_{\geq 0}$ be the set of all triples $(s,\boldsymbol \mu,\boldsymbol \nu)$ in $\St \times \{ 0,1 \}^{pf(\Phi)} \times [0,1]^{sf(\Phi)}$. Each (augmented) state $\aug{s} \in \aug{\St}_{\geq 0}$ is of the form $(s,\boldsymbol \mu,\boldsymbol \nu)$, with 
    \[
    \boldsymbol \mu=(\mu_1,\dots,\mu_{sf(\Phi)}),~\boldsymbol \nu = (\nu_1,\dots,\nu_{pf(\Phi)}),
    \]
    where 
    %the numbers 
    $\nu_i \in [ 0,1 ]$ (resp.~$\mu_i \in \{0, 1\}$) are 
    %called 
    the counters of the path subformulas $\phi_i\in \mathcal{PF}(\Phi)$ of $\Phi$, (resp. the valuations of the state subformulas $\Phi_i\in \mathcal{SF}(\Phi)$ of $\Phi$).
    %, and the integers $\mu_i$ are called the valuations corresponding to the state formulas $\psi_i$ \fb{which path and state formulas? where do these come from?}. 
    Finally, we define $\aug{\St}=\aug{\St}_{\geq 0} \cup \{\aug{s}_{-1} \}$, 
    %\item{Initial State:} 
    with initial state $\aug s_{-1}$.

    \item Set of Actions: Let $\aug{s} = (s,\boldsymbol \mu,\boldsymbol \nu)\in \aug{S}$, $\Act$ the set of actions available at $s$ in $\MDP$, and $\{s_1,\dots,s_m\}$ the successors of $s$ in $\MDP$. Then, for any $\aug s\in\aug{\St}\setminus \{\aug s_{-1}\}$ we define the (augmented) set of actions $\aug{\mathcal{A}}(\hat s)$ as the set of $(a,\boldsymbol\mu^1,\ldots,\boldsymbol\mu^m,\boldsymbol\nu^1,\ldots,\boldsymbol\nu^m)$ such that $a\in\mathcal A(s)$, each $\boldsymbol \nu^i$ (resp. $\boldsymbol \mu^i$) is a set of counters (resp. valuations) for the successor $s_i$, and $(s_i,\boldsymbol \mu^i,\boldsymbol \nu^i)\in \aug\St$. 

    Moreover, the set $\aug{\mathcal{A}}(\aug s_{-1})$ is defined as the set of all $(\boldsymbol \mu,\boldsymbol \nu)$ such that $(s_0,\boldsymbol \mu,\boldsymbol \nu)\in \aug\St$, where $s_0$ is the initial state of $\mathcal M$. 
    
    \item Transitions: For each $\aug a = (a,\boldsymbol \mu^1,\ldots, \boldsymbol \mu^m, \boldsymbol \nu^1,\ldots, \boldsymbol \nu^m)\in \AugAct$, %$\boldsymbol  \mu=(\mu_1,\dots,\mu_{st(\Phi)})$, $\boldsymbol \nu=(\nu_1,\dots,\nu_{pf(\Phi)})   $, for $\aug s' = (s_i,\boldsymbol \mu',\boldsymbol \nu')$ 
    $P(\aug s,\aug a)$ is the mixture of Dirac distributions $\sum_{i=1}^m P(s_i|s,a)\delta_i$ such that the support of $\delta_i$ is $\{(s_i,\boldsymbol \mu^i,\boldsymbol \nu^i)\}$.
    % \[
    % \begin{aligned}
    % &P(\aug s'| \aug a,\aug s) = P(s_i|a,s), \text{if } \boldsymbol \mu'=\boldsymbol \mu \text{ and } \boldsymbol \nu'=\boldsymbol \nu; \\
    % &P(\aug s' | \aug a,\aug s) = 0,~\text{otherwise} 
    % %\mu'\neq \mu \text{ or } \nu'\neq \nu,\\
    % \end{aligned}
    % \]
%    To put it differently, the 
    That is, action $\aug a$ performs similarly to action $a$ in $\MDP$, but the agent can choose the valuations and counters for the successors and ends up in the augmented state $(s_i,\boldsymbol \mu^i, \boldsymbol\nu^i)$ rather than $s_i$.

    Moreover, for any action $\aug a=(\boldsymbol\mu,\boldsymbol\nu)$ available from the initial state, $P(\aug s_{-1},\aug a)$ is the Dirac distribution with support $\{ (s_0,\boldsymbol \mu,\boldsymbol \nu) \}$ where $s_0$ is the initial state of $\MDP$.

    %for any $\aug s'=(s',\boldsymbol\mu',\boldsymbol\nu') \in \aug\St$, $P(\aug s' \mid \aug a,\,\aug s_{-1})=1$ if $s'=s_0$, $\boldsymbol\nu'=\boldsymbol\nu$ and $\boldsymbol\mu'=\boldsymbol\mu$, and  $P(\aug s' \mid \aug a,\aug s_{-1})=0$ otherwise.

    % \item Finally, the initial state is $\aug s_{init}$, the set of labels of $(s,\mu,\nu)$ is the same as the set of labels of $s$, and the reward of $(s,\mu,\nu)$ is the same as the one of $\aug s_{init}$. The reward of $s$ is $0$ and the set of labels of $\aug s_{init}$ is arbitrary.
\end{itemize}
\end{definition}

The augmented MDP $\MDP\times [\Phi]$ is infinite, as the set of states is the entire $\aug{ \mathcal{S}} = \mathcal{S}\times \{ 0,1 \}^{sf(\Phi)} \times [0,1]^{pf(\Phi)}$. However, in Section ~\ref{sec:VI} we show how to navigate a finite portion of states and actions to obtain a policy that satisfies $\Phi$ in $\MDP$.

We next define valued policies. Intuitively, these corresponds to policies $\aug \pi$ on $\MDP\times[\Phi]$ that are deterministic on the counters and valuations of the augmented actions.
%but stochastic on the part of the action corresponding to actions in the original MDP.
%
\begin{definition}\textbf{(Valued Policy)}\label{def:valued}
    %For a policy $\aug \pi$ of $\AugMDP=\MDP[\Phi]$, we say that $\aug \pi$ is a \textbf{valued} policy if and only if, for every state $\aug s \in \AugSt$, there exist $m=\#\suc(s)$ counters and valuations $(\nu_i)$ and $(\nu_i)$ for the successors such that for any $\aug a = (a,(\nu_i')_i,(\nu_i')_i)$, if there exists $j$ such that $\nu_j\neq \nu_j'$ or $\nu_j\neq \nu_j'$
    %the action $\aug a$ is never selected, i.e.
    %\[
    %P(\aug a| \aug \pi,\aug s) =0.
    %\]
%
    A policy $\aug \pi$ is {\em valued} iff there exist functions $\Theta$ and $\Delta$ s.t.~for $m(\hat s)=\#\text{succ}_\MDP(s)$ and $\aug s$ an augmented state in $\mathcal{M}\times[\Phi]$,
    \[
    \Theta:\aug s\mapsto ((\boldsymbol \mu^1,\boldsymbol  \nu^1),\dots,(\boldsymbol \mu^m,\boldsymbol \nu^m)),~\Delta:\aug s \mapsto \delta \in \mathcal{D}(\mathcal{A}(s)),
    \]
    %For $\aug s = (s,\boldsymbol \mu,\boldsymbol \nu)$.
    such that for any $a \in \mathcal{A}(s)$, with $\delta=\Delta(\aug s)$, we have
    \begin{equation}\label{defvaluedpolicy}
    \aug \pi(\aug a|\aug s) = \delta(a),~\text{ for } \aug a= (a,\Theta(\aug s)),
    \end{equation}
    and $ \aug \pi(\aug a|\aug s) = 0 $ for any other $\aug a$. With an abuse of notation, we denote $\Theta(\hat s)( s_i) ::= (\boldsymbol\mu^i,\boldsymbol\nu^i) $.
    
    %In other words, for each augmented state $\aug s \in \AugSt$, the policy $\aug \pi$ choses counters and valuations $G(\aug s)$, and then starts a distribution $\delta_{\aug s}$ on the actions $\Act$.

    %\[
   % \begin{aligned}
    %&P( \aug a | \aug \pi ) = P(a|F(\aug s)),~\text{ for } (\nu_i,\mu_i)_i=(\nu_i',\mu_i')_i,\\
    %&P( \aug a | \aug \pi ) =0 \text{ otherwise}.
    %\end{aligned}
   % \]
\end{definition}

    Additionally, we denote by $\aug s_{0,\aug \pi}$ the state reached from $\aug s_{-1}$ following $\aug \pi$ after one step, and we write $\aug s_{0}$ when $\aug\pi$ is clear from the context. Since the policy is valued it chooses a unique set of valuations and counters for each state $s\in \mathcal{S}$, it is in particular deterministic at $\aug s_{-1}$ and $\hat s_0$ is well defined.
    
    Finally, we denote by $\aug{S}_{\aug \pi}$ the set of states that are reachable from $\aug s_{0,\aug \pi}$ following $\aug \pi$.

    We say that a state $\aug{s}$ of $\MDP\times [\Phi]$ is {\em realizable} if there exists a policy such that starting from this state, all constraints are satisfied. 

We now introduce the compatibility conditions.

%We now introduce the two local conditions that will encode that the augmented policy satisfies a safe-PCTL formula $\Phi$. For an augmented state $\aug s= (s,\nu,\nu)$ where $\nu=(\nu_1,\dots,\nu_{N_1})$ are the counters and $\nu = (\mu_1,\dots,\mu_{N_2})$ the valuations, the \emph{path compatibility} will ensure that the counters $\nu$ \emph{could} correspond to the probabilities of the path formulas, and that the valuations \emph{could} correspond to the valuations of the state formulas.

\begin{definition}\label{def:state-compatibility}
    Let $\aug \pi$ be a valued policy on $\AugMDP = \MDP\times[\Phi]$, and denote the corresponding $\Theta$ and $\Delta$ as in Def.~\ref{def:valued}. 
    %Let $\aug{S}_{\aug \pi,0}$ the set of states reachable from $\aug s_0$ following $\aug \pi$.

    We say that $\aug \pi$ satisfies the \textbf{state compatibility} if for every $\aug s=(s,\boldsymbol\mu,\boldsymbol\nu)\in \aug{S}_{\aug \pi}$ with $\boldsymbol\nu=(\nu_1,\dots,\nu_{pf(\Phi)})$ and $\boldsymbol\mu=(\mu_1,\dots,\mu_{sf(\Phi)})$ the following are satisfied for any state formulas $\Phi_{j}$, $\Phi_{j_1}$, $\Phi_{j_2}$,   
    \begin{itemize}
        \item[(i)] 
        %For any state formula, 
        %of the form 
        If $\Phi_{j}=\mathbb{P}_{\geq p_j}(\phi_j)$,      
        %\[
          then  $(\mu_j = 1) \Rightarrow (\nu_j\geq p_j)$;
        %\]
        \item[(ii)] 
        %For every state formula of the form 
        If $\Phi_j = \Phi_{j_1}\wedge \Phi_{j_2}$,
   %     \[
     then $(\mu_j =1) \Rightarrow (\mu_{j_1}=1)$ and $(\mu_{j_2}=1)$. 
    %    \]
        \item[(iii)] 
        %For every state formula of the form 
        If $\Phi_j = \Phi_{j_1}\vee \Phi_{j_2}$,
      %  \[
       then $(\mu_j =1) \Rightarrow (\mu_{j_1}=1)$ or $(\mu_{j_2}=1)$. 
      %  \]
        \item[(iv)] 
        %For every state formula of the form 
        If $\Phi_j = b$ for some $b\in AP$, 
 %       \[
   then $(\mu_j = 1) \Rightarrow s \models b$.
  %      \]
    \end{itemize}
    
\end{definition}

The state compatibility ensures that the valuations of the augmented states that $\aug \pi$ can reach are coherent. 
%For instance, if $\Phi_3 = \Phi_1 \wedge \Phi_2$, we can't have $\mu_3 = 1$ and $\mu_1=0$. 
We introduce compatibility conditions for path formulas too.
\begin{definition}\label{def:path-compatibility}
    Let $\aug \pi$ be a valued policy on $\AugMDP = \MDP \times [\Phi]$, and denote the corresponding $\Theta$ and $\Delta$ as in Def.~\ref{def:valued}.

    %We say that 
    The policy $\aug \pi$ satisfies the \textbf{path compatibility} iff, for every augmented state $\aug s =(s,\boldsymbol\mu,\boldsymbol\nu) \in \aug{S}_{\aug\pi}$, with 
    \[
    \begin{aligned}
    &\Theta(\aug s)(s_i) = (\boldsymbol\mu^i,\boldsymbol\nu^i),~\boldsymbol\nu^i=(\nu_1^i,\dots,\nu_{pf(\Phi)}^i),~\boldsymbol\mu^i = (\mu_1^i,\dots,\mu_{sf(\Phi)}^i),\\
    &\boldsymbol\nu=(\nu_1,\dots,\nu_{pf(\Phi)}),~\boldsymbol\mu=(\mu_1,\dots,\mu_{sf(\Phi)}),~\delta=\Delta(\aug s),
    \end{aligned}
    \]
    the following conditions are satisfied with $m=\#\text{succ}_\MDP(s)$:

    \begin{itemize}
        \item For every $\phi_j$ path formula of the form $\phi_j = \Phi_{j_1}\mathbf{W} (\Phi_{j_1}\wedge \Phi_{j_2})$ (the continuing-weak-until), we have 
        \[
        \nu_j \leq \max\left(\mu_{j_1} \left( \sum_{i=1}^m \sum_a P(s_i|s,a) \delta(a) \nu_j^i \right),\mu_{j_1}\mu_{j_2}\right). 
        \]
        \item For every $\phi_j$ path formula of the form $\phi_j = \Phi_{j_1}\mathbf{W} \Phi_{j_2}$, we have 
        \[
        \nu_j \leq \max\left(\mu_{j_1} \left( \sum_{i=1}^m \sum_a P(s_i|s,a) \delta(a) \nu_j^i \right),\mu_{j_2}\right). 
        \]
        %\[
        %\nu_j \leq \mu_{j_1} \left( \sum_{i=1}^m P(s_i|a) \delta_{\aug \pi,\aug s}(a) \max(\nu_j^i,\mu_{j_2}^i) \right). 
        %\]
        \item For every $\phi_j$ path formula of the form $\phi_j = \mathbf{X}\Phi_{j_1}$, we have
        \[
        \nu_j \leq \left( \sum_{i=1}^m \sum_a P(s_i|s,a) \delta(a) \mu_{j_1}^i \right). 
        \]
    \end{itemize}
    
\end{definition}

The path compatibility ensures that the counters are coherent between an augmented state and its successors.
%, in the sense that it satisfies the standard properties of probabilities.
%
\begin{restatable}[Coherence]{theorem}{thcoherence}\label{th:coherence} 
    Let $\aug \pi$ be a valued policy that 
    %$\aug \pi$ 
    also satisfies the state and path compatibility conditions.
    %and the state compatibility. 
    Then, for every augmented state $\aug s = (s,\boldsymbol \mu,\boldsymbol \nu)$  reachable from $\aug s_0$ through $\aug \pi$, we have 
    %with $\mu=(\mu_1,\dots,\mu_{sf(\Phi)})$ and $\nu=(\nu_1,\dots,\nu_{pf(\Phi)})$,  
    %Then, the policy $\aug \pi$ defines a policy $\pi$ on the MDP $\MDP$ such that for every $\aug s = (s,\mu,\nu)$ augmented state reachable by $\aug \pi$, we have with $\nu=(\nu_1,\dots,\nu_N)$ and $\mu=(\nu_1,\dots,\nu_N)$, we have 
    \begin{itemize}
        \item The counters are coherent: for every $j \in \{1,\dots,{pf(\Phi)}\}$ and path formula $\phi_j$,
        %\[
        %\mathbb{P}(\phi_j | \aug s,\aug \pi) \geq \nu_j,
        %\]
        %or equivalently,
        $(\AugMDP_{\aug \pi}, \aug s) \models \mathbb{P}_{\geq \nu_j }[\phi_j]$.

        \item The valuations are coherent: for every $j \in \{1,\dots,sf(\Phi) \}$ and state formula $\Phi_j$, 
        $
        (\mu_j = 1) \Rightarrow (\AugMDP_{\aug \pi},\aug s) \models \Phi_j.
        $
    \end{itemize}
\end{restatable}

    %Not that the converse is not needed, meaning that we do not require $(\nu_j = 0) \Rightarrow \AugMDP_{\aug \pi}(\aug s) \not\models \Phi_j$.

\begin{proof}
    The detailed proof is deferred to Appendix~\ref{appendix:coherencethm} and illustrates how our formalism can be used to show how infinite inductive propagations of the conditions ensure the satisfaction of a global formula. 
\end{proof}

\section{Value Iteration for \CPCTL$\,$Synthesis}

\label{sec:VI}

Using the Coherence Theorem~\ref{th:coherence}, our goal is to construct iteratively 
%a set of 
augmented states $(s,\boldsymbol\mu,\boldsymbol\nu)$ that are \emph{realizable}, meaning that there exists a policy $\aug \pi$ on $\MDP[\Phi]$ such that 
for all $i$, $\mu_i = 1\Rightarrow \left(\MDP_{\aug \pi},(s,\boldsymbol \mu,\boldsymbol \nu)\right) \models \Phi_i$, and for all $j$, $\left(\MDP_{\aug \pi},(s,\boldsymbol \mu,\boldsymbol \nu)\right) \models \mathbb{P}_{\geq \nu_j} [\phi_j]$. Using the state and path conditions as the backbone of the Bellman operator that we will define, we compute augmented states 
%for which we are certain 
that are provably realizable. For $\Phi_j=b\in AP$ an atomic proposition and $b\in L(s)$, we can immediately deduce that the state $(s,\boldsymbol \mu,\boldsymbol \nu)$ where $\mu_k=1$ for $j=k$, $\mu_k=0$ otherwise, and $\nu_k=0$ for all $k$ is realizable. This initialization, however, does not suffice for the Value-Iteration to reach all the realizable states. We know for instance from \cite{DBLP:journals/tcs/HaddadM18} that formulas involving the temporal operator $\mathbf G\Phi$ require the pre-computation of $\mathbb{P}_{=1}[\mathbf G\phi]$. We extend this idea to the continuing operator $\Phi_1 \mathbf{W} (\Phi_1\wedge \Phi_2)$ and construct a suitable initialization. We then combine these elements and design an algorithm, $\mathrm{CPCTL-VI}$, which solves the synthesis problem. In Theorem~\ref{thm:VI-optimality}, we prove its soundness and optimality under a generalized Slater's assumption.

\subsection{A Value Iteration Algorithm for \CPCTL}

For tuples $(\boldsymbol{\mu},\boldsymbol{\nu})$ in $\counterValuationSet{\Phi}$,
%that we denote from now on $(\hat \eta,\hat \mu)$, 
we define the partial order $\leq$ as
\begin{eqnarray*}
(\boldsymbol{\mu},\boldsymbol{\nu}) \leq (\boldsymbol{\mu}', \boldsymbol{\nu}') & \text{ iff} & \text{for all } i,j, \mu_i\leq \mu'_i \text{ and } \nu_j \leq \nu'_j.
\end{eqnarray*}

For any closed subset $E\subseteq \counterValuationSet{\Phi}$, we let $\dclosure{E}$ be the set of maximal elements of $E$ w.r.t.~$\leq$.

\begin{definition}[Bellman Operator for \CPCTL$\,$]\label{def:Bellman}

    The {\em extended Bellman operator} $\mathcal{B}_\Phi$ acts on elements of $$\mathcal{D}\left(\counterValuationSet{\Phi}\right)^{\#\mathcal{S}},$$ and is such that, for any $s\in S$, $\mathcal B_\Phi[V](s)=\dclosure{E}$, where $E$ is the set of all couple of tuples $(\boldsymbol\mu,\boldsymbol\nu)=((\mu_i)_{i\in \{1,\ldots,sf(\Phi)\}},(\nu_j)_{j\in\{1,\ldots,pf(\Phi)\}})$ such that there exists an action distribution $\delta$ over the actions $A(s)$ and there exists $(\boldsymbol\mu^{s'},\boldsymbol\nu^{s'})_{s'\in \mathcal{S}}$ with $(\boldsymbol\mu^{s'},\boldsymbol\nu^{s'}) \in V(s')$ for all $s'\in S$ such that all the following are satisfied for all $i,j$.

\noindent\textbf{Valuations for state formulas:}
    \begin{itemize}
    \item If $\Phi_i=a$, then $\mu_i=1$ if $a\in L(s)$ and $0$ otherwise
    \item If $\Phi_i=\neg a$, then $\mu_i=0$ if $a\in L(s)$ and $1$ otherwise
    \item If $\Phi_i=\Phi_{i_1}\land\Phi_{i_2}$, then $\mu_i = \mu_{i_1}\mu_{i_2}$
    \item If $\Phi_i = \mathbb P_{\geq p}\left(\phi_j\right)$, then $\mu_{i}=1$ if $\nu_{j}\geq p$ and $0$ otherwise
    \end{itemize}
    \textbf{Counters for path formulas:}
    \begin{itemize}
        \item If $\phi_j=\Phi_{j_1} \mathbf{W} (\Phi_{j_1} \wedge \Phi_{j_2})$, then \[
        \nu_j =
\begin{cases}
0 \text{ if } \mu_{j_1}=0 \\[1.2ex]
1 \text{ if } \mu_{j_1}=\mu_{j_2} = 1\\[1.2ex]
\sum\limits_{a \in A(s),\, s' \in S'} \delta(a)\, P(s'|s,a)\, \nu^{s'}_{j} \text{ otherwise.}
\end{cases}
        \]
    \end{itemize}

\end{definition}

%Lemma \ref{lemma-alit-W} shows that the synthesis for the for-sure PCTL formula $\phi=\Phi_1\mathbf{W}(\Phi_1\land\Phi_2)$ can be reduced to the synthesis problem for its subformula $\Phi_1\land\Phi_2$. This result enables us, in the following, to show that a simple initialization is sufficient for our VI algorithm.

\begin{definition}[Initial Value Vector]\label{def:initial}
We let $\mathcal I_{\mathcal M}$ be such that, for any $s\in S$,  $\mathcal I_{\mathcal M}(s)$ is the set of tuples $(\boldsymbol\mu,\boldsymbol\nu)\in \counterValuationSet{\Phi}$ satisfying: there exists a policy $\pi$ such that for all $\phi_j = \Phi_{j_1} \mathbf{W} (\Phi_{j_1}\wedge \Phi_{j_2})$
    \[ 
    \nu_{j}=1 \text{ when } \mathcal M_{\pi}^{s}\models \mathbb P_{\geq 1}[\mathbf{G}(\alit{\Phi_j})],\quad \nu_j=0\text{ otherwise.}
    \]
    Additionally, the valuations are maximal for the partial order $\geq$ among the set of valuations such that $\mathcal{I}_{\mathcal{M}}$ satisfies the state compatibility (as in Definition \ref{def:state-compatibility}).

\end{definition}

Our Value Iteration algorithm is given in Algorithm \ref{alg:CPCTL-VI}. 

\begin{algorithm}
\caption{\textsc{CPCTL-VI}}
\label{alg:CPCTL-VI}
\begin{algorithmic}[1]
\State \textbf{Input:} MDP $\mathcal M$ with initial state $s_0$, \CPCTL$\,$formula $\Phi=\land_{i=1}^k\mathbb{P}_{\geq p_i}[\phi_i]$
\State Initialize $V\leftarrow \mathcal I_{\mathcal M}$
\While{for any $\boldsymbol\mu,\boldsymbol\nu\in V(s_{0})$, there exists $i$ such that $\nu_{i}< p_i$}
    \State $V\leftarrow \mathcal B_{ \Phi}[V]$
\EndWhile
\State \textbf{return $V$}
\end{algorithmic}
\end{algorithm}

The algorithm constructs a set of tuples $V(s)$ for each state $s$, starting with the values provided by Definition \ref{def:initial} and iteratively applies the Bellman Operator defined in Definition \ref{def:Bellman}. The algorithm's goal is to maintain the following property: for each state $s$ and each tuple $(\boldsymbol \mu,\boldsymbol \nu)$, the augmented state $(s,\boldsymbol \mu,\boldsymbol \nu)$ is realizable. Moreover, for each such augmented state, the algorithm associates a policy that realizes this augmented state. This construction is described in more detail in the next section.

\subsection{CPCTL-VI: Soundness and Optimality}

This section is devoted to the formal soundness and optimality of the algorithm. 
%We use this as an opportunity to describe how a policy can be constructed once the algorithm answers positively to the Synthesis problem. 
We first introduce the following lemma, showing how policies of the augmented MDP can be "de-augmented" back to the original MDP.

\begin{restatable}[Projection onto original MDP]{lemma}{projection}\label{projection}
     Let $\MDP$ be a MDP, $\Phi$ a \CPCTL$\,$formula, $\AugMDP$ the augmented MDP, and $\aug \pi$ be a valued policy satisfying the state and path compatibilities. Let $\aug s_0 = (s_0,\boldsymbol\mu,\boldsymbol\nu)$ be the state chosen from $s_{-1}$ by $\aug \pi$. With $\boldsymbol\mu=(\mu_1,\dots,\mu_{sf(\Phi)})$, if $\mu_{j_0}=1$ with $j_0$ the index of $\Phi$ as its own subformula, there exists $\pi$ a policy on $\MDP$, such that $
    \mathcal{M}_\pi^{s_0} \models \Phi$.
    Moreover, for any path subformula $\phi_j$, $
    \mathcal{M}_\pi^{s_0} \models \mathbb{P}_{\geq \nu_j} (\phi_j)$.
\end{restatable}

We now introduce the Soundness theorem which states that each value $(\mu,\nu)$ that we add to the set $E_n(s_0)$ corresponds to a realizable augmented state $\hat s$. In fact, Theorem 3 shows how the algorithm can be extended to output for each such added value a policy that realizes the corresponding augmented state. 
%We use this theorem as an opportunity to present the construction of the policies that the algorithm provides.

\begin{restatable}[Soundness]{theorem}{thsoundness}
    Let $\MDP$ be a MDP and $\Phi$ be a formula in \CPCTL$\,$.
    Let $s\in \mathcal{S}$, and $E_n(s)$ be the realisability set for $s$ obtained after $n$ steps using \textsc{CPCTL-VI}, then for every $(\boldsymbol \mu,\boldsymbol \nu) \in E_n(s)$, there exists a policy $\pi$ such that
    \[
    \left\{
    \begin{aligned}
    &\forall j\leq pf(\Phi),~ \mathbb{P}(\phi_j|\mathcal{M}_\pi,s)\geq \nu_j,\\
    &\forall j\leq sf(\Phi),~\mu_j = 1 \Rightarrow \MDP_\pi^s \models \Phi_j.
    \end{aligned}
    \right.
    \]
    Moreover, the policy can be taken as the memoryful policy $\pi$ with initial memory state $(\boldsymbol \mu,\boldsymbol \nu)$ and such that, if $(\boldsymbol \mu',\boldsymbol \nu')$ is its current memory state, then
    %Moreover, for any $(\boldsymbol \mu^0,\boldsymbol \nu^0)\in E_n(s_0)$, with $\pi_m^0$ the memoryful policy defined as: 
    
    %\noindent\textbf{Initial memory state:} Initially, $m_0=(\boldsymbol \mu^0,\boldsymbol \nu^0)$.

    %\noindent\textbf{Definition of $\pi_{(\boldsymbol \mu,\boldsymbol \nu)}(s)$:}
    \begin{itemize}[leftmargin=4mm] 
        \item If $\forall j$ such that $\mu'_j \neq 0$, $\Phi_j=b_j$ or $\neg b_j$ for some $b_j\in AP$, then the policy $\pi$ takes an arbitrary action, and its next memory state is a pair of two null tuples.
        \item If for all $j$, either $\nu'_j=0$ or $\nu'_j=1$, then the policy $\pi$ follows indifinitely a memoryless policy $\pi_{\boldsymbol \nu'}$ such that $\mathcal{M}_{\pi_{\nu'}}^s \models \mathbb{P}_{\geq 1} (\mathbf{G} \alit{\Phi_j})$ for any $j$ such that $\nu'_j=1$ and its memory state does not change.
        %\item If there exists $\Phi_j=\mathbb{P}_{\geq p_j}(\phi_j)$ such that $\mu_j=\nu_j=1$ and for all $k$ such that $\mu_k\neq 0$, $\Phi_k$ is a subformula of $\Phi_j$ and $\exists \pi_j$, $\mathcal{M}_{\pi_j}^s \models \mathbb{P}_{\geq 1} (\mathbf{G} \alit{\Phi_j})$, then $\pi^0_m$ follows the policy $\pi_j$ indefinitely.
        \item Otherwise, there exists $\delta$, a distribution of actions of $\mathcal{A}(s)$ and tuples $(\boldsymbol \mu^i,\boldsymbol \nu^i)$ for each successor $s_i$ such that $(\delta,\boldsymbol\mu',\boldsymbol\nu',(\boldsymbol \mu^i)_i,(\boldsymbol \nu^i)_i)$ satisfies Bellman's inequality. The policy $\pi$ follows the distribution of action $\delta$, and when reaching the state $s_i$, its memory state is modified to the new value $m=(\boldsymbol \mu^i,\boldsymbol \nu^i)$.
    \end{itemize}  
    %The policy $\pi$ is well defined and satisfies for all $j\leq pf(\Phi)$, $\mathcal{M}_{\pi^0_{m_0}}\models \mathbb{P}_{\geq j}(\phi_j)$, and for all $j\leq sf(\Phi)$, $\mu_j=1\Rightarrow \mathcal{M}_{\pi_{m_0}} \models \Phi_j$.

\end{restatable}

%\end{definition}
Finally, we show in the next theorem that under the generalized Slater's Assumption, the algorithm is optimal, meaning that it finds a policy satisfying the Synthesis Problem when such a policy exists.

\begin{restatable}[Optimality of $\mathrm{CPCTL-VI}$]{theorem}{optimality}\label{thm:VI-optimality}
    We consider $\MDP$ a MDP and $\Phi$ a \CPCTL$\,$formula. Denote by $V_n$ the value frontier obtained after $n$ steps of $\mathrm{CPCTL-VI}$, i.e.
    $V_n = \left(\mathcal B_{ \Phi}\right)^n[\mathcal I_{\mathcal M}]$.
    If $\MDP,\Phi$ satisfy the generalized Slater's Assumption, there exists $n_0\in \mathbb{N}$,$
    \exists (\boldsymbol \mu,\boldsymbol \nu)\in V_{n_0}(s_0),~\nu_i \geq p_i$ for any $i\in\{1,\ldots,k\}$.
\end{restatable}

\section{Numerical experiments} \label{sec:experiments}

\noindent\textit{Presentation of the problem:} 
A robot moves over a surface where the ground becomes increasingly slippery toward the unsafe borders. 
On the right side, the slipperiness remains constant and moderate, while near the left edge the surface is highly polished, making the robot prone to sliding unpredictably in any direction. 
The robot must reach the goal region at the top while avoiding the unsafe edges and can not cross the central wall that divides the terrain. 

We consider the formula $\mathbb{P}_{\geq p_1} \left( \mathbf{G} (\mathbb{P}_{\geq p_2} (\neg d \mathbf{W} G)) \right)$,
where $p_2=0.6$ and our goal is to maximize $\phi=\mathbb{P} \left( \mathbf{G} (\mathbb{P}_{\geq p_2} (\neg d \mathbf{W} G)) |\pi, S \right)$. Compared to  flat formulas, this reduces the bias from starting at a given state, as the robot is asked to minimize the risk of reaching high risk states later on.
We implemented CPCTL-VI for $\phi$ in Python, and the experiments were carried on Windows 10 with an Intel(R) Core(TM) i7-12650H CPU.
\subsection{First Model: a Gridworld with a central Wall}

We introduce the first $7\times 7$ Gridworld model with a central wall in figure \ref{fig:model1}. The Pareto Curve outputted by the Algorithm for this model is presented in Figure \ref{fig:pareto1}. For any $a$ on the curve, there exists a policy such that $\MDP_\pi \models \mathbb{P}_{\geq a} (\mathbf{G} \mathbb{P}_{\geq p_2} (\neg d \mathbf{W} (\neg d \wedge G)))$, while $b$ is only used internally by the algorithm.
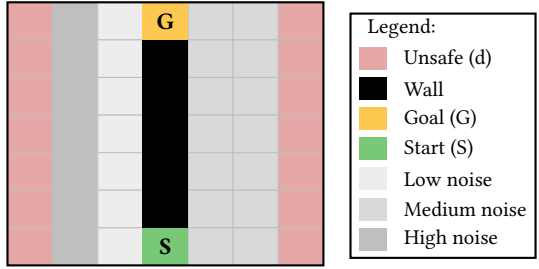
\begin{figure}[h]
\centering
\begin{tikzpicture}[x=0.6cm,y=0.5cm]
  % Parameters
  \def\W{7}        % width  (columns 0..6)
  \def\H{7}        % height (rows    0..6)
  \def\midc{3}     % middle column with the wall
  \def\startR{6}   % start row
  \def\startC{3}   % start col
  \def\goalR{0}    % goal row
  \def\goalC{3}    % goal col

  % Colors
  \definecolor{unsafe}{RGB}{200,60,60}
  \definecolor{startclr}{RGB}{120,200,120}
  \definecolor{goalclr}{RGB}{252,200,80}
  % Slip gradient colors
  \colorlet{slipLow}{black!7}      % far from unsafe (very low slip)
  \colorlet{slipMed}{black!15}     % intermediate slip
  \colorlet{slipHigh}{black!25}    % close to unsafe (high slip)
  \colorlet{slipConst}{black!15}   % right side constant slip (you can tune this)

  % Draw base grid
  \foreach \r in {0,...,\numexpr\H-1\relax}{
    \foreach \c in {0,...,\numexpr\W-1\relax}{
      % Unsafe columns
      \ifnum\c=0
        \fill[unsafe!45] (\c,-\r) rectangle ++(1,-1);
      \else\ifnum\c=\numexpr\W-1\relax
        \fill[unsafe!45] (\c,-\r) rectangle ++(1,-1);
      \else
        % Slip intensity gradient
        \ifnum\c=1
          \fill[slipHigh] (\c,-\r) rectangle ++(1,-1);
        \else\ifnum\c=2
          \fill[slipLow] (\c,-\r) rectangle ++(1,-1);
        \else
          \fill[slipConst] (\c,-\r) rectangle ++(1,-1);
        \fi\fi
      \fi\fi
      % Cell outline
      \draw[gray!50] (\c,-\r) rectangle ++(1,-1);
    }
  }

  % Full vertical wall in the middle column (except start/goal)
  \foreach \r in {1,...,\numexpr\H-2\relax}{
    \ifnum\r=\goalR\relax\else
    \ifnum\r=\startR\relax\else
      \fill[black] (\midc,-\r) rectangle ++(1,-1);
    \fi\fi
  }

  % Goal cell
  \fill[goalclr] (\goalC,-\goalR) rectangle ++(1,-1);
  \node[font=\bfseries] at (\goalC+0.5,-\goalR-0.5) {G};

  % Start cell
  \fill[startclr] (\startC,-\startR) rectangle ++(1,-1);
  \node[font=\bfseries] at (\startC+0.5,-\startR-0.5) {S};

  % Outer border
  \draw[thick] (0,0) rectangle (\W,-\H);

  % Legend
  \begin{scope}[shift={(\W+0.6, -0.3)}]
    \draw[thick] (0,0) rectangle (4.1,-6.5);
    \node[anchor=west] at (0.2,-0.4) {\small Legend:};
    \fill[unsafe!45] (0.2,-0.9) rectangle ++(0.6,-0.6);
    \node[anchor=west] at (1.0,-1.2) {\small Unsafe (d)};
    \fill[black] (0.2,-1.7) rectangle ++(0.6,-0.6);
    \node[anchor=west] at (1.0,-2.0) {\small Wall};
    \fill[goalclr] (0.2,-2.5) rectangle ++(0.6,-0.6);
    \node[anchor=west] at (1.0,-2.8) {\small Goal (G)};
    \fill[startclr] (0.2,-3.3) rectangle ++(0.6,-0.6);
    \node[anchor=west] at (1.0,-3.6) {\small Start (S)};
    % Slip color legend (light = low slip, dark = high slip)
    \fill[slipLow] (0.2,-4.1) rectangle ++(0.6,-0.6);
    \node[anchor=west] at (1.0,-4.4) {\small Low noise};
    \fill[slipMed] (0.2,-4.9) rectangle ++(0.6,-0.6);
    \node[anchor=west] at (1.0,-5.2) {\small Medium noise};
    \fill[slipHigh] (0.2,-5.7) rectangle ++(0.6,-0.6);
    \node[anchor=west] at (1.0,-6.0) {\small High noise};
  \end{scope}
\end{tikzpicture}
\caption{First Grid World: $7\times7$ grid with a central wall and slip gradient represented as gray shading (darker means higher slip probability). The slip is a uniform stochastic noise in all the available directions. The goal (\textbf{G}) is safe and the start (\textbf{S}) is the initial state.}
\label{fig:model1}
\end{figure}

In this model, the safest strategy, i.e. the strategy $\pi$ that maximizes $\mathbb{P}(\phi|\pi,\text{Start})$ is to enter the right corridor and to try and stay in the fifth column. This strategy however corresponds to a higher risk to reach unsafe states in average, but decreases the probability of reaching high risk states compared to a strategy that uses the left corridor.

\begin{figure}[h]
\centering

% --- Left figure ---
\begin{subfigure}[t]{0.48\linewidth}
\centering
\includegraphics[width=\linewidth]{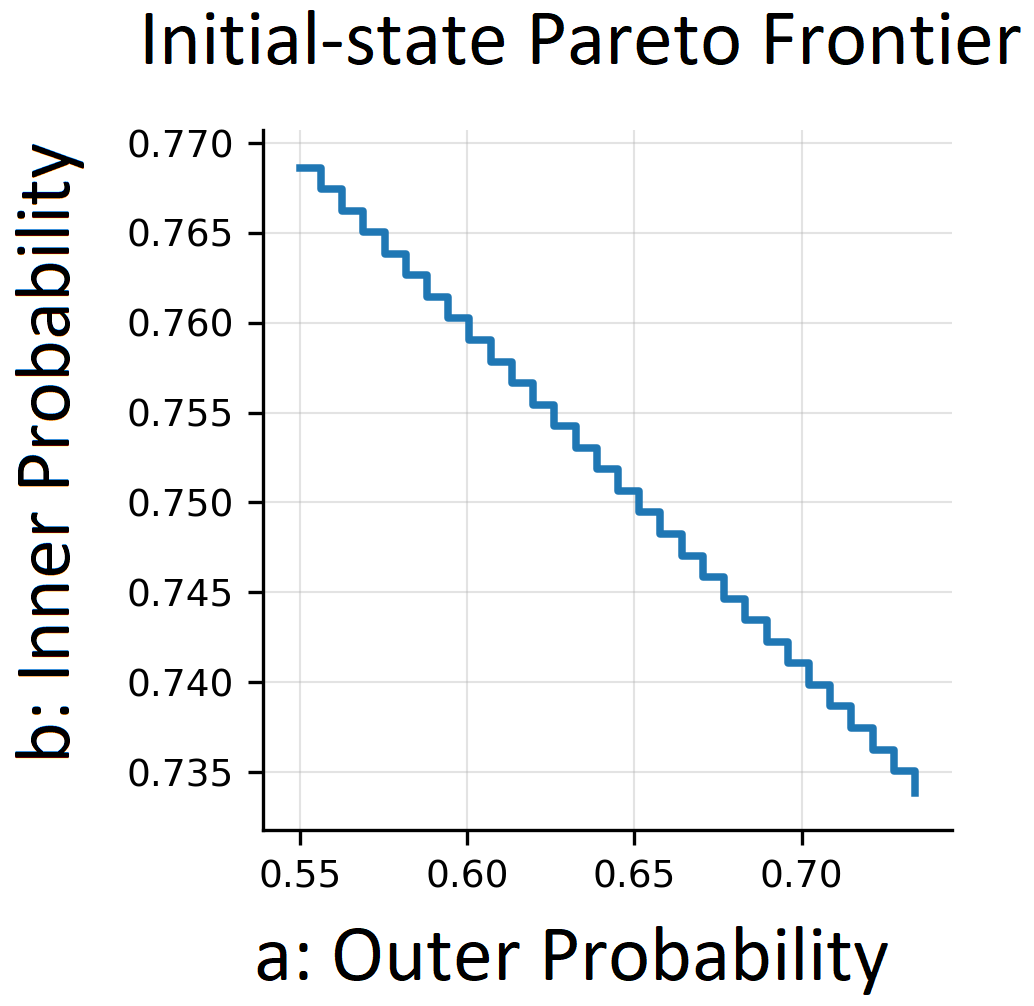}
\caption{First Model.}
\label{fig:pareto1}
\end{subfigure}
\hfill
% --- Right figure ---
\begin{subfigure}[t]{0.48\linewidth}
\centering
\includegraphics[width=\linewidth]{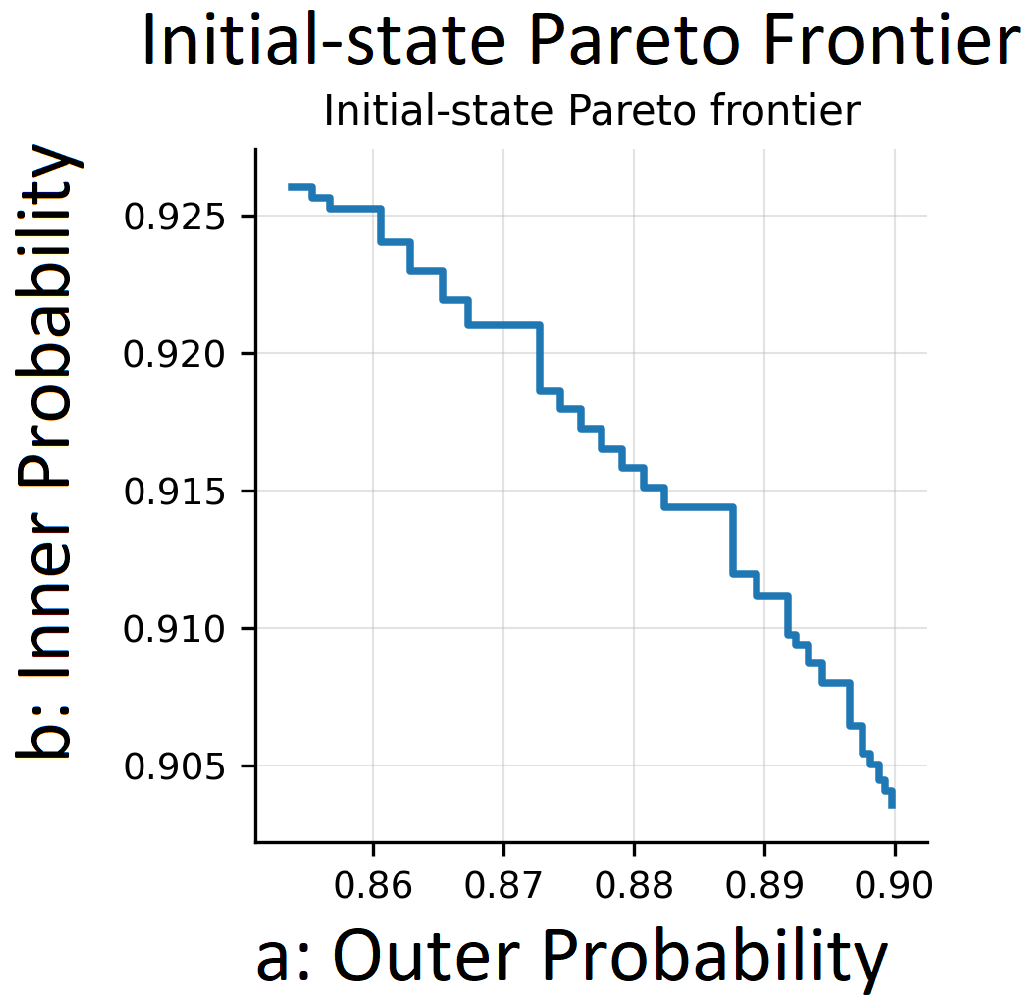}
\caption{Second Model.}
\label{fig:pareto2}
\end{subfigure}

\vspace{-0.2cm}
\caption{Results of the Algorithm \textsc{CPCTL-VI} for the two models. The curves are step functions because the convexity of the Pareto Frontier is not guaranteed. The points are not joined to ensure soundness and guarantee safety. The curve obtained in the second model is an example of non-convex Pareto Curve.}
\label{fig:frontier-compare}
\end{figure}

\subsection{Second Model: Gridworld with a hole}

We present the second model in figure \ref{fig:model2}. A hole is now present in the wall and the agent can move through it.
%It possesses a central wall containing a hole behaving like any right-side state. 
The optimal strategy is now to follow the right side of the wall (sixth column) or go in the direction of this column when the agent is moved away from it by the stochastic noise. When too far on the left and blocked by the wall, the agent follows the left side of the wall (fourth column). The Pareto Curves obtained by the algorithm are presented in Figure \ref{fig:pareto2}.

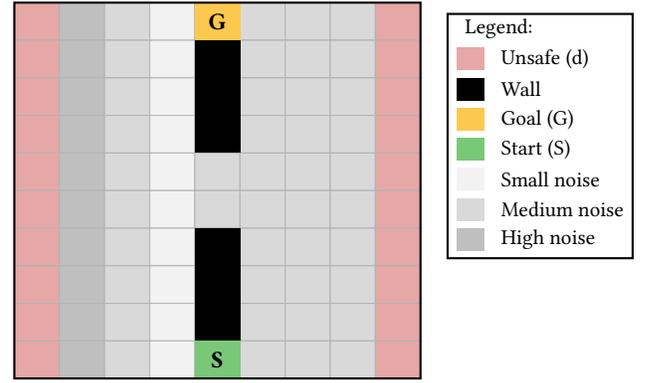
\begin{figure}[H]
\centering
\begin{tikzpicture}[x=0.6cm,y=0.5cm]
  % Parameters (match Python demo)
  \def\W{9}        % width  (columns 0..8)
  \def\H{10}       % height (rows    0..9)
  \def\midc{4}     % middle column with the wall
  \def\startR{9}   % start row (bottom)
  \def\startC{4}   % start col (middle)
  \def\goalR{0}    % goal row (top)
  \def\goalC{4}    % goal col (middle)

  % Colors
  \definecolor{unsafe}{RGB}{200,60,60}
  \definecolor{startclr}{RGB}{120,200,120}
  \definecolor{goalclr}{RGB}{252,200,80}
  % Slip gradient: light gray = low slip, dark gray = high slip
  \colorlet{slipLow}{black!5}
  \colorlet{slipMed}{black!15}
  \colorlet{slipHigh}{black!25}
  \colorlet{slipConst}{black!15}

  % Draw base grid
  \foreach \r in {0,...,\numexpr\H-1\relax}{
    \foreach \c in {0,...,\numexpr\W-1\relax}{
      % Unsafe columns
      \ifnum\c=0
        \fill[unsafe!45] (\c,-\r) rectangle ++(1,-1);
      \else\ifnum\c=\numexpr\W-1\relax
        \fill[unsafe!45] (\c,-\r) rectangle ++(1,-1);
      \else
        % Slip intensity by column (inverse gradient)
        \ifnum\c=1
          \fill[slipHigh] (\c,-\r) rectangle ++(1,-1); % darkest = closest to unsafe
        \else\ifnum\c=2
          \fill[slipMed] (\c,-\r) rectangle ++(1,-1);
        \else\ifnum\c=3
          \fill[slipLow] (\c,-\r) rectangle ++(1,-1); % lightest = far from wall
        \else
          \fill[slipConst] (\c,-\r) rectangle ++(1,-1); % uniform right side
        \fi\fi\fi
      \fi\fi
      % Empty cell outline
      \draw[gray!60] (\c,-\r) rectangle ++(1,-1);
    }
  }

  % Two-piece vertical wall in middle column
  \foreach \r in {6,7,8}{
    \fill[black] (\midc,-\r) rectangle ++(1,-1);
  }
  \foreach \r in {1,2,3}{
    \fill[black] (\midc,-\r) rectangle ++(1,-1);
  }

  % Goal cell
  \fill[goalclr] (\goalC,-\goalR) rectangle ++(1,-1);
  \node[font=\bfseries] at (\goalC+0.5,-\goalR-0.5) {G};

  % Start cell
  \fill[startclr] (\startC,-\startR) rectangle ++(1,-1);
  \node[font=\bfseries] at (\startC+0.5,-\startR-0.5) {S};

  % Outer border
  \draw[thick] (0,0) rectangle (\W,-\H);

  % Legend
  \begin{scope}[shift={(\W+0.6, -0.3)}]
    \draw[thick] (0,0) rectangle (4.1,-6.5);
    \node[anchor=west] at (0.2,-0.4) {\small Legend:};
    \fill[unsafe!45] (0.2,-0.9) rectangle ++(0.6,-0.6);
    \node[anchor=west] at (1.0,-1.2) {\small Unsafe (d)};
    \fill[black] (0.2,-1.7) rectangle ++(0.6,-0.6);
    \node[anchor=west] at (1.0,-2.0) {\small Wall};
    \fill[goalclr] (0.2,-2.5) rectangle ++(0.6,-0.6);
    \node[anchor=west] at (1.0,-2.8) {\small Goal (G)};
    \fill[startclr] (0.2,-3.3) rectangle ++(0.6,-0.6);
    \node[anchor=west] at (1.0,-3.6) {\small Start (S)};
    % slip color legend (light = low slip, dark = high slip)
    \fill[slipLow] (0.2,-4.1) rectangle ++(0.6,-0.6);
    \node[anchor=west] at (1.0,-4.4) {\small Small noise};
    \fill[slipMed] (0.2,-4.9) rectangle ++(0.6,-0.6);
    \node[anchor=west] at (1.0,-5.2) {\small Medium noise};
    \fill[slipHigh] (0.2,-5.7) rectangle ++(0.6,-0.6);
    \node[anchor=west] at (1.0,-6) {\small High noise};
  \end{scope}
\end{tikzpicture}
\caption{Second Grid World: $9\times 10$. The agent starts at the initial state $S$ and must reach the safe goal $\mathbf{G}$ while avoiding the unsafe borders. The central wall is impassable but possesses a hole in the middle that behaves like any right-side state.  Darker gray indicates higher slip probability. The slip is a stochastic noise uniform in all the available directions.}
\label{fig:model2}
\end{figure}

\section{Conclusions} \label{sec:conclusions}

In this paper we made the following key contributions to the synthesis problem for probabilistic temporal specifications. 
%
%\begin{enumerate}%
%\item \textbf{An analysis of temporal and probabilistic specifications:} 
First, we introduced Continuing \PCTL\, (\CPCTL), an expressive fragment of \PCTL$\,$, which allows the nesting of probabilistic operators. thus generalizing multi-objective avoidance specifications.
%We show that there currently exists no decidability result for such specifications. Additionally, we establish structural results for \CPCTL$\,$ such as weak reduction to literal projections.
%  \item 
%\textbf{A key theoretical result:} 
Then, we provided a novel \emph{augmented MDP} construction for Safe PCTL that encodes global satisfaction as local linear inequalities. In particular, we proved the coherence result Theorem \ref{th:coherence}.
%show that their satisfaction guarantees the satisfaction of the corresponding formula. 
%  \item 
Further, we presented \textsc{CPCTL-VI}, a value-iteration algorithm for \CPCTL\ synthesis, and proved its optimality  under a generalized version of Slater's assumption.
  %\item A safety-preserving reduction from Safe PCTL-constrained optimization to unconstrained optimization in a safe augmented MDP, and a method \textsc{CPCTL-OP} for constrained optimization with\CPCTL$\,$constraints, compatible with standard policy optimization algorithms and guaranteeing constraint satisfaction during optimization.
%\end{enumerate}
Finally, we evaluated \textsc{CPCTL-VI} experimentally, thus showing the feasibility of our method in practice.

\paragraph{Acknowledgments.} The research described in this paper was partially supported by the EPSRC (grant number EP/X015823/1).

%\bibliography{sample}

%\section*{Comment from reviewers}

%\begin{enumerate}
%    \item \st{Parts of Example 1 is confusing. Do $G$ and $\textbf{G}$ refer to the same thing? What does $S_{\phi_a}$ refer to?}
%    \item 
%        \begin{itemize}
%            \item \st{
%            The writing of the paper can be improved. The main approach section (Section 4) reads as a collection of definitions and theorems, it would have been better to explain the relevance of these definitions, and their connections to one another.}
            
%            \item \textcolor{red}{While the theoretical results are elegant, there’s little discussion about the scalability of the approach or its potential bottlenecks in practice. Value iteration over augmented MDPs can easily blow up, and it’s unclear how the method would perform on anything but toy examples.}
%    \end{itemize}
%    \item \textcolor{red}{
%Is there a specific reason why Algorithm 1 and 2 not applied to an example problem?}
%    \item \textcolor{red}{
%The biggest weakness is that there are no computational experiments or empirical evaluation. This is a pretty serious omission since the method is algorithmic and is intended to have practical implications. (Reproducibility checklist is checked as "yes" for experiments/code, but no experiments are reported. Please clarify.)}
%\end{enumerate}

% \bibliographystyle{ACM-Reference-Format}
% \bibliography{AAMAS2026:CPCTL_Synthesis/sample}

\bibliographystyle{ACM-Reference-Format}
\bibliography{sample}

\newpage

\clearpage
\newpage
\onecolumn

\appendix 

\section{Auxiliary Definitions}

\begin{definition} Let $\Phi \in $\PCTL. We write $path(\Phi)=\top$ if and only if there exists $\phi$, $\Phi=\mathbb{P}_{\geq p}(\phi)$. We define:

\noindent \textbf{Nesting depth:} The Nesting Depth $\mathcal{D}_\mathcal{N}$ is defined by induction as
    \[
    \mathcal{D}_\mathcal{N} :\left\{
    \begin{aligned}
        &\mathbb{P}_{\geq p}(\mathbf X \Phi_1) \mapsto 1 +\mathcal{D}_\mathcal{N} (\Phi_1),\\
        &\mathbb{P}_{\geq p}(\Phi_1 \mathbf{W} (\Phi_1 \wedge \Phi_2)) \mapsto 1 +\max ( \mathcal{D}_\mathcal{N} (\Phi_1), \mathcal{D}_\mathcal{N}(\Phi_2)),\\
        &\mathbb{P}_{\geq p}(\Phi_1 \mathbf{W} \Phi_2) \mapsto 1 +\max ( \mathcal{D}_\mathcal{N} (\Phi_1), \mathcal{D}_\mathcal{N}(\Phi_2)),\\
        &\Phi_1 \wedge \Phi_2 \mapsto \max ( \mathcal{D}_\mathcal{N} (\Phi_1), \mathcal{D}_\mathcal{N}(\Phi_2)),\\
        &\Phi_1 \vee \Phi_2 \mapsto \max ( \mathcal{D}_\mathcal{N} (\Phi_1), \mathcal{D}_\mathcal{N}(\Phi_2)),\\
        &a,\neg a,~a\in AP \mapsto 0.
    \end{aligned}
    \right.
    \]
\noindent \textbf{Total depth:} The Total Depth $\mathcal{D}_t$ is defined by induction as
    \[
    \mathcal{D}_t :\left\{
    \begin{aligned}
        &\mathbb{P}_{\geq p}(\mathbf X \Phi_1) \mapsto 1 +\mathcal{D}_\mathcal{t} (\Phi_1),\\
        &\mathbb{P}_{\geq p}(\Phi_1 \mathbf{W} (\Phi_1 \wedge \Phi_2)) \mapsto 1 +\max ( \mathcal{D}_t (\Phi_1), \mathcal{D}_t(\Phi_2)),\\
        &\mathbb{P}_{\geq p}(\Phi_1 \mathbf{W} \Phi_2) \mapsto 1 +\max ( \mathcal{D}_\mathcal{t} (\Phi_1), \mathcal{D}_t(\Phi_2)),\\
        &\Phi_1 \wedge \Phi_2,~\mathcal{D}_t(\Phi_1)+ \mathcal{D}_t(\Phi_2)\neq 0 \mapsto 1 +\max ( \mathcal{D}_t (\Phi_1), \mathcal{D}_t(\Phi_2)),\\
        &\Phi_1 \wedge \Phi_2,~\mathcal{D}_t(\Phi_1) = \mathcal{D}_t(\Phi_2) = 0\mapsto 0,\\
        &\Phi_1 \vee \Phi_2,~\mathcal{D}_t(\Phi_1)+ \mathcal{D}_t(\Phi_2)\neq 0 \mapsto 1 +\max ( \mathcal{D}_t (\Phi_1), \mathcal{D}_t(\Phi_2)),\\
        &\Phi_1 \vee \Phi_2,~\mathcal{D}_t(\Phi_1) = \mathcal{D}_t(\Phi_2) = 0\mapsto 0,\\
        &a,\neg a,~a\in AP \mapsto 0.
    \end{aligned}
    \right.
    \]

\end{definition}

\begin{definition}
For a formula $\Phi$, a labelled state $s$, and an associated set of counters $\nu$, we define the canonical associated set of valuations $\mu=\Xi(\Phi,s)(\nu)$ to be $\mu=(\mu_1,\dots\mu_{SF(\Phi)})$ where each $\mu_j$ is defined by induction on the state formula $\Phi_j$ as 
\[
\begin{aligned}
&\Phi_j = \mathbb{P}_{\geq p_j}(\phi_j) \mapsto 1 \text{ if } \nu_j \geq p_j,~0\text{ otherwise},\\ 
&\Phi_j = \Phi_{j,1} \wedge \Phi_{j,2} \mapsto \mu_{j,1}\cdot \mu_{j,2},\\
&\Phi_j = b_j \in AP \mapsto 1 \text{ if }b_j\in L(s),~0\text{ otherwise},\\
&\Phi_j = \neg b_j,~b_j \in AP \mapsto 0 \text{ if }b_j\in L(s),~1\text{ otherwise}.\\ 
\end{aligned}
\]
When clear from the context, we simply write "let $(\mu,\nu)$ be the canonical profile associated to $\nu$ at $s$".
    
\end{definition}

\begin{lemma}
    Let $\MDP$ be a MDP and $\Phi\in$\CPCTL.
    Let $V_n(s)$ be the value-frontier sets obtained after $n$ iterations of the Algorithm \textsc{CPCTL-VI}. For any $s\in \mathcal{S}$, for any profile $(\boldsymbol{\mu,\nu})\in V_n(s)$, we have
    \[
    \boldsymbol{\mu} = \Xi(\Phi,s)(\boldsymbol{\nu}).
    \]
\end{lemma}

\begin{proof}
    The proof follows from an induction of the structure of $\Phi$, after noticing that the Bellman's operator only adds profile whose valuations satisfy the canonical valuation property.
\end{proof}

\begin{definition}
    Let $\MDP$ be an MDP and $\Phi$ be a \CPCTL\ formula. Let $s$ be a state of $\MDP$ and $\pi$ a policy on $\MDP^s$. We denote $\mathtt{P}(\pi|\mathcal{M}^s,\Phi)$, the tuple
    \[
    \mathtt{P}(\pi|\mathcal{M}^s,\Phi)=\boldsymbol{p,q}=(p_1,\dots,p_{sf(\Phi)},q_1,\dots,q_{pf(\Phi)})
    \]
    defined as the following. For any $j\leq pf(\Phi)$,
    \[
    q_j = \mathbb{P}(\phi_j | \mathcal{M}_\pi^s),
    \]
    and 
    \[
    \boldsymbol{p}= \Xi(\Phi,s)(\boldsymbol{q}).
    \]
    Intuitively, $\boldsymbol{p,q}$ correspond to the profile that is realized by $\pi$ in $\MDP$ when starting at $s$.
    
\end{definition}

\begin{definition}\textbf{(Probability Thresholds, Strict satisfaction)}
    We define the \emph{Probability Thresholds} of a \CPCTL\ formula $\Phi$ as the tuple denoted $prob(\Phi)$ such that for each state formula $\Phi_j$ of the form $\Phi_j = \mathbb{P}_{\geq p_j}(\phi_j)$, $prob(\Phi)=p_j$.

    Let $\MDP$ be a MDP, $\Phi$ be a \CPCTL\ formula and $\pi$ be a policy. We say that $\pi$ \emph{strictly satisfies} $\Phi$, denoted $\pi\models_S \Phi$, if with $\boldsymbol p=prob(\Phi)$, there exists $\boldsymbol q > \boldsymbol p$ satisfying
    \[
    \MDP_
    {\pi}\models \Psi,\quad \Psi = \mathtt{T}(\Phi,\boldsymbol q).
    \]
\end{definition}

\section{The coherence theorem} \label{appendix:coherencethm}

\thcoherence*

\begin{proof}
    We consider a MDP $\MDP$, a formula $\Phi$ and the corresponding augmented MDP $\AugMDP$. Let $\aug \pi$ be a valued policy satisfying both state compatibility and path compatibility. We will show by induction on the structure of $\Phi$ that $\aug \pi$ is coherent. 

    \textbf{Initialisation:}

    We consider $\Phi_j$ a state formula that is of the form $\Phi_j=b$ for $b \in AP$. There is no counter associated to $\Phi_j$ because there is no path formula associated. We consider any augmented state $\aug s = (s,\boldsymbol \mu,\boldsymbol \nu) \in S_{\pi,\neq -1}$, the set of augmented states reachable following $\aug \pi$ after one step. Then, 
    \begin{itemize}
        \item If $\mu_j=0$, then coherence is satisfied automatically.
        \item If $\mu_j=1$, then by state compatibility, $s\models b$, so we indeed have $\AugMDP_{\aug \pi} (\aug s) \models \Phi_j$.
    \end{itemize}

    \textbf{Induction: Logical operators}
    We consider $\Phi_j= \Phi_{j_1} \wedge \Phi_{j_2}$ and a state $\aug s\in \mathcal{S}_r$. By state compatibility, we deduce that $\mu_{j_1}=1$ and $\mu_{j_2}=1$. Using the induction hypothesis, we deduce that $\aug s\models \Phi_{j_1}$ and $\aug s\models \Phi_{j_2}$. The proof is similar for $\Phi_j=\Phi_{j_1}\wedge \Phi_{j_2}$.

    \textbf{Induction: Temporal operators}
    We start with $\phi_j = \mathbf X\Phi_{j_1}$. We have by path compatibility
    \[
        \nu_j \leq \left( \sum_{i=1}^m P(s_i|s,a) \delta_{\aug \pi,\aug s}(a) \mu_{j_1}^i \right). 
    \]
    Using the induction hypothesis for $\Phi_{j_1}$, we have that 
    \[
    \mu_{j_1}^i \leq \mathbb{P}(\Phi_{j_1}|\MDP_{\aug\pi}^{\aug s_i}),
    \]
    which means 
    \[
    \mathbb{P}(X\Phi_{j_1}|\aug \pi,\aug s) = \sum_{i=1}^m P(s_i|s,a) \delta_{\aug \pi,\aug s}(a) \mathbb{P}(\Phi_{j_1}|\MDP_{\aug\pi}^{\aug s_i}) \geq \nu_j.
    \]
    
    We now consider $\phi_j = \Phi_{j_1} \mathbf{W} \Phi_{j_2}$. By induction, we have that for any $\aug s \in \aug{S}_*$,
    \[
    (\mu_{j_1} =1) \Rightarrow (\AugMDP_{\aug \pi}^{\aug s} \models \Phi_{j_1}),\quad(\mu_{j_2} =1) \Rightarrow (\AugMDP_{\aug \pi}^{\aug s} \models \Phi_{j_2}).
    \]
    Let $\phi_{j_1}$ and $\phi_{j_2}$ the path formulas such that 
    \[
    \Phi_{j_1} = \mathbb{P}_{\geq p_1}[ \phi_{j_1}],~\Phi_{j_1} = \mathbb{P}_{\geq p_1} [\phi_{j_1}].
    \]
    By induction, we also have 
    \[
    \AugMDP_{\aug \pi}^{\aug s }\models \mathbb{P}_{\geq \eta_{j_1}}(\phi_{j_1}),~\AugMDP_{\aug \pi}^{\aug s} \models \mathbb{P}_{\geq \eta_{j_2}}(\phi_{j_2}).
    \]

    Now, with $\Theta(\aug s)$ the associated counters and valuations for the successors $\{ s_1,\dots,s_m \}$ of $s$ and $\delta$ the distribution on $\Act$, we also have by induction that for each $s_i$,
    \[
    \AugMDP_{\aug \pi},(s_i,\boldsymbol \mu^i,\boldsymbol \nu^i) \models \mathbb{P}_{\geq \nu_{j_1}^i} (\phi_{j_1}),~\AugMDP_{\aug \pi},(s_i,\boldsymbol \mu^i,\boldsymbol \nu^i) \models \mathbb{P}_{\nu_{j_2}^i} (\phi_{j_2}),
    \]
    in other words, % $\mathbb{P}_{\AugMDP,\aug \pi}(\phi_{j_1}|(s_i,\boldsymbol\mu^i,\boldsymbol \nu^i)) \geq \nu_{j_1}^i$ and 
    \[
    \mathbb{P}(\phi_{j_1}|(s_i,\boldsymbol \mu^i,\boldsymbol \nu^i),\AugMDP_{\aug \pi}) \geq \nu_{j_1}^i,~\mathbb{P}(\phi_{j_2}|(s_i,\boldsymbol \mu^i,\boldsymbol \nu^i),\AugMDP_{\aug \pi}) \geq \nu_{j_2}^i.
    \]
    Now, we again consider the formula $\phi_j=\Phi_{j_1} \mathbf W (\Phi_{j_1} \wedge \Phi_{j_2})$, and distinguish three cases.
    \begin{itemize}
        \item[(i)] If $\Phi_{j_1}$ is not satisfied at $\aug s$, then $\mathbb{P}(\Phi,\aug s) =0$. But if $\Phi_{j_1}$ is not satisfied at $\aug s$, then necessarily, by induction hypothesis, $\mu_{j_1}=0$. The path compatibility condition
        \[
        \nu_j \leq \mu_{j_1} \left( \sum_{i=1}^m P(s_i|s,a) \delta_{\aug \pi,\aug s}(a) \max(\nu_j^i,\mu_{j_1}^i\mu_{j_2}^i) \right)
        \]
        implies that $\nu_j=0$. The coherence condition 
        \[
        \mathbb{P} (\phi_j|\aug s,\AugMDP_{\aug \pi}) \geq \nu_j
        \]
        is thus satisfied. The state compatibility also provides $\mu_j=0$, so the coherence condition is satisfied.
        \item[(ii)] If $\Phi_{j_2}$ is satisfied, then $\mu_{j_2}$ can be either $0$ or $1$. In that case, the coherence condition is automatically satisfied since 
        \[
        \mathbb{P}(\phi_j|\aug s,\aug M_{\aug \pi}) =1 \geq \nu_j.
        \]
        \item[(iii)] If $\Phi_{j_1}$ is satisfied and $\Phi_{j_2}$ is not satisfied, then we have by state compatibility that $\mu_{j_2}=0$. At the state $\aug s$, the policy $\aug \pi$ chooses the distribution $\delta$ and the counters and valuations $(\boldsymbol \mu^i, \boldsymbol \nu^i)$, $\boldsymbol \mu^i = (\mu_1^i,\dots ,\mu_m^i)$, $\boldsymbol \nu^i = (\nu_1^i,\dots,\nu_m^i)$. By path compatibility, 
        \[
        \nu_j \leq \mu_{j_1} \left( \sum_{i=1}^m P(s_i|s,a) \delta_{\aug \pi,\aug s}(a) \max(\nu_j^i,\mu_{j_2}) \right).
        \]
        In this case, we have $\mu_{j_2} = 0$ so
        \[
        \nu_j \leq \mu_{j_1} \left( \sum_{i=1}^m P(s_i|s,a) \delta_{\aug \pi,\aug s}(a) \nu_j^i \right).
        \]
        If $\mu_{j_1} = 0$, then $\mu_j=0$ so the coherence condition is satisfied. If $\mu_{j_1}=1$, we get 
        \[
        \nu_j \leq \left( \sum_{i=1}^m P(s_i|s,a) \delta_{\aug \pi,\aug s}(a) \nu_j^i \right).
        \]
        For a path $\xi=\xi_0\xi_1\dots$, we have that $\xi\models \phi_j$ if and only if 
        \[
        (\xi\models \Phi_{j_1}\wedge \Phi_{j_2}) \text{ or } \left(\xi \models \Phi_{j_1} \wedge (\xi_{[\geq 1]}) \models \phi_j \right).
        \]
        Here, we assumed that $\Phi_{j_2}$ is not satisfied at $\aug s$ and that $\Phi_{j_1}$ is satisfied at $\aug s$, so
        \[
        \mathbb{P}_{\aug \pi}(\phi_j,\aug s) = \sum_{s_i} \sum_{a\in \Act} P(s_i|s,a) \delta(a) \mathbb{P}(\phi_j|(s_i,\boldsymbol \mu^i, \boldsymbol \nu^i),\aug \MDP_{\aug \pi}).
        \]

        Moreover, $\nu_j$ satisfies

        \[
        \nu_j \leq \sum_{s_i} \sum_{a\in \Act} P(s_i|s,a) \delta(a) \nu_j^i.
        \]

        We denote by $D_{\Phi_j}=\aug{S}_{\phi_j,\leq}$ the set of augmented states $\aug s=(s,\boldsymbol \mu, \boldsymbol \nu)$ such that 
        \[
        \nu_j \leq \mathbb{P}(\phi_j | \aug \pi,\aug s,\aug \MDP_{\aug \pi}).
        \]

        Note that for any $\aug s$ such that $\phi_{j_1}\wedge \phi_{j_2}$ is satisfied, $\aug s \in \aug{S}_{\phi_j,\leq}$. Also, for any $\aug s$ such that $\phi_{j_1}$ is not satisfied, $\aug s \in \aug{S}_{\phi_j,\leq}$. Finally, for every $\aug s$ such that $succ(\aug s)\subseteq \aug{S}_{\phi_j,\leq}$, then $\aug s \in \aug{S}_{\phi_j,\leq}$.

    \end{itemize}
    
    We consider the following sets. 
    \[
    \left\{
    \begin{aligned}
        &A_{\Phi_{j}} = \{ \aug s\in \aug{S}_*,~\nu(j_1)=1 \},\\
        &\tilde{A}_{\Phi_{j}} = \{ \aug s\in \aug{S}_*,~\AugMDP_{\aug \pi}(\aug s)\models \Phi_{j_1} \},\\
        &B_{\Phi_{j}} = \{ \aug s\in \aug{S}_*,~(\nu(j_1)=1) \wedge (\nu(j_2)=1)  \},\\
        &\tilde{B}_{\Phi_j} = \{ \aug s\in \aug{S}_*,~\AugMDP_{\aug \pi}(\aug s)\models \Phi_{j_1} \wedge \Phi_{j_2} \},\\
        &C_{\Phi_j} = \aug{S}_* \setminus ( A_{\phi_j} \cup B_{\Phi_j} ).\\
    \end{aligned}
    \right.
    \]

    The induction hypothesis gives 
    \[
    A_{\Phi_j} \subseteq \tilde{A}_{\Phi_j},\quad B_{\Phi_j}\subseteq \tilde{B}_{\Phi_j},
    \]
    and we proved that $\tilde B_{\Phi_j} \subseteq D_{\Phi_j}$ and $C_{\Phi_j} \subseteq D_{\Phi_j}$.

    For a path $w=w_0w_1\dots$, we have 
    \[
    \begin{aligned}
         w\models \Phi_j\Leftrightarrow &(\forall i,~w_i\in \tilde A_{\Phi_j}) \\
        &\vee \left( \exists i_0,~w_{i_0}\in \tilde{B}_{\Phi_j} \wedge \forall i< i_0,~w_i\in \tilde{A}_{\Phi_j}\right).
    \end{aligned}
    \]

    We consider
    \[
    \begin{aligned}
    &\Phi_j^k(w) = \left( \forall i\leq k,~w_i\in \tilde{A}_{\Phi_j} \right) \\
    &\vee \left( \exists i_0 \leq k,~w_{i_0} \in \tilde{B}_{\Phi_j} \wedge \forall i\leq i_0,w_i\in \tilde{A}_{\Phi_j} \right).
    \end{aligned}
    \]
    Then,
    \[
    \Phi_j \Leftrightarrow \bigwedge_{k=0}^\infty \Phi_j^k.
    \]
    We also remark that $\Phi_{j}^{k_2} \Rightarrow \Phi_{j}^{k_1}$ for $k_2\geq k_1$. We now look at the state 
    \[
    D_{\Phi_j}^k = \{ \aug s\in \aug{S}_*,~\eta_j \leq \mathbb{P}(\phi_j^k|\aug s,\aug \pi) \}.
    \]

    We will show by induction that $D_{\Phi_j}^k = \aug{S}_*$ for any $k$. For $k=0$, we have
    \[
    \phi_j^0(w) = (w_0 \in \tilde{A}_{\Phi_j}).
    \]
    Then, $\mathbb{P}(\phi_j^0|\aug s,\aug \pi) =1$ on $\tilde{A}_{\Phi_j}$ and $0$ on the rest. By state compatibility and the induction hypothesis, $\eta_j=0$ on the complementary of $\tilde{A}_{\Phi_j}$ so $P_0$ is satisfied and $D_{\Phi_j}^0 = \aug{S}_*$.

    We now consider $P_k\Rightarrow P_{k+1}$. We hence assume that $D_{\Phi_j}^k=\aug{S}_*$ which means that for any $\aug s \in \aug{S}_*$,
    \[
    \eta_j \leq \mathbb{P}(\phi_j^k|\aug s, \aug \MDP_{\aug \pi}).
    \]
    We now consider $\aug s\in \aug{S}_*$. Note that for any $w$ emerging from $\aug s$, we have 
    \[
    w \models \phi_j^{k+1} \Leftrightarrow \left( w_0 \in \tilde{A}_{\Phi_j} \wedge (w_1\dots) \models \phi_j^k \right) \vee w_0 \in \tilde{B}_{\Phi_j}.
    \]

    We distinguish the cases.
    \begin{itemize}
        \item $\aug s \in \tilde{B}_{\Phi_j}$, then the property is automatically satisfied and $\aug s \in D_{\Phi_j}^{k+1}$.
        \item $\aug s \in \tilde{C}_{\Phi_j}$, then by state compatibility and induction hypothesis, $\eta_j=0$ so the property is satisfied and $\aug s \in D_{\Phi_j}^{k+1} $.
        \item $\aug s \in \tilde{A}_{\Phi_j}$. Then, $\nu_j$ satisfies
        \[
        \nu_j \leq \sum_{i} \sum_a \delta(a) P(\aug s_i|s,a) \nu_j^i.
        \]
        Since $\aug{S}_*\subseteq D_{\phi_j}^k$ by induction, we have 
        \[
        \nu_j^i \leq \mathbb{P}(\phi_j^k|\aug s_i,\aug \pi)
        \]
        for each successor $\aug s_i$. Since when $\aug s \in \tilde{A}_{\Phi_j}$, we have 
        \[
        \mathbb{P}(\phi_j^{k+1}|\aug s,\aug \pi) = \sum_{i} \sum_a \delta(a) P(\aug s_i|s,a) \mathbb{P}(\phi_j^{k}|\aug s_i,\aug \MDP_{\aug \pi}),
        \]
        we deduce that 
        \[
        \nu_j \leq \mathbb{P}(\phi_j^{k+1}|\aug s,\aug \MDP_{\aug \pi}).
        \]
        
    \end{itemize}

\end{proof}

%\section{For-sure \CPCTL$\,$formula}
\section{Properties of \CPCTL$\,$formulas}

We first show Lemma \ref{lemma-alit-W}. To that end, we show that if $\Phi$ is satisfied, its literal projection has to be also satisfied.
\begin{lemma}\label{lemma-avoid}
  For any \CPCTL$\,$state formula $\Phi$, any MDP $\mathcal M$, and any policy $\pi$ of $\mathcal M$, if $\mathcal M_\pi\models \Phi$ then $\mathcal M_\pi\models \alit{\Phi}$.
\end{lemma}
\begin{proof}
    We show the lemma by induction on the state 
    \CPCTL$\,$formula $\Phi$. If $\Phi$ is a literal, then the lemma is immediate. If $\Phi=\Phi_1\land\Phi_2$, then $\mathcal M_\pi \models \alit{\Phi_1}$ and $\mathcal M_\pi\models \alit{\Phi_2}$, so $\mathcal M_\pi \models \alit{\Phi_1}\land \alit{\Phi_2}=\alit{\Phi_1\land\Phi_2}$. If $\Phi=\mathbb P_{\geq p}\left[\Phi_1\mathbf{W}(\Phi_1\land \Phi_2)\right]$, then either $p=0$ and $\mathcal M_\pi\models \top=\alit{\Phi}$ or $p>0$ and $\mathcal M_\pi\models \Phi_1$ by definition of PCTL semantics, which implies by induction that $\mathcal M_\pi\models \alit{\Phi_1}=\alit{\Phi}$.
\end{proof}

We then show that, intuitively, if we can either go to states satisfying $\Phi$ through states satisfy its literal projection, or stay on on those states forever, then $\Phi$ is satisfied.
\begin{lemma}\label{lemma-tree-avoid}
    For any state \CPCTL$\,$formula $\Phi$, any MDP $\mathcal M$, and any policy $\pi$ of $\mathcal M$, we have $\mathcal M_\pi\models \Phi$ if and only if $\mathcal M_\pi\models \mathbb P_{\geq 1}[\alit{\Phi}\mathbf{W}\Phi]$.
\end{lemma}
\begin{proof}
    If $\mathcal M_\pi=\Phi$, then $\mathcal M_\pi\models \mathbb P_{\geq 1}[\alit{\Phi}\mathbf{W}\Phi]$ by the definition of PCTL semantics. We now suppose that $\mathcal M_\pi\models \mathbb P_{\geq 1}[\alit{\Phi}\mathbf{W}\Phi]$, and we show that $\mathcal M_\pi\models \Phi$ by induction on $\Phi$. The cases where $\Phi$ is a literal and where $\Phi=\Phi_1\land\Phi_2$ are immediate.

    We now do the case where $\Phi=\mathbb P_{\geq p}\left[\Phi_1\mathbf{W}(\Phi_1\land \Phi_2)\right]$, then if $p=0$, we have immediately $\mathcal M_\pi\models \Phi$. We thus now suppose that $p>0$. Recall that the states of the MDP induced by $\mathcal M$ are all finite paths $s_0\cdots s_n$ in $\mathcal M$ that can be obtained following $\pi$. For any finite path $s_0\cdots s_n$ in $\mathcal M_\pi$, such that $\mathcal M_\pi^{s_0\cdots s_n}\models\Phi$, we have $\mathcal M_\pi^{s_0\cdots s_n}\models \Phi_1$ by PCTL semantics. As a consequence, by induction, for any path $(s_n)_{n\in\mathbb N}$ of $\mathcal M_\pi$, if, for any $j$, $\mathcal M_\pi^{s_0\cdots s_i}\not\models \Phi_1\land\Phi_2$ for any $i<j$, then $\mathcal M_\pi^{s_0\cdots s_i}\models \Phi_1$ for any $i<j$.
    
    Now, let $U_1$ be the set of paths $(s_n)_{n\in\mathbb N}$ of $\mathcal M_\pi$ such that $\mathcal M_\pi^{s_0\cdots s_i}$ does not satisfy $\Phi_1\land\Phi_2$ for all $i\in\mathbb N$, and let $U_2$ be the set of paths of $\mathcal M_\pi$ that are not in $U_1$. Every path in $U_1$ satisfies $\mathbf{G}\Phi_1$ and thus $\Phi_1\mathbf{W}(\Phi_1\land\Phi_2)$. Furthermore, for any path $(s_n)_{n\in\mathbb N}$ in $U_2$, if $j$ is the minimum of the set of integers $i$ such that $\mathcal M_\pi^{s_0\cdots s_i}\models \Phi_1\land \Phi_2$, then since $\mathcal M_\pi^{s_0\cdots s_i}\models \Phi_1$ for any $i<j$ the measure of the paths with prefix $s_0\cdots s_j$ that satisfy $\Phi_1\mathbf{W}(\Phi_1\land\Phi_2)$ is at least the measure of the paths satisfying $\Phi_1\mathbf{W}(\Phi_1\land\Phi_2)$ in $\mathcal M_\pi^{s_0\cdots s_j}$, which is itself more than $p$. As a consequence, the measure of all paths satisfying $\Phi_1\mathbf{W}(\Phi_1\land\Phi_2)$ in $\mathcal M_\pi$ is more than the measure of all paths in $U_1$ plus $p$ times the measure of all paths in $U_2$, which is always more than $p$.
\end{proof}

We now use the above lemmas to characterize the states from which afor-sure PCTL formula is satisfiable.

\forsurecharacterization*
\begin{proof}
    If $\mathcal M_{\pi}\models \mathbb P_{\geq 1}[\phi]$, we have $\mathcal M_{\pi}\models \mathbb P_{\geq 1}[\alit{\Phi_1}\mathbf{W}(\Phi_1\land\Phi_2)]$ by Lemma \ref{lemma-avoid}. Conversely, suppose $\mathcal M_{\pi}\models \mathbb P_{\geq 1}[\alit{\Phi_1}\mathbf{W}(\Phi_1\land\Phi_2)]$, and let $(s_n)_{n\in\mathbb N}$ be a path in $\mathcal M_{\pi}$. Then, for any $j$ such that, for all $i<j$, $M_\pi^{s_1\cdots s_j}\not\models \Phi_1\land\Phi_2$, since $M_\pi^{s_1\cdots s_j}\models \mathbb P_{\geq 1}[\alit{\Phi_1}\mathbf{W}\Phi_1]$ by definition of PCTL semantics, we have by Lemma \ref{lemma-tree-avoid} that $M_\pi^{ s_1\cdots s_j}\models \Phi_1$. The lemma follows by definition of PCTL semantics.
\end{proof}

\section{The Algorithm \textsc{CPCTL-VI}}

\subsection{Preliminaries}

We first show Lemma \ref{projection}. To that end, we define from a valued policy $\aug\pi$ of the augmented MDP a projection $\aug\pi$ on the original MDP.
\begin{definition}[Projected policy]
    Let $\mathcal M$ be an MDP, $\Phi$ be a Safe PCTL policy, and $\aug\pi$ be a valued policy of $\mathcal M[\pi]$. We define the projection of $\aug\pi$ on $\mathcal M$ as the policy $\aug{\pi}$ of $\mathcal M$ such that, if we let $\Theta$, $\Delta$ and $\aug s_{0,\aug\pi} $ be as in Definition \ref{def:valued},
    \begin{itemize}
        \item $(\boldsymbol\mu_0,\boldsymbol\nu_0)$ are such that $s_{0,\aug\pi}=(\aug s_{0},\boldsymbol\mu_0,\boldsymbol\nu_0)$
        \item $\aug\pi(s_1\cdots s_n)=\Delta(s,\boldsymbol\mu_n,\boldsymbol\nu_n)$
        \item $(\boldsymbol\mu_{n+1},\boldsymbol\nu_{n+1})=(\boldsymbol\mu'_{s_{n+1}},\boldsymbol\nu'_{s_{n+1}})$, where $\Theta(s_n,\boldsymbol\mu_n,\boldsymbol\nu_n)=(\boldsymbol\mu'_{s'},\boldsymbol\nu'_{s'})_{s'\in S'}$ and $S'$ is the support of $\Delta(s,\boldsymbol\mu_n,\boldsymbol\nu_n)$.
    \end{itemize}
\end{definition}

Lemma \ref{projection} follows by construction of the projection.
\projection*

\begin{definition}\label{def:subformulasandvaluefrontier}
    % Let $\theta_1,\theta_2\in [0,1]^N$. We denote $\theta_1 \leq \theta_2$ if and only if for every $i \in \{1,\dots,N\}$,
    % \[
    % (\theta_1)_i \leq (\theta_2)_i.
    % \]
    Let $E \subseteq [0,1]^N$. We define the lower closure of $E$, denoted $E_\leq$, as 
    \[
    E_\leq = \left\{ p \in [0,1]^N,~\exists q \in E, ~p\leq q \right\}.
    \]
    With $\MDP$ a MDP, $\Phi$ a \CPCTL\ formula, and $V$ a family of subsets of $[0,1]^N$ indexed by $s\in \mathcal{S}$, we denote $V_{\leq}$ the family of sets defined for each $s\in \mathcal{S}$ as 
    \[
    V_{\leq}(s) = \left( V(s) \right)_{\leq}.
    \]
    Moreover, for $A_n\subseteq [0,1]^{N}$ and $B \subseteq [0,1]^{N}$, we write
    $A_n \rightarrow_n^* B$, or simply $A_n\rightarrow B$, if and only if, for every $b\in B$, for every $\varepsilon>0$, there exists $n_0$ such that for $n\geq n_0$, there exists $a\in A_n$ such that 
    \[
    \max(|a_i-b_i|) \leq \varepsilon.
    \]
    %Finally, for an MDP $\MDP$ and a \CPCTL$\,$formula $\Phi$, \introducePathStateSubFormPhi
    Then, we denote $V^*$ the \emph{Value Frontier}, as
    \[
    V^*(s)=\{(q_1,\dots,q_{pf(\Phi)}),~\exists \pi,\forall i,~~\MDP_{\pi}(s)\models \mathbb{P}_{\geq q_i} [\phi_i]\}.
    \] 
\end{definition}

\begin{definition}\textbf{(Finitely Decisive, non-evaporating policies).} We define $\mathtt{FD}(\Phi)$ by induction on both state and path formulas as:
\[
\mathtt{FD}:
\left\{
\begin{aligned}
    &b\in AP\mapsto b,\\
    &\neg b\in AP\mapsto \neg b,\\
    &\Phi_1 \wedge \Phi_2 \mapsto \mathtt{FD}(\Phi_1)\wedge \mathtt{FD}(\Phi_2),\\
    &\Phi_1 \mathbf W (\Phi_1 \wedge \Phi_2 ) \mapsto \mathtt{FD}(\Phi_1) \mathbf{U} \big[ \mathbb{P}_{=1}(\mathbf G \Phi_1) \vee (\mathtt{FD}(\Phi_1)  \wedge \mathtt{FD} (\Phi_2)) \big],\\
    &\mathbb{P}_{\geq p} (\phi) \mapsto \mathbb{P}_{\geq p} (\mathtt{FD}(\phi)).
\end{aligned}
\right.
\]

    Let $\MDP$ be a MDP and $\pi$ a policy. We say that $\pi$ is non-evaporating on a set of formulas $F$, if for all state formulas $\Phi \in F$, for any history $\rho$,
    \[
    \rho \models \Phi \Leftrightarrow \rho\models \mathtt{FD}(\Phi),
    \]
    and for any path formula $\phi\in F$, for all histories $\rho$ and almost every paths $\xi$ emerging from $\rho$,
    \[
    \rho,\xi \models \phi \Leftrightarrow \rho,\xi \models \mathtt{FD}(\phi).
    \]

    When considering a MDP $\MDP$ and a formula $\Phi$, we sometimes say that $\pi$ is non-evaporating without specifying $F$, which is implicitly chosen to be $SF(\Phi)\cup PF(\Phi)$.

\end{definition}

% \begin{lemma}
%     Let $\MDP$ be an MDP, $\Phi= \Phi_1 \mathbf{W}(\Phi_1\wedge \Phi_2)$ a \CPCTL$\,$formula and $\pi$ be a non-evaporating policy. Then
%     \[
%     \mathbb{P}( \Phi|s,\pi) = \mathbb{P}(\Psi|s,\pi),
%     \]
%     where 
%     \[
%     \Psi = \Phi_1 \mathbf{U} \left( (\Phi_1 \wedge \Phi_2)\vee \mathbb{P}_{=1}[\mathbf{G}\Phi_1] \right).
%     \]
    
% \end{lemma}

% \begin{proof}
%     We consider the set of paths
%     \[
%     H_1 = \{ \xi,~\xi\models \Phi \},~H_2 = \{ \xi,~\xi\models \Psi \}.
%     \]
%     We then define 
%     \[
%     A = \{ \xi\in H_1,~\xi\not\in H_2\},~B=\{ \xi\in H_2,~\xi \notin H_1\}.
%     \]
%     We see that for any $\xi \in B$, then there exists $i_0$ such that $\xi_{i_0}\models \mathbb{P}_{p=1}[\mathbf{G}\Phi_1]$. All the paths containing $\xi_{\leq i_0}$ as a prefix have probability $1$ to satisfy $\mathbf{G}\Phi_1$, hence, the total measure of $B$ is zero.

%     Conversly, we take $\xi\in A$. This means that for every $j$, $\xi_j\not\models \mathbb{P}_{=1}[\mathbf{G}\Phi_1]$. However, there exists $\varepsilon>0$ by non-evaporation property such that 
%     \[
%     \mathbb{P}(\mathbf{G}\Phi_1|\xi_j) \geq 1-\varepsilon \Rightarrow \mathbb{P}(\mathbf{G}\Phi_1|\xi_j) = 1.
%     \]
%     Hence, we have that 
%     \[
%     \forall j,~\mathbb{P}(\mathbf{G}\Phi_1|\xi_j) \leq 1-\varepsilon.
%     \]
%     We deduce that for any $j$, the probability of not satisfying $\mathbf{G}\Phi_1$ from $\xi_j$ is bigger than $\varepsilon$. In particular, the probability of staying on such a path while satisfying $\mathbf{G}\Phi_1$ is $0$.
    
% \end{proof}

\begin{lemma}
    For any profile $\boldsymbol{x,y}\in V_n(s)$ generated by the algorithm, the corresponding projected policy $\pi$ is non-evaporating. 
\end{lemma}

\begin{proof}
    This is due to the fact that $\pi$ has finite memory, which implies that $\pi$ is non-evaporating; see Lemma \ref{lem:finitemem}. Indeed, quick induction shows that the memory size of a projected policy corresponding to a profile $\boldsymbol{x,y}\in V_n(s)$ is at most $n$.
\end{proof}

\subsection{Proof of Soundness} 

\thsoundness*
\begin{proof}
    %We will use coherence theorem, stating that any policy satisfying the path compatibility and the state compatibility in the augmented MDP corresponds to a policy satisfying $\Phi$. 
    We consider $\MDP$ and the corresponding $\Phi$-augmented MDP $\AugMDP=\mathcal M[\Phi]$. The state $E_n$ provides us with counters. In order to embed our MDP in $\AugMDP$, we define a valuation function $\xi:\mathcal{S}\times \mathcal{SF}(\Phi) \times [0,1]^{pf(\Phi)}\to \{0,1\}$ by induction as follows. For any $s\in \mathcal{S}$ and $(\eta_1,\dots,\eta_{pf(\Phi)})\in E_n(s)$, we set (and ommit the third parameter for readability)
    \[
    \xi( \mathbb{P}_{\geq p_j }\phi_j,s) = \left\{
    \begin{aligned}
        &0,~\text{ for }\eta_j<p_j,\\
        &1,~\text{ for }\eta_j\geq p_j,
    \end{aligned} 
    \right.
    \]
    \[
    \xi( \Phi_1 \wedge \Phi_2,s) = \xi(\Phi_1,s) \xi(\Phi_2,s), 
    \]
    \[
    \xi(a,s) = \left\{
    \begin{aligned}
        &0,~\text{ for }a\in L(s),\\
        &1,~\text{ for }a\notin L(s),
    \end{aligned} 
    \right.
    \]
    \[
    \xi(\neg a,s) = 1-\mu(a,s).
    \]

    We denote $\boldsymbol\mu=\xi(s,\Phi,(\eta_1,\dots,\eta_{pf(\Phi)}))$ and $\mu_i=\xi(s,\Phi_i,(\eta_1,\dots,\eta_{pf(\Phi)}))$. Since $sf(\Phi) \geq pf(\Phi)$, we can reorder so that for $j\leq pf(\Phi)$, $\Phi_j = \mathbb{P}_{\geq p_j} (\phi_j)$. Finally, for $\theta\in [0,1]^{pf(\Phi)}$ and $s\in \mathcal{S}$, we denote this whole operation $(\mu_1,\dots,\mu_{sf(\Phi)}) = \Xi(s,\Phi,\boldsymbol\theta)$, or simply $\Xi(s,\boldsymbol\theta)$. The goal is now to show that there exists a policy $\aug \pi$ such that $\aug \pi$ satisfy both the path-compatibility and the state-compatibility.

    \textbf{Definition of the policy}
    We consider $\mathcal{S}_n$ the set of augmented states $\aug s = (s,\boldsymbol\mu,\boldsymbol\nu)$ such that $\theta \in E_n(s)$ and $\mu = \Xi(s,\Phi,\theta)$. By definition of $E_n(s)$, there exists an action distribution $\delta$ and for every successor $s_i$, there exist counters $\boldsymbol\nu^i\in E_{n-1}(s)$ such that the conditions are satisfied. 
    We define $\aug \pi$ as a valued policy. For $\aug s$, we set for the choice function and for the distribution
    \[
    \Theta(\aug s) = (\boldsymbol\mu^i,\boldsymbol\nu^i)_i,\quad \Delta(\aug s) = \delta,\quad \boldsymbol\mu^i=\Xi(s_i,\boldsymbol\theta^i).
    \]

    Note that we only defined $\aug \pi$ on $\mathcal{S}_n$. However, since for any $j$, $E_j(s)\subseteq E_{j+1}(s)$, we obtain in particular that $\boldsymbol\nu^i \in E_{n}(s)$. Hence, $\aug \pi$ is well defined from any starting state $s\in \mathcal{S}_n$. 
    
    \textbf{$\aug \pi$ is path-compatible and state-compatible}

    Let $\aug s = (s,\boldsymbol\mu,\boldsymbol \nu) \in \mathcal{S}_n$, and let $\Phi_j$ be a path formula.
    \begin{itemize}
        \item[(i)] If $\Phi_j= \mathbb{P}_{\geq p_j} (\phi_j)$, then $\mu_j=\Xi(s,\Phi_j,\boldsymbol\nu)$, so by definition of $\Xi$, $\mu_j=1$ implies $\nu_j\geq p_j$.
        \item[(ii)] If $\Phi_j= \Phi_{j_1}\wedge \Phi_{j_2}$, then $\mu_j=\Xi(s,\Phi_j,\boldsymbol \nu)$, so by definition of $\Xi$, $\mu_j=\Xi(s,\Phi_{j_1},\boldsymbol\nu)\Xi(s,\Phi_{j_2},\boldsymbol\nu)=\mu_{j_1}\mu_{j_2}$. 
        \item[(iii)] The cases $\Phi_j=\Phi_{j_1}\wedge \Phi_{j_2}$ and $\Phi_j=b\in AP$ are also obtained using the definition of $\Xi$.
    \end{itemize}
    
    Hence, $\aug \pi$ is state-compatible. 
    
    We now go on with the path-compatibility. We consider $\phi_j=\Phi_{j_1}\mathbf W(\Phi_{j_1}\wedge \Phi_{j_2})\in \mathcal{PF}(\Phi)$, where  
    \[
    \Phi_{j_1} = \bigwedge_{i=1}^{n_1} \Phi_{l_i},\quad \Phi_{j_2} = \bigwedge_{i=1}^{n_2} \Phi_{k_i}.
    \]
    
    Using the definition of $\Xi$, it is easy to verify that
    \begin{equation}\label{comp1}
    \mu_{j_1} = \prod_{i=1}^{n_1} \mu_{l_i},\quad \mu_{j_2} = \prod_{i=1}^{n_2} \mu_{k_i}.
    \end{equation}
    Hence, $\mu_{j_1} = 0$ implies the existence $i_0$ such that $\mu_{l_{i_0}}=0$. By definition of $\mu_{l_{i_0}}$ and $\Xi$, this implies $\nu_{l_{i_0}} < p_{l_{i_0}}$. The Bellman operator then imposes that $\nu_{j_1}=0.$

    In the case where $\mu_{j_1}=1$, we must verify 
    \[
    \nu_{j} \leq \max \left( \sum_{i=1}^m P(s_i|s,a) \delta(a) \nu_j^i, \mu_{j_2} \right)
    \]
    We first assume $\mu_{j_2}=0$. In that case, with $\aug s_i = (s_i,\theta_i,\mu_i)$ the successors of $\aug s$, by construction, we have 
    \[
    \mu_j = \sum_{s'}\sum_a \delta(a) P(s'|s,a) \mu_j^i,
    \]
    so the state compatibility is verified.
    Finally, if $\mu_{j_2}=1$ and $\mu_{j_1}=1$, the condition is automatically satisfied, as it reads $\nu_j\leq 1$. 

    We show $\nu_j=1$, to illustrate that the policy $\aug \pi$ in fact satisfies the limit case of the path compatibility.
    
    Using \eqref{comp1}, we get for all $i$ that $\mu_{l_i}=\mu_{k_i}=1$ by \eqref{comp1}. Since $\mu_{l_i}=\xi(s,\Phi_{l_i},\boldsymbol\nu)$ and $\mu_{k_i}=\xi(s,\Phi_{k_i},\boldsymbol\nu)$, we have by definition of $\xi$ 
    \[
    \forall i\leq n_1,~\nu_{l_i}\geq p_{l_i},\quad \forall i \leq n_2, ~\nu_{k_i}\geq p_{k_i}.
    \]
    Hence, by the definition of the Bellman operator, $\nu_j=1$.
    Since the policy satisfies both path compatibility and state compatibility, we obtain that the policy is coherent starting for any $\aug s \in S_n$. In particular, for any $\boldsymbol\nu = (\nu_1,\dots,\nu_{pf(\Phi)}) \in E_n(s_0)$, 
    \[
    \aug s_{\Phi} = (s_0,\Xi(s,\Phi,\theta),\boldsymbol\nu) \in S_n,
    \]
    so when $\Phi$ has the form $\Phi=\mathbb{P}_{\geq p} (\phi)$,
    \[
    (\aug s_{\Phi},\aug \pi) \models \mathbb{P}_{\geq \eta} (\phi).
    \]
    and in all cases,
    \[
    \mu_{j_0} = 1 \Rightarrow (\aug s_{\Phi},\aug \pi) \models \Phi.
    \]
\end{proof}

\subsection{Proof of Completeness}

We recall \cref{thm:VI-optimality}:
\optimality*

% \begin{lemma}
%     Let $\Phi=\Phi_1 \mathbf{W} \Phi_2$ be a \CPCTL$\,$formula and $\MDP$ an MDP. Under the generalized Slater's assumption, for any policy $\pi$, there exists $\tilde \pi$ non-evaporating such that  
%     \[
%     \mathbb{P}(\Phi|s,\pi) \leq \mathbb{P}(\Phi|s,\tilde \pi).
%     \]
% \end{lemma}

\newenvironment{indhyp}
  {\par\smallskip\begin{leftbar}}  % start of environment
  {\end{leftbar}\par\smallskip}    % end of environment

\begin{lemma}

    % We assume the existence of $\pi_s$ non-evaporating such that with $\boldsymbol{r,t} = \mathtt{P}(\pi_s|\mathcal{M}^s,\Psi)$, then $\boldsymbol{t} > \boldsymbol{p}$.

    % and denote $\boldsymbol{r,t} = \mathtt{P}(\pi|\mathcal{M}^s,\Psi)$. Denote $V_n$ the value frontier obtained after running \textsc{CPCTL-VI} for $n$ steps on the formula $\Phi$ and MDP $\MDP$. For any $\varepsilon >0$, there exists $n_0\in \mathbb{N}$ such that for $n\geq n_0$, there exists $\boldsymbol{x,y}\in V_n(s)$ such that 
    % \[
    % (\boldsymbol{x,y}) \geq (\boldsymbol{r,t})-\varepsilon.
    % \]

    % We assume Slater's generalized assumption and that $\mathtt{T}(\Phi,\boldsymbol q)$ is satisfiable with $\boldsymbol q \geq \boldsymbol p$. For any $\varepsilon\geq 0$, 
    
    % Denote $V_n$ the value frontier obtained after running \textsc{CPCTL-VI} for $n$ steps. Let $\Psi$ a non-evaporating dominating formula provided by Slater's generalized assumption. For any starting state $s$, and any policy $\pi$ of $\mathcal{M}^s$, with 
    % \[
    % \boldsymbol{r,t} = \mathtt{P}(\pi|\mathcal{M}^s,\Psi),  
    % \]
    % then for any $\varepsilon>0$, there exists $n_0\in \mathbb{N}$ such that for $n\geq n_0$, there exists $\boldsymbol{x,y}\in V_n(s)$ such that 
    % \[
    % (\boldsymbol{x,y}) \geq (\boldsymbol{r,t})-\varepsilon.
    % \]
    
    We consider $\mathcal{M}$ a MDP and $\Phi$ a \CPCTL\ formula, $\boldsymbol p$ the probability thresholds of $\Phi$. We assume the existence of $\boldsymbol q > \boldsymbol p$, such that the formula $\Psi=\mathtt{T}(\Phi,\boldsymbol q)$ is satisfied by a non-evaporating policy $\pi_s$. Denote $V_n$ the value frontier obtained after running \textsc{CPCTL-VI} for $n$ steps on the formula $\Phi$ and MDP $\MDP$.
    With $\boldsymbol{r,t} = \mathtt{P}(\pi_s|\mathcal{M},\Psi)$, for any $\varepsilon >0$, there exists $n_0\in \mathbb{N}$ such that for $n\geq n_0$, there exists $\boldsymbol{x,y}\in V_n(s)$ such that 
    \[
    \left\{
    \begin{aligned}
        &\boldsymbol{x} \geq \boldsymbol{ r},\\
        &\boldsymbol{y} \geq \boldsymbol{t} - \varepsilon \boldsymbol{1}.
    \end{aligned}
    \right.
    \]

\end{lemma}

\begin{proof}

We prove the lemma using the following strategy, using nested inductions. 
\begin{itemize}
    \item Induction $\mathcal{I}_1$ shows the result by induction on the total depth of the formula, assuming that the algorithm computes paths of finite length.
    \item Induction $\mathcal{I}_2$ shows by induction (on length of paths) that the algorithm computes finite paths of growing length.
    \item Induction $\mathcal{I}_3$ is mostly a technicality and shows, by iterating over the depth of formulas, that some policies can be chosen by the Bellman operator.
\end{itemize}

\begin{indhyp}
\noindent\textbf{Induction on the total depth of $\Phi$, $\mathcal{I}_1(k)$:}

\emph{There exists $n_0$ such that for $n\geq n_0$:}

\emph{
Let $s \in \mathcal{S}$ and $\mathcal{Q}(s)$ be the value frontier for the formula $\Psi=\mathtt{T}(\Phi,\boldsymbol{q})$ achieved by non-evaporating policies. Let $\boldsymbol{t}\in \mathcal{Q}(s)$ and $\pi$ non-evaporating such that there exists $\boldsymbol{r} $ satisfying $\boldsymbol{r,t} = \mathtt{P}(\pi|\MDP^\pi,\Psi)$.
%Let $\pi$ be a fixed policy. Then for any $s\in \mathcal{S}$, with $\boldsymbol{r,t} = \mathtt{P}(\pi|\mathcal{M}^s,\Psi)$, 
There exists $\mathbf{x},\mathbf{y}\in V_{n}(s)$ satisfying
\[
\forall \phi \in \mathcal{PF}(\Phi),~\mathcal{D}_t(\phi) \leq k \Rightarrow~y_\phi\geq t_\phi-\varepsilon,
\]
which we denote $clip(\mathbf{y},\leq k) \geq clip(\mathbf{t},\leq k)$, and 
\[
\forall \Psi \in \mathcal{SF}(\Phi),~\mathcal{D}_t(\Psi) \leq k \Rightarrow~ x_\Psi=r_\Psi,
\]
or in other words, $clip(\mathbf{x},\leq k) = clip(\mathbf{r},\leq k)$. %(*) Moreover, the index $n$ can be chosen uniformly with respect to $s$.
}
\end{indhyp}

    %(*) Is a direct consequence of the finiteness of $\mathcal{S}$ and the compactness of $[0,1]^k$ for any $k$.

    \begin{itemize}
        \item \textbf{Induction $\mathcal{I}_1$, case $k=0$:}
             We consider $\Phi$ such that the total depth of $\Phi$ is zero, i.e. $\mathcal{D}_t(\Phi)=0$. $\Phi\in$\CPCTL\ does not involve any temporal operator, i.e. 
            \[
            \Phi = \left( \bigwedge_{k=1}^{l_1} b_{j_k}\right) \wedge \left( \bigwedge_{k=l_1+1}^{l_2} \neg b_{j_k}\right). 
            \]
            Let $\mathcal{S}_\Phi=\llbracket \Phi \rrbracket$ be such that $s \in \mathcal{S}_\Phi$ if and only if $\mathcal{M}^s \models \Phi$. After one step of the algorithm, for all $s \in \mathcal{S}_{\Phi}$, $(\mu,\nu)\in V_1(s)$, where $\mu_j = 1$ if and only if there exists $k\leq l_2$ such that $j=j_k$, and moreover the valuation of the formula $\Phi$ is set to one, i.e. $\mu_{j_\Phi}=1$. Indeed, the vector $(\mu,\nu)$ satisfies the Bellman condition:
            \begin{enumerate}
                \item For all $\Psi=b_{j_k}$, $b_{j_k}\in L(s) $.
                \item For all $\Psi=\neg b_{j_k}$, $b_{j_k}\notin L(s)$.
                \item For $\Phi = \bigwedge_k \Psi_k$, then $$\mu_{j(\Phi)}= \prod_{k=1}^{l_1} \mu_{j(b_{j_k})} \cdot \prod_{k=l_1+1}^{l_2} \mu_{j(\neg b_{j_k})}=1$$.
            \end{enumerate}

        \item \textbf{Induction $\mathcal{I}_1$, $\mathcal{I}_1(k)\Rightarrow \mathcal{I}_1(k+1)$:} 

    \noindent \textbf{First case:}

            We assume that $\Phi=\Phi_0$ is of the form
                \[
                \Phi_0 = \left( \bigwedge_{j=1}^{k_1} \Phi_{j}\right) \wedge \left( \bigwedge_{j=k_1+1}^{k_2} \Phi_{j}\right),
                \]
                where
                \[
                \forall 1\leq j\leq k_1,~\Phi_{j}=\mathbb{P}_{\geq p}(\phi_j), 
                \]
                and 
                \[
                \forall k_1<j\leq k_2,~\Phi_j=b_j \text{ or }\Phi_j=\neg b_j \text{ for }b_j\in AP.
                \]

                Using a similar argument as in \textbf{base case}, it is easy to check that for every profile $\mathbf{v},\mathbf{w} \in V_n(s)$ , we also have $\mathbf{x},\mathbf y \in V_n(s)$ with 

                \[
                \left\{
                \begin{aligned}
                    & x_j = \max(v_j,1) \text{ if } b_j\in L(s) \text{ and } \Phi_j=b_j \text{ for } k_1<j\leq k_2,\\
                    & x_j = \max(v_j,1) \text{ if } b_j\notin L(s) \text{ and } \Phi_j=\neg b_j \text{ for } k_1<j\leq k_2,\\
                    &x_j = v_j \text{ otherwise,}\\
                    &\mathbf{y}=\mathbf{w}.
                \end{aligned}
                \right.
                \]
                
                If $s \not\in S_\Psi$, then $r_{\Psi}=0$ and there is nothing more to prove. Otherwise, let $s \in \mathcal{S}_\Psi$.
                %, and let $\mathbf r,\mathbf t\in \mathcal{Q}(s)=\{ \mathtt{P} ( \Psi|s,\pi),\pi \text{ policy} \}$. 
                By definition, %there exists $\pi$ a policy on $\mathcal{M}^s$ such that 
                \[
                \MDP_\pi^s \models \mathtt{T}(\Phi,\boldsymbol{q}).
                \]
                %By definition, 
                And therefore, for all $j\leq k_1$,
                \[
                \MDP_\pi^s \models \Psi_j= \mathtt{T}(\Phi_j,\boldsymbol{q}).
                \]
                $\Phi_j$ then satisfies Slater's generalized assumption, and by assumption, since for all $j$, $\mathcal{D}_t(\Phi_j) \leq k$, there exists $\boldsymbol{x,y} \in V_n(s)$ for $n$ large enough, satisfying
                \[
                clip(\boldsymbol{x,y},\leq k) \geq clip(\boldsymbol{r,t},\leq k).
                \]
                Additionally, one can choose 
                \[
                \forall k_1<  j\leq k_2,~x_{j}=1.
                \]
                The tuple $\boldsymbol{v,w}$ defined as 
                \[
                \forall j\neq 0, v_j=x_j,~v_{0} =1,~\forall j,~w_j=y_j,
                \]
                satisfies the Bellman's condition. Indeed, all the conditions imposed by the subformulas $\phi_j$ and $\Phi_j$ for $j\neq 0$ are already satisfied, and the condition imposed by $\Phi_{0}$ is of the form
                \[
                1=v_{0} = \prod_{j=1}^{k_1} v_{j} \cdot \prod_{j=k_1+1}^{k_2} v_{j} = 1. 
                \]

  \noindent \textbf{Case two, $\Phi=\mathbb{P}_{\geq p}(\phi),~\phi=\Phi_1 \mathbf{W}(\Phi_1 \wedge \Phi_2)$:}
        
Let $\phi=\Phi_1 \mathbf{W}(\Phi_1 \wedge \Phi_2)$. We denote $\psi=\Psi_1 \mathbf{W} (\Psi_1 \wedge \Psi_2)$ the corresponding Slater's modified formulas. By induction, for $n$ large enough, there exists $\boldsymbol{x,y} \in V_n(s)$ such that 
\[
clip(\boldsymbol{x,y},\leq k) \geq clip(\boldsymbol{r,t},\leq k).
\]

We distinguish three cases. 

\noindent \textbf{(i) $\mathcal{M}_\pi^s \not\models \Psi_1$:} 
This case is trivial, as the new condition is of the form $x_1 \geq 0.$

\noindent \textbf{(ii) $\mathcal{M}_\pi^s \models \Psi_1 \wedge \Psi_2$}: then 
    $t_1 \geq q_1$ and $t_2 \geq q_2$. There exists $\boldsymbol{x,y}$ satisfying $clip(\boldsymbol{x,y},\leq k) \geq clip(\boldsymbol{r,t},\leq k)$, which implies in particular that $x_1 \geq q_1 \geq p_1$ and $x_2 \geq q_2\geq p_2$. In particular, the new profile with $x_0=x_\Phi=1$ and $y_0=1$ satisfies the Bellman's condition.

\noindent \textbf{(iii) $\mathcal{M}_\pi^s \models \Psi_1$ and $\mathcal{M}_\pi^s \not\models \Psi_2$}. This case is done separately, by induction.
\end{itemize}

\end{proof}
\begin{proof}\textbf{Of (iii).}

%Let $s\in \mathcal{S}$ and $\boldsymbol t\in \mathcal{Q}(s)$ a value of the value frontier for $\Psi=\mathtt{T}(\Phi,\boldsymbol q)$ achieved by the policy $\pi_s$ [TODO], and define $\boldsymbol r$ such that $\boldsymbol{r},\boldsymbol{t}=\mathtt{P}(\pi_s|\MDP^s,\Psi)$. 
\begin{indhyp}
\noindent \textbf{Induction assumption $\mathcal{I}_2(k)$:}
    We consider 
    \[
    \Phi = \mathbb{P}_{\geq p_{j_0}} (\phi_{j_0}),~\phi_{j_0}= \Phi_1\mathbf{W}(\Phi_1\wedge \Phi_2),~\Phi_{1} = \bigwedge_{l=1}^{n_1} \Phi_{1,j_l},~\Phi_{2} = \bigwedge_{l=1}^{n_2} \Phi_{2,k_l},
    \]
    %and we define the set
    % \[
    % A_\Phi = \{ \rho,~\forall l\leq n_1,~\mathbb{P}(\phi_{1,j_l}|\rho,\pi_s)\geq q_{1,j_l} \}.
    % \]
    
%And let $j_0$ (resp. $j_1$, $j_2$) such that $\Phi=\Phi_{j_0}$ (resp. $\Phi_1=\Phi_{j_1}$, $\Phi_2=\Phi_{j_2}$).
    
For any $\Phi_l$ subformula of $\Phi$, we denote $\Psi_l$ the corresponding subformula of $\Psi=\mathtt{T}(\Phi,\boldsymbol q)$, and we define 
\[
\varphi_k = \Psi_1 \mathbf{U}^{\leq k} \left( \mathbb{P}_{=1}[\mathbf{G}\Psi_1] \vee \Psi_2 \right).
\]
%With $\boldsymbol{r}^i,\boldsymbol{t}^i=\mathtt{P}(\pi_s|\mathcal{M}_{\pi_s}^{s_i},\mathtt{T}(\Phi,\boldsymbol{q}))$, there 
Let $s \in \mathcal{S}$ and $\boldsymbol{t} \in \mathcal{Q}(s)$. Let $\pi$ be a policy such that there exists $\boldsymbol{ r}$ satisfying $\boldsymbol{r},\boldsymbol{t}= \mathtt{P}(\pi|\mathcal{M}^s,\Psi)$. There exists $n_0$ such that for $n\geq n_0$, there exists $\boldsymbol{y}$ satisfying

%With $\boldsymbol{r},\boldsymbol{t}=\mathtt{P}(\pi_s|\mathcal{M}_{\pi_s}^{\rho},\mathtt{T}(\Phi,\boldsymbol{q}))$, there
%exists $\boldsymbol{y}=(y_1,\dots,y_{pf(\Phi)})$ such that 
\[
\left\{
\begin{aligned}
    &\forall j\neq j_0,~y_j \geq t_j-\varepsilon,\\
    &y_{j_0} \geq \mathbb{P}(\varphi_k|\pi,s)-\varepsilon,
\end{aligned}
\right.
\]
and $(\Xi(\Phi,s)(\boldsymbol{y}),\boldsymbol{y})\in V_n(s)$.

\end{indhyp}

The base case directly follows the definition of the Bellman operator as well as by Lemma 1 and Lemma 6.

We now assume $\mathcal{I}_2(k)$ to be true and will show  $\mathcal{I}_2(k+1)$.

%Let $\boldsymbol{t}\in \mathcal{Q}(s)$, the value frontier for the modified formula $\mathtt{T}(\Phi,\boldsymbol q)$. It thus exists a policy $\pi$ such that, starting from $s$ and considering the Markov Chain $\MDP_{\pi}$, we have
Let $\boldsymbol{r,t}$ and $\pi$ be defined as in the induction assumption $\mathcal{I}_2$. The policy $\pi$ satisfies
\[
\MDP_{\pi} \models \mathbb{P}_{\geq t_j}[\psi_j],~\forall j.
\]
where $\psi_j$ is the path formula associated to the state formula $\Psi_j$. 
%In this case, we have $\boldsymbol{r,t}=\mathtt{P}_{\mathtt{T}(\Phi,\boldsymbol q)}(s,\pi_s)$. 
Denote $\delta$ the distribution chosen by $\pi$ over the actions starting at $s$, and with $s^i$ the successors of $s$ that are reachable in one step following $\pi$ from $s$, we denote $\boldsymbol{r}^i,\boldsymbol{t}^i=\mathtt{P}(\pi|\mathcal{M}^{ss^i},\mathtt{T}(\Phi,\boldsymbol{q}))$, where we abuse notation by writing $\mathtt{P}(\pi|\mathcal{M}^{ss^i},\mathtt{T}(\Phi,\boldsymbol{q}))$ for $\mathtt{P}(\pi'|\mathcal{M}^{s^i},\mathtt{T}(\Phi,\boldsymbol{q}))$ with $\pi'(\rho)=\pi(s\rho)$.  Then, for any 
\[
\Psi_{j}= \mathbb{P}_{\geq q_j} \left[ \psi_j \right],~ \psi_j= \left( \bigwedge_{i=1}^{m_{j_1}} \Psi_{j_i} \right) \mathbf{W} \left[\left( \bigwedge_{i=1}^{m_{j_1}} \Psi_{j_i} \right)\land \left( \bigwedge_{i=1}^{m_{j_2}} \Psi_{k_i} \right)\right],
\]
we have 
\[
\mathbb{P}(\psi_j|s,\pi) = \left\{
\begin{aligned}
    &0,~\text{ if } \exists i_0,~\mathbb{P}(\psi_{j_{i_0}}|\pi,s) < q_{j_{i_0}},\\
    &1,~\text{ if } \forall i\leq m_{j_1}, \mathbb{P}(\psi_{j_i}|\pi,s) \geq q_{j_i},\\
    &\text{ and }~\forall i\leq m_{j_2}, \mathbb{P}(\psi_{k_i}|\pi,s) \geq q_{k_i},\\
    &\sum_{a}\sum_{s^i} P(s^i|s,a)\delta(a) \mathbb{P}(\psi_j|\pi',s^i)~\text{ otherwise}.
\end{aligned}
\right.
\]
With our notations, 
\[
t_j = \left\{
\begin{aligned}
    &0,~\text{ if } \exists i_0,~t_{j_{i_0}} < q_{j_{i_0}},\\
    &1,~\text{ if } \forall i\leq m_{j_1}, t_{j_i} \geq q_{j_i},\\
    &\text{ and }~\forall i\leq m_{j_2}, t_{k_i} \geq q_{k_i},\\
    &\sum_{a}\sum_{s^i} P(s^i|s,a)\delta(a) t_j^i \text{ otherwise}.
\end{aligned}
\right.
\]
By definition, the points $\boldsymbol t^i$ satisfy $\boldsymbol t^i \in \mathcal{Q}(s^i)$, so using the Induction Hypothesis 2 $\mathcal{I}_2$, there exists $\boldsymbol x^i,\boldsymbol y^i \in V_n(s^i)$, such that 
\[
\left\{
\begin{aligned}
    &\boldsymbol y^i \geq_{\neq j_0} \boldsymbol t^i - \varepsilon \boldsymbol{1},\\
    & y_{j_0}^i \geq \mathbb{P}(\varphi_k|\pi',s^i) - \varepsilon. 
\end{aligned}
\right.
\]
We will therefore show that there exists a point $\boldsymbol y$ that satisfies our requirements. We distinguish the different cases. We start by showing that we can select appropriate values for the other values of $j$, meaning the subformulas $j\neq j_0$. We do it by induction.

\begin{indhyp}
\indent\textbf{Induction Hypothesis 3:} If $t_j \geq q_j$, we can choose $y_j\geq t_j-\varepsilon > p_j$.
\end{indhyp}
We start with the atomic propositions and obtain the result in this case. We now consider a subformula $\phi_j$, and we have shown that for any subformula $\phi_{j_i}$ of $\phi_j$, if we have $t_{j_i} \geq q_{j_i}$, we can choose $y_{j_i}\geq t_{j_i}-\varepsilon$.

We look at the different cases provided by the Bellman operator and choose the distribution $\delta$.
\begin{itemize}
    \item If $\exists i_{0}$ such that $y_{j_{i_0}} < p_{j_{i_0}}$, we have to choose $y_j=0$. By \textbf{induction assumption 3}, it implies that $t_{j_{i_0}}<q_{j_{i_0}}$, so $t_j=0$.
    \item If for all $i\leq m_{j_1}$, $y_{j_{i}} \geq p_{j_{i}}$ and for all $i\leq m_{j_2}$, $y_{k_i} \geq p_{k_{i}}$, then we can choose $y_{j}=1$, so the induction step also holds.
    \item Otherwise $y_j$ must satisfy
    \[
    y_j\leq \sum_{a}\sum_{s^i} P(s^i|s,a)\delta(a) y_{j}^i.
    \]
    Since we have
    \[
    t_j\leq \sum_{a}\sum_{s^i} P(s^i|s,a)\delta(a) t_{j}^i,
    \]
    and $y_{j}^i \geq t_j^i-\varepsilon$, we can choose $y_j=t_j-\varepsilon.$
\end{itemize}
This concludes the proof of \textbf{Induction 3}. We hence have the existence of an element $\boldsymbol y$ such that $\boldsymbol x, \boldsymbol y\in V_n^{k+1}(s)$, with $\boldsymbol x = \Xi(\Phi,s)(\boldsymbol y)$, satisfying \textbf{Induction 3}. We now consider the last coordinate $y_1$. Since this value does not intervene in the constraints satisfied by the values corresponding to the subformulas, we can choose any $y_1$ as long as it satisfies the constraints imposed by the Bellman operator. We will now show \textbf{Induction 2}, and recall that by \textbf{Induction assumption 2}, each $\boldsymbol y^i$ satisfies $y_{j_0}^i \geq \mathbb{P}(\varphi_{k}|s^i,\pi')$. 
%We re-index $\Phi=\Phi_1 = \mathbb{P}_{\geq p_1}(\phi_1),~\phi_1 = \Phi_2 \mathbf{W} (\Phi_2 \wedge \Phi_3)$. 
The component $y_{j_0}$ has to satisfy
\begin{itemize}
    \item $y_{j_0}=0$ if there exists $i_0$, $y_{j_{i_0}}\leq p_{j_{i_0}}$. By \textbf{induction 2}, this means that $t_{j_{i_0}}< q_{j_{i_0}}$, so choosing $y_{j_0}=0$ satisfies our assumption in that case, since $ \mathbb{P}(\varphi_{k+1}|s,\pi)=0.$
    \item In the second case $y_{j_0}=1$, there is nothing to prove.
    \item Finally, if $\forall i<n_1$, $y_{j_i} \geq p_{j_i}$ and there exists $i_0$, $y_{k_{i_0}}< p_{k_{i_0}}$. By \textbf{Induction 2},  we also have $t_{k_{i_0}}<q_{k_{i_0}}$. In that case, since $\MDP_{\pi}(s) \models \Psi_2$, and $\MDP_{\pi}(s)\not\models \Psi_3$, we have 
    \[
    \mathbb{P}(\varphi_{k+1}|s,\pi) = \sum_{a}\sum_{s^i} \delta(a) P(s^i|s,a) \mathbb{P}(\varphi_k|\pi',s^i).
    \]
    Finally, we can choose 
    \[
    \begin{aligned}
    &y_{j_0}  = \sum_{a}\sum_{s^i} \delta(a) P(s^i|s,a) y_{j_0}^i \\
    &\geq \sum_{a}\sum_{s^i} \delta(a) P(s^i|s,a) \left( \mathbb{P}(\varphi_k|\pi',s^i) - \varepsilon \right) \\
    &\geq \mathbb{P}(\varphi_{k+1}|s,\pi) - \varepsilon.
    \end{aligned},
    \]
    which satisfies \textbf{Induction 2}.

    We therefore obtain the following. 
\[
y_{j_0} \geq \mathbb{P}(\varphi_\infty|s,\pi)-\varepsilon,\quad \varphi_\infty = \Psi_1 \mathbf{U} (\mathbb{P}_{=1} [\mathbf{G}\Psi_1]\vee \Psi_2)%= \mathtt{FD}(\Psi_1 \mathbf{W} (\Psi_1 \wedge \Psi_2)).
\]
In particular, for any $\pi$ non-evaporating, we obtain 
\[
y_{j_0} \geq \mathbb{P}(\Psi_1 \mathbf{W} (\Psi_1 \wedge \Psi_2)|s,\pi)-\varepsilon=\mathbb{P}(\psi|\pi,s)-\varepsilon.
\]

\end{itemize}

Therefore, for any $s\in \mathcal{S}$, for any $\boldsymbol t\in \mathcal{Q}(s)$, for any non-evaporating policy $\pi$ such that there exists $\boldsymbol r$ satisfying $ \boldsymbol{r,t} = \mathtt{P}(\pi|\MDP^s,\Psi)$, there exists $n(s,\varepsilon)$ such that $n\geq n(s,\varepsilon)$ implies
\[
clip(\boldsymbol{x},(\boldsymbol{y}+\varepsilon\boldsymbol{1}),\leq k+1) \geq clip( \mathtt{P}(\pi|\MDP^{s},\Psi),\leq k+1),
\]
where $\Xi(\Phi,s)(\boldsymbol{y}),\boldsymbol{y}\in V_n(s)$. Since the family of value frontiers $V_n$ are indexed by a finite amount of elements $s\in \mathcal S$ and they are non-decreasing, they converge uniformly toward their limit. Hence, one can choose a single $n_{k+1}$ such that for all $s \in \mathcal{S}$, for all $\boldsymbol t\in \mathcal{Q}(s)$, then any set $V_{n}$ for $n\geq n_{k+1}$ satisfies the requirements, which shows $\mathcal{I}_{1}(k+1)$ and concludes the proof. 

%We now consider $\pi_s$, the policy provided by Slater's generalized assumption which is moreover assumed to be non-evaporating, and therefore obtain
% \[
% \boldsymbol{x},(\boldsymbol{y}+\varepsilon\boldsymbol{1}) \geq \mathtt{P}(\pi|\MDP^{s},\Psi).
% \]

% \[
% y_{j_0} \geq \mathbb{P}(\Psi_1 \mathbf{W} (\Psi_1 \wedge \Psi_2)|s,\pi)-\varepsilon=\mathbb{P}(\psi|s,\pi)-\varepsilon.
% \]

Since $V_n$ is a non-decreasing set in a compact domain, the value-frontier converges uniformly toward its limit, which concludes \textbf{Induction} $\mathcal{I}_1$. 
%Since we have $y_1 \geq \mathbb{P}(\psi_{k}|\rho,\pi_s)$ for any history $\rho$ and policy $\pi_s$, we can then find $n_0$ such that $n\geq n_0$ implies that the value frontier is at a distance at most $\varepsilon/2$ in infinity norm, which concludes \textbf{Induction 1}.
    
\end{proof}

\subsection{Close non-evaporating policy}

\begin{lemma}\label{lemma-tree-almost-1}\textbf{(Finite tree decomposition)}
    Let $\Phi$ be a \CPCTL\ formula of the form $\Phi=\mathbb{P}_{\geq p}(\phi)$ and let $\MDP$ be a Markov Chain. For any $\eta>0$ and $\varepsilon>0$, there exists a depth $h$ such that with $T$ the tree rooted in $s_0$ of depth $h$, $\mathcal{F}$ the set of leaves, $\mathcal{F}_{\neq 0}$ and $\mathcal{F}_{=0}$ defined as
    \[
    \mathcal{F}_{\neq 0} = \{\, s \in T,\ \exists \xi,\ \xi[0] = s_0,\ \xi[h(\varepsilon)] = s \, ,~\forall j\leq h,~\xi[j] \models \Phi_1\},
    \]
    \[
    \mathcal{F}_{=0} = \{\, s \in T,\ \exists \xi,\ \xi[0] = s_0,\ \xi[h(\varepsilon)] = s \, ,~\exists j\leq h,~\xi[j] \not\models \Phi_1\},
    \]
    we have
    \[
    \mathbb{P}\left( \left\{ s'\in \mathcal{F}_{\neq 0},~\mathbb{P}(\phi|s')\leq 1-\eta \right\} \right) \leq \varepsilon.
    \]
    
\end{lemma}

\begin{proof}
Let $\Phi = \mathbb{P}_{\geq p}(\phi)$ be a \CPCTL\ formula, and let $\MDP$ be a Markov Chain.

\begin{itemize}
    \item \textbf{Case 1:} $\phi = \mathbf{G}\Phi_1$.  
    The Bellman equation is given by
    \[
    \mathbb{P}(\mathbf{G}\Phi_1 \mid s) =
    \begin{cases}
        \sum_{s'} P(s' \mid s)\, \mathbb{P}(\mathbf{G}\Phi_1 \mid s') & \text{if } s \models \Phi_1,\\
        0 & \text{otherwise.}
    \end{cases}
    \]
    Let $p_0(s) = \mathbb{P}(\mathbf{G}\Phi_1 \mid s)$.  
    If $p_0 = 1$, the literal projection lemma applies directly.  
    Otherwise, consider the alternative formula
    \[
    \Psi_\varepsilon = \mathbb{P}_{\geq p_0(s_0) + \varepsilon}(\mathbf{G}\Phi_1),
    \]
    which is a safe formula such that $\mathcal{M} \not\models \Psi_\varepsilon$.  
    Therefore, by~\cite{katoen2014probably}, there exists a finite tree rooted at $s_0$ that witnesses the violation of $\Psi_\varepsilon$.  
    Let $h$ denote the depth of the smallest such tree $T$.

    We define
    \[
    \mathcal{F}_h = \{\, s \mid \exists \xi,\ \xi \models \mathbf{G}\Phi_1,\ \xi[0] = s_0,\ \xi[h] = s \,\},
    \]
    and
    \[
    \mathcal{F}_0 = \{\, s \in T \mid s \not\models \Phi_1,\ \exists \xi,\ \xi[0] = s_0,\ \exists i>0,\ \xi[i] = s \,\}.
    \]
    Construct $\mathcal{T}$ by replacing each element of $\mathcal{F}_0$ with a leaf having identical labels, and let $\mathcal{F}^*$ denote the set of leaves of $\mathcal{T}$.

    For each non-leaf node $s \in \mathcal{T}$, we have $s \models \Phi_1$.  
    By the law of total probability applied to the subtree $\mathcal{T}$ of $\MDP$, we obtain
    \[
    \begin{aligned}
    \mathbb{P}(\mathbf{G}\Phi_1 \mid s)
        &= \sum_{s' \in \mathcal{F}} P(s' \mid s)\, \mathbb{P}(\mathbf{G}\Phi_1 \mid s') \\
        &= \sum_{s' \in \mathcal{F}_h} P(s' \mid s)\, \mathbb{P}(\mathbf{G}\Phi_1 \mid s') + \sum_{s' \in \mathcal{F}_0} P(s' \mid s)\, \mathbb{P}(\mathbf{G}\Phi_1 \mid s') \\
        &= \sum_{s' \in \mathcal{F}_h} P(s' \mid s)\, \mathbb{P}(\mathbf{G}\Phi_1 \mid s').
    \end{aligned}
    \]

    Let $\eta > 0$, and define
    \[
    D_\eta = \{\, s' \in \mathcal{F}_h \mid \mathbb{P}(\mathbf{G}\Phi_1 \mid s') \leq 1 - \eta \,\}.
    \]
    Then
    \[
    \mathbb{P}(\mathbf{G}\Phi_1 \mid s)
        = \sum_{s' \in D_\eta} P(s' \mid s)\, \mathbb{P}(\mathbf{G}\Phi_1 \mid s')
        + \sum_{s' \in \mathcal{F}_h \setminus D_\eta} P(s' \mid s)\, \mathbb{P}(\mathbf{G}\Phi_1 \mid s').
    \]
    We have
    \[
    \sum_{s' \in \mathcal{F}_h \setminus D_\eta} P(s' \mid s)\, \mathbb{P}(\mathbf{G}\Phi_1 \mid s') \leq \mathbb{P}(\mathcal{F}_h \setminus D_\eta),
    \qquad
    \sum_{s' \in D_\eta} P(s' \mid s)\, \mathbb{P}(\mathbf{G}\Phi_1 \mid s') \leq \mathbb{P}(D_\eta)\,(1-\eta),
    \]
    and also
    \[
    \mathbb{P}(D_\eta) + \mathbb{P}(\mathcal{F}_h \setminus D_\eta) \leq p_0 + \varepsilon.
    \]
    Consequently,
    \[
    p_0 \leq \mathbb{P}(\mathbf{G}\Phi_1) \leq \mathbb{P}(D_\eta)\,(1-\eta) + (p_0 + \varepsilon - \mathbb{P}(D_\eta)),
    \]
    which implies
    \[
    \mathbb{P}(D_\eta) \leq \frac{\varepsilon}{\eta}.
    \]
    The result follows by choosing $\eta = \sqrt{\varepsilon}$.

    \item \textbf{Case 2:} $\phi = \Phi_1 \mathbf{W}(\Phi_1 \wedge \Phi_2)$.  
    Let $p_0 = \mathbb{P}(\phi \mid s_0)$, and consider a finite tree $\mathcal{T}$ witnessing $\MDP \not\models \mathbb{P}_{\geq p_0 + \varepsilon}(\Phi_1 \mathbf{W} (\Phi_1 \wedge \Phi_2))$.
    We adapt the previous argument by defining the following sets of leaves:
    \[
    \begin{aligned}
    \mathcal{F}_0 &= \{\, s \in T \mid s \not\models \Phi_1,\ \exists \xi,\ \xi[0] = s_0,\, \exists k,\ \xi[k] = s,\, \forall j < k,\ \xi[j] \models \Phi_1 \,\},\\[2mm]
    \mathcal{F}_1 &= \{\, s \in T \mid s \models \Phi_1 \wedge \Phi_2,\ \exists \xi,\ \xi[0] = s_0,\, \exists k,\ \xi[k] = s,\, \forall j < k,\ \xi[j] \models \Phi_1 \,\},\\[2mm]
    \mathcal{F}_h &= \mathrm{leaves}(\mathcal{T}) \setminus (\mathcal{F}_0 \cup \mathcal{F}_1).
    \end{aligned}
    \]
    By the total probability law,
    \[
    p_0 = \mathbb{P}(\phi \mid s)
        = \sum_{s' \in \mathcal{F}_h} P(s' \mid s)\, \mathbb{P}(\phi \mid s')
        + \sum_{s' \in \mathcal{F}_0} P(s' \mid s)\, \mathbb{P}(\phi \mid s')
        + \sum_{s' \in \mathcal{F}_1} P(s' \mid s)\, \mathbb{P}(\phi \mid s'),
    \]
    hence
    \[
    p_0 = \sum_{s' \in \mathcal{F}_h} P(s' \mid s)\, \mathbb{P}(\phi \mid s') + \sum_{s' \in \mathcal{F}_1} P(s' \mid s).
    \]

    Since $\mathcal{T}^*$ is a witness, we have $\mathbb{P}(\mathcal{T}_0) \geq 1 - p_0 - \varepsilon$.  
    For $\eta > 0$, define
    \[
    \alpha(\eta) = \mathbb{P}(\{\, s' \in \mathcal{F}_h \mid \mathbb{P}(\phi \mid s') \leq 1 - \eta \,\}), \quad
    \beta(\eta) = \mathbb{P}(\{\, s' \in \mathcal{F}_h \mid \mathbb{P}(\phi \mid s') \geq 1 - \eta \,\}), \quad
    \gamma(\eta) = \mathbb{P}(\mathcal{T}_1).
    \]
    We have $\alpha(\eta) + \beta(\eta) + \gamma(\eta) \leq p_0 + \varepsilon$, thus
    \[
    p_0 \leq \alpha(\eta)(1-\eta) + (\beta(\eta) + \gamma(\eta)) \leq \alpha(\eta)(1-\eta) + (p_0 + \varepsilon - \alpha(\eta)).
    \]
    Choosing $\eta = \sqrt{\varepsilon}$ yields the desired conclusion.
\end{itemize}
\end{proof}

\begin{definition}\textbf{(Total Variational distance between policies)}
For any two memoryful policies $\pi$, $\pi'$ on an MDP $\MDP$, we define the total variational distance $d_{TV}(\pi,\pi')$ between $\pi$ and $\pi'$ as the Total Variational distance between the path measures $\mathbb P_{\pi}$ and $\mathbb P_{\pi'}$ induced by $\pi$ and $\pi'$ respectively, \emph{i.e.} as $$\sup_{A\in \mathcal F} | \mathbb P_{\pi}(A) - \mathbb P_{\pi'}(A) |,$$ where $\mathcal F$ is the $\sigma$-algebra of the path measure space, generated by the cylinder sets.
\end{definition}

\begin{lemma}\label{leamma-weak-continuity}\textbf{(Weak Upward Continuity of safe \PCTL\ for the Total Variational distance)}
    Let $\MDP$ be an MDP, and $n\in \mathbb N$. For any $\kappa>0$, there exists $\epsilon$ such that for any $\Phi,\Psi$ safe formulas of size smaller than $n$ with probability thresholds $\boldsymbol p$ and $\boldsymbol q$ respectively such that $\Psi = \mathtt{T}(\Phi, \boldsymbol{q})$, $\boldsymbol p< \boldsymbol q$, and $\min{(\boldsymbol q-\boldsymbol p)}\geq\kappa$, for any two policies $\pi,\pi'$ with $d_{TV}(\pi,\pi')\leq\epsilon$, we have \[
    \MDP_{\pi} \models \Psi \Rightarrow \MDP_{\pi'} \models \Phi. 
    \]
\end{lemma}
\begin{proof}
    Let $\kappa>0$. We do the proof by induction in $n$.

    \begin{itemize}
        \item \textbf{(Base case:)} Let $\Phi,\Psi$ be conjunctions or literals and their negations. In that case, $\Phi=\Psi$, and the satisfaction of the formula only depends on the labels of the initial state.
        \item \textbf{(Induction:)} %Let $\Phi = \Phi_1 \mathbf{W} (\Phi_1 \wedge \Phi_2)$ and $\Psi=\Psi_1 \mathbf{W} (\Psi_1\wedge \Psi_2)$. We let  
        %We consider a history $\rho$ such that $\MDP_\pi^\rho \models \Psi$. 
        %[TODO]
        %There exists two sets of paths $A,B$ such that 
        %\[
        %\mathbb{P}(A)\leq \varepsilon,~\forall \rho<\xi,\xi\in B,|\pi_n(\rho)-\pi(\rho)| \leq \varepsilon.
        %\]
        Suppose by induction that for any $\Phi,\Psi$ safe formulas of size smaller than $n$ with probability thresholds $\boldsymbol p$ and $\boldsymbol q$ respectively such that $\Psi = \mathtt{T}(\Phi, \boldsymbol{q})$, $\boldsymbol p< \boldsymbol q$, and $\min{(\boldsymbol q-\boldsymbol p)}\geq\kappa$, for any two policies $\pi,\pi'$ with $d_{TV}(\pi,\pi')\leq\eta$, we have \[
        \MDP_{\pi} \models \Psi \Rightarrow \MDP_{\pi'} \models \Phi. 
        \] Let $\Phi,\Psi$ of size $n+1$ with probability thresholds $\boldsymbol p$ and $\boldsymbol q$ respectively such that $\Psi = \mathtt{T}(\Phi, \boldsymbol{q})$, $\boldsymbol p< \boldsymbol q$, and $\min{(\boldsymbol q-\boldsymbol p)}\geq\kappa$, and let $\epsilon=\frac{\eta\kappa}{4}$, and suppose that $d_{TV}(\pi,\pi')\leq \epsilon$. Let $R=\{\rho \text{ minimal} \mid d_{TV}(\pi_\rho,\pi'_\rho)\geq\eta\}$, let $R_1=\{\rho \text{ minimal}\mid \sup_{A}\mathbb P_{\pi_\rho}(A)-\mathbb P_{\pi'_\rho}(A)\geq\eta\}$, and let $R_2=\{\rho \text{ minimal}\mid \sup_{A}\mathbb P_{\pi'_\rho}(A)-\mathbb P_{\pi_\rho}(A)\leq\eta\}$. There exists $j\in\{1,2\}$ such that $\mu_{\pi'}(R_j)\geq \frac{\mu_{\pi'}(R)}{2}$.
        \begin{itemize}
            \item Suppose that $j=1$. Then, for any $\rho\in R$, there exists a set of measurable paths $A_\rho\in S^\star$ such that $\mathbb P_{\pi_\rho}(A_\rho)-\mathbb P_{\pi'_\rho}(A_\rho)\geq \eta$. However, since $d_{TV}(\pi,\pi')\leq \epsilon$, we have $$\left|\sum_{\rho\in R_1} \pi(\rho)\mathbb P_{\pi_\rho}(A_\rho)-\pi'(\rho)\mathbb P_{\pi'_\rho}(A_\rho)\right|\leq \epsilon.$$
            Furthermore \begin{gather*}
                \sum_{\rho\in R_1} \pi(\rho)\mathbb P_{\pi_\rho}(A_\rho)-\pi'(\rho)\mathbb P_{\pi'_\rho}(A_\rho)=
                \sum_{\rho\in R_1} \pi(\rho)\mathbb P_{\pi_\rho}(A_\rho)-\pi'(\rho)\mathbb P_{\pi_\rho}(A_\rho)+ \pi'(\rho)\mathbb P_{\pi_\rho}(A_\rho) -\pi'(\rho)\mathbb P_{\pi'_\rho}(A_\rho)=\\
                \sum_{\rho\in R_1} (\pi(\rho)-\pi'(\rho))\mathbb P_{\pi_\rho}(A_\rho)+ \pi'(\rho)(\mathbb P_{\pi_\rho}(A_\rho) -\mathbb P_{\pi'_\rho}(A_\rho)).
            \end{gather*}
            But $$\sum_{\rho\in R_1} (\pi(\rho)-\pi'(\rho))\mathbb P_{\pi_\rho}(A_\rho)=\int_{t=0}^1 \sum_{\{\rho\in R_1\}} (\pi(\rho)-\pi'(\rho))\mathds  {1}_{\{\mathbb P_{\pi_\rho}(A_\rho)>t\}} dt,$$ and we thus have, since $\sum_{\rho\in R_1} |(\pi(\rho)-\pi'(\rho))|\leq 2$, $$\left|\sum_{\rho\in R_1} (\pi(\rho)-\pi'(\rho))\mathbb P_{\pi_\rho}(A_\rho)\right|\leq \int_{t=0}^1 \left|\sum_{\{\rho\in R_1\}} (\pi(\rho)-\pi'(\rho))\mathds  {1}_{\{\mathbb P_{\pi_\rho}(A_\rho)>t\}}\right| dt \leq \epsilon,$$ since for any countable set $E$, $|\sum_{\{\rho\in E\}} (\pi(\rho)-\pi'(\rho))|\leq d_{TV}(\pi,\pi')$. 
            
            In addition, $\sum_{\rho\in R_1}\pi'(\rho)(\mathbb P_{\pi_\rho}(A_\rho) -\mathbb P_{\pi'_\rho}(A_\rho))\geq \mathbb P_{\pi'}(R_1)\alpha$ by definition of $R_1$.
            As a consequence, we have $|\epsilon-P_{\pi'}(R_1)\eta|\leq \epsilon$, and $P_{\pi'}(R_1)\leq \frac{\kappa}{2}$. As a consequence, $P_{\pi'}(R)\leq \kappa$.
            \item The case $j=2$ is similar.
        \end{itemize}
         Thus, for any history that is not in $R \cdot S^\star$, by induction, for any subformula $\Psi'$ of $\Psi$ with probability thresholds $\boldsymbol q'$, $\mathcal M_{\pi}\models \Psi'\implies \mathcal M_{\pi'}\models \Phi'$, where $\Phi'=\mathtt{T}(\Phi, \boldsymbol{p'})$. As a consequence, by definition of PCTL semantics, since $\mathbb P_{\pi'}(R\cdot S^\star)\leq \kappa$, we have $\MDP_{\pi} \models \Psi \Rightarrow \MDP_{\pi'} \models \Phi$.
    \end{itemize}
\end{proof}

\begin{definition}\textbf{(Induced Safe Policy)}
    Let $\pi$ be a policy of an MDP $\MDP$, let $\Phi$ be a propositional formula, let $V$ be the set of states $s$ of $\MDP$ such that there exists $\sigma$ with $\mathbb P(\mathbf G\Phi\mid \sigma,s)=1$, and for any $s\in V$, let $\alpha(s)$ be an action such that the support of $\alpha(s)$ is included in $V$. Furthermore, suppose that the initial state of $\MDP$ is in $V$. Then we define $\text{Safe}(\pi,\Phi)$ as a policy $\pi'$ such that for any history $\rho\in V^\star$, $$\pi'(\rho)(a)=\begin{cases}
        0 & \text{ if } \text{Supp}(a) \not\subseteq V\\
        \pi(\rho)(a) + 1-\sum_{\{b|\text{Supp}(a) \not\subseteq V\}} \pi(\rho)(b) & \text{ if } a=\alpha(s) \\
        \pi(\rho)(a) & \text{otherwise.}
    \end{cases}$$
\end{definition}

\begin{lemma}\label{lemma-safe-TV}\textbf{(Close Induced Safe Policy)}
    For any MDP $\mathcal M$, for any policy $\pi$ of $\mathcal M$, for any $\epsilon>0$, there exists $\eta$ such that, for any propositional formula $\Phi$, if $\mathbb P(\mathbf G\Phi \mid \pi)>1-\eta$, then $\pi$ induces a safe policy $\text{Safe}(\pi,\Phi)$ and we have $d_{TV}(\pi,\text{Safe}(\pi,\Phi))<\epsilon$. Furthermore, if $\pi$ is finite-memory, then $\text{Safe}(\pi,\Phi)$ is finite-memory as well.
\end{lemma}
\begin{proof}
     We let $V$ be the set of the states $s$ of $\MDP$ such that there exists a policy $\pi$ with $\mathbb P(\Phi|\pi, s)\neq 1$. For any state $s$ that is not in $V$, for any policy $\pi$, there exists a history $\rho$ of length less than the number of states $n$ of $\MDP$ that $\mathbb P^s_{\pi}(\rho)>0$ and $\last{\rho}\not\models\Phi$. As a consequence, if we let $d$ denote the smallest non-zero transition probability of $\MDP$, for any $s$ not in $V$, we have $\sup_{\pi}\mathbb P(\Phi|\pi, s)\leq 1-d^n$. 
     
     We now suppose w.l.o.g. that $\epsilon<1$, we let $\eta=\epsilon d^n$, we let $\pi$ be a policy such that $\mathbb P(\Phi \mid \pi)>1-\eta$ and we let $R$ be the set of all histories $\rho$ such that $\last{\rho}\notin V$ and that are minimal among such histories. We have $\mathbb P_{\pi}(R)d^n\leq \eta$, which implies that $\mathbb P_{\pi}(R)\leq \epsilon$. For any measurable set of paths $A$, by definition of $\text{Safe}(\pi,\Phi)$, we have $$|\mathbb P_{\text{Safe}(\pi,\Phi)}(A)-\mathbb P_{\pi}(A)|\leq \mathbb P_{\pi}(R)\leq \epsilon,$$
     which concludes the proof.
\end{proof}

\begin{theorem}\textbf{(Close finite-memory policy).} Let $\MDP$ be a MDP, $\Phi$ %=\mathbb{P}_{\geq p} (\Phi_1 \mathbf{W} \Phi_1 \wedge \Phi_2)$
be a \CPCTL\ formula, $\boldsymbol p$ be its probability threshold, and $\pi$ be a policy such that $\MDP_\pi \models_S \Phi$. For any $\eta\geq 0$, there exists a finite-memory $\pi'$ such that
\[
\mathbb{P}\left( \Phi | \mathcal M_{\pi'} \right)\geq \mathbb{P}\left( \mathtt{T}(\Phi,\boldsymbol p - \eta)| \mathcal M_{\pi} \right).
\]
\end{theorem}

\begin{proof}
    Let $F$ be a set of formulas close by the subformula partial order, $n_1$ and $n_2$ be respectively the maximal total depth of formulas in $F$ and the number of maximal formulas of this depth (for the subformula partial order). We prove the result by induction on $F$ equipped with the lexicographic order $(n_1,n_2)> (m_1,m_2)$ if and only if $n_1>m_1$ or $n_1=m_1$ and $n_2>m_2$.

    \begin{itemize}
        \item \textbf{Induction assumption:} Let $F$ be a set of formulas, all subformula of some $\Phi$ with threshold $\boldsymbol p$, and close by the subformula partial order. For any MDP $\MDP$ and policy $\pi$ of $\MDP$, for all $\eta>0$, for all $F'
        \subsetneq F$, there exists $\pi'$ finite-memory such that for every $\Psi\in F'$, $\mathcal M_{\pi}\models \mathtt{T}(\Psi,\boldsymbol p+ \eta)\Rightarrow \mathcal M_{\pi'} \models \Psi$.

        \item \textbf{Induction step, $\bigwedge$:} Let $F=\{\Psi^1,\dots,\Psi^m\}$. Without loss of generality, we suppose that $\Psi^1$ is a maximal formula in $F$ of maximal depth, and first assume that it has the form
        \[
        \Psi^1 = \bigwedge_{j=1}^l \Psi_j^1 = \bigwedge_{j=1}^l \left[ \left( \bigwedge_{k=1}^{l_1-1} b_k^j \right) \wedge \left( \bigwedge_{k=l_1}^{l_2} \neg b_k^j \right)\right].
        \]
        We define $\tilde F= (F\setminus \{ \Psi^1 \}) \cup \{ \Psi_1^1,\dots,\Psi_l^1 \}$. Either $\Psi_1$ was the only formula of its depth, and the maximal depth has been decreased by one. Or, the maximal depth in $\tilde F$ is the same, but the amount of formulas of this depth has been decreased by one. In both cases, $\tilde F < F$ and the induction hypothesis applies. The result automatically follows by induction.
        
        \item \textbf{Induction step, $\mathbf W$ :} We now consider $F$ of the form $F= \{ \Psi,\Phi_1,\dots,\Phi_m\}$, where $\Psi$ is a maximal formula of maximal total depth, and has the form

        \[
        \Psi = \Psi_1 \mathbf{W} (\Psi_1 \wedge \Psi_2),\quad \Psi_1 = \bigwedge_{j=1}^{n_1} \Psi_1^j,~\Psi_2=\bigwedge_{j=1}^{n_2} \Psi_2^j.
        \]
        By lemma \ref{lemma-tree-almost-1}, for any $\eta_1$,$\eta_2$, there exists a finite tree $T$ rooted at the start $s_0$ of depth $h$ such that for all maximal (finite) history $\xi$ in $T$ except on a set of measure smaller or equal to $\eta_1$, we have 
        \[
        \left\{
        \begin{aligned}
            &(i)\quad \text{either }\xi,\pi\not\models \Psi_1 \\
            &(ii)\quad \text{or either }\xi, \pi \models \Psi_1 \wedge \Psi_2 \text{ and }\forall k,~\xi[k],\tilde \pi \models \Psi_1.\\
            &(iii)\quad \text{or }\mathbb{P}(\mathbf G\Psi_1|\xi)\geq 1-\eta_2.
        \end{aligned}
        \right.
        \]
        We denote $\Xi_1$ the set of maximal histories of $T$ satisfying $(i)$, $\Xi_2$ the set of maximal histories satisfying $(ii)$ that are not in $\Xi_1$, $\Xi_3$ the set of maximal histories satisfying $(iii)$ that are not in $\Xi_1\cup \Xi_2$, and $\Xi_4$ the set of all the other maximal histories.

        We take $\eta_1$ given by Lemma \ref{leamma-weak-continuity} for $\kappa=\eta/2$, $\epsilon$ given by Lemma \ref{leamma-weak-continuity} for $\kappa=\eta/4$, take $\eta_2$ given by Lemma \ref{lemma-safe-TV}, and let $\mathcal M_1$ be the MDP obtained from $\mathcal M$ by removing all states that do not satisfy $\alit{\Psi_1}$. Applying the induction hypothesis to all histories $\xi$ in $\Xi_1\cup \Xi_2$, MDP $\mathcal M^\xi$, $F\setminus \Psi$ and $\eta/2$ yields finite-memory policies $\pi_\xi$. Applying the induction hypothesis to all histories $\xi_3$, MDP $\mathcal M_1^{\text{Safe}(\pi,\Psi_1)}$, 
        $\mathtt{T}(\Psi,\boldsymbol p-\eta/4)$, and $\eta/4$, yields policies finite-memory policies $\pi_\xi$. Then, if we let $\pi_1$ be the policy obtained from $\pi$ be replacing policy $\pi$ by policies $\pi_\xi$ for every $\xi$ in $\Xi_1\cup\Xi_2\cup\Xi_3$, we have \[
        \mathbb{P}\left( \Phi | \mathcal M_{\pi_1} \right)\geq \mathbb{P}\left( \mathtt{T}(\Phi,\boldsymbol p - \eta/2)| \mathcal M_{\pi} \right).
        \]
        for any $\Phi\subseteq F$.
        We now let $\pi'$ be any finite-memory policy equal to $\pi_1$ on all paths that are prefixes or suffixes of histories in $\Xi_1\cup\Xi_2\cup\Xi_3$. Lemma \ref{leamma-weak-continuity} concludes the proof. 
        \end{itemize}
\end{proof}

\begin{lemma}\label{lem:finitemem}
    For any finite Markov Chain $\MDP$, for all $s$,
    \[
    \left\{
    \begin{aligned}
        &\forall \Phi \text{ state formula},~s \models \Phi \Leftrightarrow s \models \mathtt{FD}(\Phi),\\
        &\forall \phi \text{ path formula},~\mathbb{P}(\phi|s) = \mathbb{P}(\mathtt{FD}(\phi)|s).
    \end{aligned}
    \right.
    \]

    As a consequence, the lemma holds for Markov Decision Processes equipped with finite-memory policies.

\end{lemma}
\begin{proof}
We recall the definition of $\mathtt{FD}:$
\[
    \mathtt{FD}:
    \left\{
    \begin{aligned}
        &b\in AP\mapsto b,\\
        &\neg b\in AP\mapsto \neg b,\\
        &\Phi_1 \wedge \Phi_2 \mapsto \mathtt{FD}(\Phi_1)\wedge \mathtt{FD}(\Phi_2),\\
        &\Phi_1 \mathbf W (\Phi_1 \wedge \Phi_2 ) \mapsto \mathtt{FD}(\Phi_1) \mathbf{U} \big[ \mathbb{P}_{=1}(\mathbf G \Phi_1) \vee (\mathtt{FD}(\Phi_1)  \wedge \mathtt{FD} (\Phi_2)) \big],\\
        &\mathbb{P}_{\geq p} (\phi) \mapsto \mathbb{P}_{\geq p} (\mathtt{FD}(\phi)),
    \end{aligned}
    \right.
    \]

and show by induction over the total depth of $\Phi\in$ \CPCTL\ (resp. $\phi\in$ \CPCTL) a state (resp. path) formula the following:  For any finite Markov Chain $\MDP$, for all $s$,
\[
\left\{
\begin{aligned}
    &s \models \Phi \Leftrightarrow s \models \mathtt{FD}(\Phi),\\
    &\mathbb{P}(\phi|s) = \mathbb{P}(\mathtt{FD}(\phi)|s).
\end{aligned}
\right.
\]

\begin{itemize}
    \item If the total depth of $\Phi$ is $0$, then $\Phi=\texttt{FD}(\Phi)$.
    \item If $\Phi = \Phi_1 \wedge \Phi_2$, for any Markov Chain $\MDP$, for all $s$, $s \models \Phi_1 \Leftrightarrow s \models \mathtt{FD}(\Phi_1)$ and $s \models \Phi_2 \Leftrightarrow s \models \mathtt{FD}(\Phi_2)$. Hence, for all $s$, $s \models \Phi_1 \wedge \Phi_2 \Leftrightarrow s \models \mathtt{FD}(\Phi_1) \wedge \texttt{FD}(\Phi_2)$.
    \item If $\phi = \Phi_1 \mathbf{W} (\Phi_1 \wedge \Phi_2)$. We consider a finite Markov Chain $\MDP$ and a state $s$. Let $\xi$ be a path emerging from $s$. By induction, we have for any $j\in \mathbb{N}$,
    \[
    \xi[j]\models \Phi_k \Leftrightarrow \xi[j] \models \texttt{FD}(\Phi_k),~k=1,2.
    \]
    Hence, for any such $\xi$, 
    \[
    \xi \models \mathtt{FD}(\Phi_1) \mathbf{U} \big[ \mathbb{P}_{=1}(\mathbf G \Phi_1) \vee (\mathtt{FD}(\Phi_1)  \wedge \mathtt{FD} (\Phi_2)) \big] \Leftrightarrow \xi \models \Phi_1 \mathbf{U} \big[ \mathbb{P}_{=1}(\mathbf G \Phi_1) \vee (\Phi_1  \wedge \Phi_2) \big]
    \]
    Now, it is clear that $\xi \models \Phi_1 \mathbf{U} (\mathbb{P}_{=1}(\mathbf{G}\Phi_1))$ implies $\xi \models \Phi_1 \mathbf{W} (\Phi_1 \wedge \Phi_2) $ since it implies $\mathbf{G} \Phi_1$. Similarly, $\Phi_1 \mathbf{U}(\Phi_1 \wedge \Phi_2)$ implies $\Phi_1 \mathbf{W}(\Phi_1 \wedge \Phi_2)$. Conversely, let $\xi\models \Phi_1 \mathbf{W} (\Phi_1 \wedge \Phi_2)$. Either $\xi \models \Phi_1 \mathbf{U}(\Phi_1 \wedge \Phi_2)$, in which case $\xi \models \Phi_1 \mathbf{U} \big[ \mathbb{P}_{=1}(\mathbf G \Phi_1) \vee (\Phi_1  \wedge \Phi_2) \big]$. Or, $\xi \models \mathbf{G} \Phi_1$. With 
    \[
    u_s = \mathbb{P}(\mathbf{G}\Phi_1|s),
    \]
    we define 
    \[
    U_\MDP = \sup \mathcal{S}_{\neq 1},\quad  \mathcal{S}_{\neq 1}=\{ u_s \neq 1,~s\in \mathcal{S} \}.
    \]
    Since the MDP is finite, either $ \mathcal{S}_{\neq 1}= \emptyset$, in which case every path satisfies both formulas, or $U_{\MDP} = \alpha <1$.

    Hence, either $\xi$ reaches $s_*\in  \mathcal{S}_{\neq 1}$, which means that $\xi \models \Phi_1 \mathbf{U} (\mathbb{P}_{=1}(\mathbf{G}\Phi_1))$, or for each $n\in \mathbb{N}$, $\mathbb{P}(\mathbf{G}\Phi|\xi[n])\leq \alpha$. Assume that $\xi \models \mathbf{G}\Phi_1$. Consider $\varphi(n)$ a set of indices such that $\xi[\varphi(n)]=\xi[\varphi(n+1)]=s_*$. Either on all paths looping at $s_*$, $\Phi_1$ holds at every intermediate state, in which case $s_*\in \mathcal{S}\setminus S_{\neq 1}$, or there exists a path with non-zero probability such that $\neg \Phi_1$ holds, in which case $\xi$ belongs to a set of zero-measure, the set of paths looping indefinitely from $s_*$ to $s_*$ and never choosing this option. The countable union, indexed by $\mathcal{S}$, of such zero-measure sets is of measure zero.

\end{itemize}
\end{proof}

\end{document}